\documentclass[11pt,letterpaper]{article}
\usepackage[margin=1in]{geometry}
\usepackage[T1]{fontenc}
\usepackage[utf8]{inputenc}
\usepackage{lmodern}
\usepackage{amsmath,amsthm}
\usepackage{algorithm}
\usepackage{caption}
\usepackage{cite}
\usepackage{hyperref}
\newtheorem{theorem}{Theorem}[section]
\newtheorem{proposition}[theorem]{Proposition}
\newtheorem{lemma}[theorem]{Lemma}
\newtheorem{corollary}[theorem]{Corollary}
\theoremstyle{definition}
\newtheorem{definition}[theorem]{Definition}
\newtheorem{remark}[theorem]{Remark}
\numberwithin{equation}{section}
\numberwithin{algorithm}{section}
\newcommand{\sectionpunct}{}

\usepackage{amssymb,mathtools}
\usepackage{microtype}
\usepackage{algpseudocode}
\usepackage{enumitem}
\usepackage{booktabs,tabularx,array}
\usepackage{xcolor}
\usepackage{tikz}
\usetikzlibrary{arrows.meta,backgrounds,fit,positioning}
\algrenewcommand\algorithmicrequire{\textbf{Input:}}
\algrenewcommand\algorithmicensure{\textbf{Output:}}
\newcommand{\R}{\mathbb{R}}
\newcommand{\Z}{\mathbb{Z}}
\newcommand{\1}{\mathbf 1}
\newcommand{\ee}{\mathbf e}
\newcommand{\ip}[2]{\left\langle #1,#2\right\rangle}
\newcommand{\norm}[1]{\left\lVert #1\right\rVert}
\newcommand{\normD}[1]{\left\lVert #1\right\rVert_D}
\newcommand{\normDinv}[1]{\left\lVert #1\right\rVert_{D^{-1}}}
\newcommand{\Xzero}{\mathcal X_0}
\DeclareMathOperator{\supp}{supp}
\DeclareMathOperator{\argmin}{argmin}

\newcommand{\papertitle}{Solving Hypergraph Laplacian Systems in Almost-Linear Time}
\newcommand{\paperauthor}{Yuichi Yoshida}
\newcommand{\paperaffiliation}{National Institute of Informatics}
\newcommand{\paperemail}{yyoshida@nii.ac.jp}
\newcommand{\paperfunding}{Supported by JSPS KAKENHI Grant Numbers 22H05001 and 25K24465.}

\hypersetup{pdftitle={Solving Hypergraph Laplacian Systems in Almost-Linear Time},pdfauthor={Yuichi Yoshida}}

\hypersetup{colorlinks=true,linkcolor=blue!60!black,citecolor=blue!60!black,urlcolor=blue!60!black}
\title{\papertitle}
\author{\paperauthor\thanks{\paperfunding}\\\paperaffiliation\\\texttt{\paperemail}}
\date{}

\begin{document}
\hypersetup{pageanchor=false}
\pagenumbering{gobble}
\maketitle

\begin{abstract}
For a connected weighted hypergraph, we solve the Poisson problem for the cut-based hypergraph Laplacian in randomized almost-linear time, measured in the natural input size \(P=\sum_{e\in E}|e|\), the sum of hyperedge sizes.
For every fixed constant \(C>0\), our randomized algorithm runs in \(P^{1+o(1)}\) time and, with high probability over its internal randomness, returns a primal point and a dual certificate, with additive optimality gap at most \(\exp(-\log^C P)\).
A key step is to rewrite the Fenchel dual as a convex-flow problem on a lifted \(O(P)\)-arc graph, yielding a near-optimal lifted dual flow.
The main difficulty is primal recovery, because this lifted flow does not by itself determine a primal potential.
Our main new ingredient is a recovery theorem showing that, for primal recovery, the detailed routing of the lifted dual flow inside each hyperedge gadget can be discarded: one nonnegative scalar per hyperedge is enough.
After the necessary finite-precision rounding, these hyperedge scalars define a min-cost-flow instance with linear costs on the same lifted graph, and solving it exactly recovers a primal potential.
Finally, by adding a ground vertex, we reduce regularized objectives to the Poisson solver. This gives randomized almost-linear-time algorithms for evaluating resolvents and proximal maps for the same cut-based hypergraph Laplacian.
\end{abstract}
\pagestyle{empty}
\newpage
\tableofcontents
\newpage
\hypersetup{pageanchor=true}
\pagenumbering{arabic}
\pagestyle{plain}

\section{Introduction\sectionpunct}

Solving graph Laplacian systems is one of the basic algorithmic primitives of modern theoretical computer science.
It underlies flow, diffusion, random-walk, learning, and network-analysis routines, and in graphs it admits nearly-linear-time algorithms \cite{Vishnoi2013Lxb,KyngSachdeva2016}.
A natural question is whether one can obtain an analogous primitive for higher-order network models themselves.
Hypergraphs capture multiway interactions directly, but the passage from graphs to hypergraphs quickly destroys the linear structure that makes graph Laplacian systems so tractable.

There are several inequivalent notions of hypergraph Laplacian. We study the cut-based / Lov\'asz-extension model of Hein, Setzer, Jost, and Rangapuram \cite{Hein2013TVH} and Chan, Louis, Tang, and Zhang \cite{ChanLouisTangZhang2018HypergraphLaplacian}, a standard nonlinear model that is also a special case of Yoshida's submodular-Laplacian framework \cite{Yoshida2019CheegerSubmodular}.
For this framework, Fujii, Soma, and Yoshida gave polynomial-time algorithms for submodular Laplacian systems, with more efficient formulations in the hypergraph case \cite{FujiiSomaYoshida2018}.
Other work studies tasks built from the same cut-based operator or its regularized variants, including pairwise potential-difference queries, diffusion, PageRank, localized ratio cut, resolvents, and regularized quadratic decomposable submodular function minimization (QDSFM) \cite{VeldtBK20,TakaiMiyauchiIkedaYoshida2020,liu2021strongly,AmeranisEtAl2023Resolvents,LiHeMilenkovic2020}.
These results motivate studying Poisson and regularized solvers for this operator, but they do not provide an almost-linear-time solver in the incidence size $P$ for the Poisson problem itself.
That is the gap we close here.

\subsection{Problem setting\sectionpunct}

For a weighted hypergraph $H=(V,E,w)$, define for $x\in\R^V$ and $e\in E$
\[
R_e(x)=\max_{u\in e} x_u-\min_{v\in e} x_v,
\]
and let
\[
\mathcal E_H(x):=\frac12\sum_{e\in E} w_e R_e(x)^2,
\qquad
L_H:=\partial \mathcal E_H.
\]
Here \(\partial \mathcal E_H\) denotes the subdifferential of \(\mathcal E_H\). Thus $L_H$ is a set-valued nonlinear Laplacian, and the associated Poisson problem asks for a potential $x$ satisfying the inclusion
\[
L_H(x)\ni s
\]
for a prescribed demand vector $s\in\R^V$. Throughout, this is the only notion of hypergraph Laplacian considered here.

We measure the input size by
\[
P:=\sum_{e\in E} |e|,
\]
the incidence size of the hypergraph.

Let
\[
D=\operatorname{diag}(d_v),
\qquad
 d_v=\sum_{e\ni v} w_e,
\qquad
\Xzero:=\{x\in\R^V:\ip{Dx}{\1}=0\}.
\]

Except where explicitly stated otherwise, assume henceforth that $H$ is connected. Adding the same constant to every coordinate leaves both $\mathcal E_H$ and $L_H$ unchanged, so solutions to the Poisson problem are defined only up to an additive constant. We remove this ambiguity by restricting $x$ to $\Xzero$.
The corresponding variational formulation is the convex optimization problem
\begin{equation}
\label{eq:poisson-primal}
\mathcal P(x):=\mathcal E_H(x)-\ip{s}{x},
\qquad
\mathrm{OPT}:=\min_{x\in\Xzero} \mathcal P(x).
\end{equation}
This normalization removes only the additive ambiguity: even on \(\Xzero\), \(\mathcal E_H\) need not be strictly convex, so the minimizer need not be unique.
A point \(x^\star\in\Xzero\) is optimal for \eqref{eq:poisson-primal} if and only if
\[
s\in L_H(x^\star)+\operatorname{span}(D\1).
\]
The extra term \(\operatorname{span}(D\1)\) is the normal space of the constraint \(x\in\Xzero\).
Moreover, since \(\ip{D\1}{x}=0\) for every \(x\in\Xzero\), replacing \(s\) by \(s+cD\1\) does not change \eqref{eq:poisson-primal} on \(\Xzero\). Thus, at the conceptual level, one may choose the representative in this affine class whose coordinate sum is zero. In the algorithmic statements below, however, we do not use this as a finite-precision preprocessing step: the supplied demand vector is assumed to satisfy \(\1^\top s=0\) as part of the input.

When every hyperedge has size two, $R_{\{u,v\}}(x)=|x_u-x_v|$ and the problem reduces to the usual graph Poisson equation. The hypergraph case is substantially harder algorithmically: the operator is nonlinear, the equation is a subdifferential inclusion rather than a linear system, and standard graph-based escape routes such as clique expansion or star expansion change the model.
Our goal is to solve this hypergraph Poisson problem itself in time almost linear in $P$.

All algorithmic statements below assume that the input lists, for each hyperedge \(e\), the vertices contained in \(e\); thus reading the hypergraph takes \(\Theta(P)\) time. We state all bit-complexity bounds under the following dyadic input model.\footnote{The dyadic assumption is used only to keep the finite-precision analysis explicit: sums of \(O(P)\) input numbers and the integer scalings used in exact min-cost-flow then keep a common denominator with exponent \(\log^{O(1)}P\). The same arguments would extend to rational inputs with a common denominator \(Q\) satisfying \(\log Q=\log^{O(1)}P\), whereas arbitrary polylog-bit denominators can have a least common multiple with \(P\log^{O(1)}P\) bits.}

\begin{definition}[Dyadic input model]
\label{def:dyadic-input}
A scalar $z$ is \emph{$L$-dyadic} if
\[
z=a2^{-q}
\]
for integers $a$ and $q\ge 0$ with
\[
q\le L,
\qquad
\log(|a|+1)\le L.
\]
Throughout the bit-complexity statements, for some $L=\log^{O(1)}P$, all input weights $w_e$ and demands $s_v$ are $L$-dyadic. Equivalently, every such input scalar has a binary representation of length $\log^{O(1)}P$, and the denominator exponent is explicitly bounded by $\log^{O(1)}P$.
\end{definition}

In addition, the phrase ``polynomially bounded'' means that for some constant $K_0\ge 1$,
\[
\norm{s}_\infty \le P^{K_0},
\qquad
P^{-K_0}\le w_e\le P^{K_0}\quad(\forall e\in E).
\]
\subsection{Our results\sectionpunct}

Our main result is an almost-linear-time randomized solver for the hypergraph Poisson problem in the natural input size $P$. Informally, under the dyadic input model and polynomial-boundedness assumptions, the algorithm returns a rational primal point together with an explicit dyadic dual certificate, and both outputs have $\log^{O(1)}P$-bit coordinates.

Writing $\mathcal D$ for the dual objective and $\mathcal D^\star$ for its optimum, we use the convention in which the dual is also a minimization problem, so strong duality reads $\mathrm{OPT}+\mathcal D^\star=0$. The first main guarantee is then the following.

\begin{theorem}[Informal main theorem]
\label{thm:intro-main}
Assume that the hypergraph is connected, that $\1^\top s=0$, and that the input data satisfy the dyadic input model and are polynomially bounded in $P$. Then for every fixed constant $C>0$, there is a randomized algorithm that runs in $P^{1+o(1)}$ time and, with high probability over its internal randomness, returns a rational primal point $x\in\Xzero$ and an explicit dyadic feasible dual certificate $\widehat\eta$, with every coordinate of $x$ and $\widehat\eta$ having bit length $\log^{O(1)}P$, such that
\[
\mathcal D(\widehat\eta)\le \mathcal D^\star + \exp(-\log^C P),
\qquad
\mathcal P(x)\le \mathrm{OPT} + \exp(-\log^C P).
\]
Using strong duality in the form $\mathrm{OPT}+\mathcal D^\star=0$, we therefore obtain
\[
0\le \mathcal P(x)+\mathcal D(\widehat\eta)\le 2\exp(-\log^C P).
\]
This statement is proved formally as Theorem~\ref{thm:primal-recovery} in Section~\ref{sec:primal-recovery}.
\end{theorem}

Theorem~\ref{thm:intro-main} should be compared with the polynomial-time algorithms of Fujii, Soma, and Yoshida for submodular Laplacian systems \cite{FujiiSomaYoshida2018}.
Their framework already establishes broad tractability, and for hypergraphs it gives a flow-like formulation with variables for ordered vertex pairs inside each hyperedge.
Thus, if \(M=\sum_{e\in E}|e|(|e|-1)\), their hypergraph specialization can be solved by a polynomial-time quadratic-cost flow algorithm, for example in \(O(M^4\log M)\) time by the strongly polynomial algorithm of V\'egh \cite{Vegh2016QuadraticFlow}.
Our result is more specialized but substantially faster in the cut-based Poisson setting studied here: it gives an almost-linear-time randomized solver in the incidence size $P$ together with explicit primal and dual certificates.

A natural extension adds the diagonal quadratic term $\frac{\lambda}{2}\langle Dx,x\rangle$ to the Poisson objective.
This regularized problem is the resolvent, or proximal, primitive for the nonlinear Laplacian.
It asks for a point $x\in\R^V$ satisfying
\[
s\in L_H(x)+\lambda Dx.
\]
Equivalently, it asks to minimize
\[
\mathcal P_\lambda(x)=\frac12\sum_{e\in E} w_e R_e(x)^2+\frac{\lambda}{2}\langle Dx,x\rangle-\langle s,x\rangle.
\]
This regularized problem admits a particularly simple reduction: add a ground vertex and solve one augmented Poisson instance. This gives the following companion theorem.

\begin{theorem}[Informal regularized Poisson theorem]
\label{thm:intro-shifted}
Assume that every vertex has positive weighted degree \(d_v>0\). Under the dyadic input model and polynomial boundedness assumptions on the hypergraph weights and demand vector, and additionally assuming that $\lambda>0$ is $L_\lambda$-dyadic for some $L_\lambda=\log^{O(1)}P$ and satisfies
\[
P^{-K_\lambda}\le \lambda\le P^{K_\lambda}
\]
for some constant $K_\lambda\ge 1$, the regularized objective
\[
\mathcal P_\lambda(x)=\frac12\sum_{e\in E} w_e R_e(x)^2+\frac{\lambda}{2}\langle Dx,x\rangle-\langle s,x\rangle
\]
admits, for every fixed constant $C>0$, a randomized algorithm that runs in $P^{1+o(1)}$ time and, with high probability over its internal randomness, returns a rational primal point $x\in\R^V$ and an explicit dyadic regularized dual point $\widehat\eta$, with every coordinate of $x$ and $\widehat\eta$ having bit length $\log^{O(1)}P$, such that
\[
\mathcal D_\lambda(\widehat\eta)\le \mathcal D_\lambda^\star+\exp(-\log^C P),
\qquad
\mathcal P_\lambda(x)\le \mathrm{OPT}_\lambda+\exp(-\log^C P).
\]
This companion result also does not require \(H\) to be connected. It is proved formally as Theorem~\ref{thm:shifted-cut-based-solver} in Section~\ref{sec:applications-limitations}.
\end{theorem}

The polynomial-time framework of Fujii, Soma, and Yoshida \cite{FujiiSomaYoshida2018} applies here as well, since the reduction used in this work adds a ground vertex and turns this regularized cut-hypergraph problem into another submodular Laplacian system.
By contrast, Theorem~\ref{thm:intro-shifted} gives, for the cut-based model studied here, a global almost-linear-time solver for the regularized Poisson problem with a certified primal--dual objective gap; unlike existing diffusion and resolvent guarantees \cite{liu2021strongly,AmeranisEtAl2023Resolvents}, whose running time depends on localization parameters or a spectral gap, its \(P^{1+o(1)}\) running time has no such dependence.

A concrete application is personalized PageRank on hypergraphs: Section~\ref{sec:pagerank} identifies it with the degree-scaled resolvent and gives a certified vector approximation. This supplies the PageRank computation used in local and global clustering, with the qualifications about approximate sweep cuts stated there.

\subsection{Technical overview\sectionpunct}

The proof has two stages: a dual solve and a recovery stage. The chain is
\[
\begin{aligned}
\text{First stage:}\quad
&\text{Hypergraph Poisson problem}
\Longrightarrow
\text{Fenchel dual }\eta
\Longrightarrow
\text{sparse convex flow}
\\
&\Longrightarrow
\text{near-optimal feasible flow }f
\Longrightarrow
\text{hyperedge masses }\mu .
\\
\text{Second stage:}\quad
&\text{range-budget vector }r=\mu/w
\Longrightarrow
\text{support problem }L_s(r)
\\
&\Longrightarrow
\text{node potential }\pi|_V
\Longrightarrow
\text{primal potential }x .
\end{aligned}
\]
The first stage solves the dual side of the problem using the high-accuracy convex-flow framework of \cite{ChenKyngLiuPengProbstGutenbergSachdeva2025}. Before invoking that framework, we convert the Fenchel dual into a convex-flow instance on the lifted graph introduced below; the arc costs are separable, and all data are represented with finite precision.
The resulting oracle gives a near-optimal feasible flow $f$; from this flow we read both the hyperedgewise masses $\mu$ and the induced dual vector $\eta$, but this does not yet canonically determine a primal potential $x$.  This gap is genuine: the operator is nonlinear and set-valued, and the energy need not be strongly convex on $\Xzero$.  The recovery stage therefore does not read $x$ from the detailed dual routing.  Instead, it forms the range-budget vector \(r=\mu/w\), the gradient of the quadratic mass objective \(q\) defined below, and solves the support problem $L_s(r)$.  The LP-dual node potentials of this support query, restricted to the original vertex set and shifted by a constant so that they lie in \(\Xzero\), yield the primal potential \(x\).  For ordinary graphs this is the analogue of Ohm's law: the optimal flow divided by the edge weight gives a voltage drop.  For hyperedges there is no canonical signed pairwise drop inside the hyperedge, so the scalar $\mu_e/w_e$ plays the role of the allowed voltage range across the whole hyperedge, and the support query enforces consistency of these hyperedgewise budgets across vertices to produce a vertex potential.

With an exact dual optimum, complementary slackness identifies maximizers and minimizers on each hyperedge from the positive and negative supports, while the mass fixes their difference (Remark~\ref{rem:exact-dual-recovery}). Approximate optimality does not certify these supports: arbitrarily accurate lifted flows can have positive flow on every arc (Remark~\ref{rem:approximate-flow-support}). Our recovery instead uses the masses and controls the error without deciding which flow coordinates vanish.

The next paragraphs unpack this roadmap one step at a time.

\paragraph{Step 1: Fenchel duality exposes transport inside each hyperedge.}
For a fixed hyperedge $e$, the range term has the support-function representation
\[
\begin{aligned}
R_e(x)
&=\max\{\langle \eta,x\rangle:\eta\in U_e,\ \|\eta\|_1\le 2\},
\\
U_e
&:=\{\eta\in\R^V:\operatorname{supp}(\eta)\subseteq e,\ \1^\top \eta=0\}.
\end{aligned}
\]
Consequently,
\[
\frac{w_e}{2}R_e(x)^2
=
\sup_{\eta_e\in U_e}
\left\{\langle \eta_e,x\rangle-\frac{\|\eta_e\|_1^2}{8w_e}\right\}.
\]
After summing over the hyperedges and using the constraint $x\in\Xzero$, the Poisson problem has the dual minimization form
\[
\mathcal D^\star
=
\min\left\{
\sum_{e\in E}\frac{\|\eta_e\|_1^2}{8w_e}
:
\eta_e\in U_e,\ \sum_{e\in E}\eta_e=s
\right\}.
\]
This is the first structural simplification: the nonlinear max--min coupling has become a separable hyperedge cost plus a single conservation law.
Section~\ref{sec:fenchel-dual} formalizes this derivation in Theorem~\ref{thm:poisson-fenchel-dual}.

\paragraph{Step 2: the dual is an $O(P)$-arc convex flow.}
The dual variables $\eta_e$ are signed vectors on hyperedges. We split each one into positive mass, negative mass, and its common total:
\[
p_v^e:=((\eta_e)_v)^+,
\qquad
n_v^e:=((\eta_e)_v)^-,
\qquad
\mu_e:=\frac12\|\eta_e\|_1.
\]
This replaces each signed vector on one hyperedge by nonnegative flow data on a small directed gadget for that hyperedge; we call the resulting formulation the lifted flow formulation.
Then the hyperedge cost depends only on the scalar mass,
\[
\frac{\|\eta_e\|_1^2}{8w_e}=\frac{\mu_e^2}{2w_e}.
\]
Thus each hyperedge is replaced by a directed gadget with one arc $(e^-,e^+)$ whose flow is $\mu_e$ and whose cost is $\mu_e^2/(2w_e)$; we call this the \emph{quadratic arc} of \(e\).
The gadget also has zero-cost arcs $(e^+,v)$ and $(v,e^-)$ for $v\in e$ that distribute the positive and negative parts of \(\eta_e\) among the vertices of \(e\); we call these \emph{transport arcs}.
Gluing these gadgets over all hyperedges yields a lifted graph $G^\uparrow$ with $O(P)$ arcs, and the nonnegative flow on this lifted graph is what we call the lifted flow. The construction is not a clique or star expansion of the primal hypergraph objective; it is an equivalent lifted formulation of the Fenchel dual.
Figure~\ref{fig:hyperedge-gadget} shows the local gadget for one hyperedge.

To invoke this convex-flow theorem rigorously, we still have to express the lifted problem in the finite-precision model of \cite{ChenKyngLiuPengProbstGutenbergSachdeva2025}. Concretely, we pass to epigraph form, assign explicit self-concordant barriers to the transport and quadratic arcs, prove polynomial conditioning, and construct a strictly positive circulation that interiorizes any feasible flow without changing node imbalances.
After these checks, that theorem gives the first-stage oracle: the randomized convex-flow stage runs in $P^{1+o(1)}$ time and, with high probability over its internal randomness, returns coordinate-query access to an underlying high-accuracy feasible lifted flow on $G^\uparrow$.
The gadget definitions then determine the quadratic-arc mass vector $\mu$ and a corresponding dual point $\eta$ from this underlying flow.
This first-stage flow is available only through coordinate queries, not as the explicit finite certificate required in the final theorem; when the final solver needs finite data, it queries the relevant flow coordinates to high precision and repairs the induced dual point into an explicit dyadic certificate satisfying the dual feasibility equations exactly, as formalized in Lemma~\ref{lem:dual-certificate-repair}.
The induced real vectors $\mu$ and $\eta$ satisfy
\[
q(\mu)\le \mathcal D^\star+\exp(-\log^C P),
\qquad
\mathcal D(\eta)\le q(\mu),
\]
where
\[
q(\mu):=\frac12\sum_{e\in E}\frac{\mu_e^2}{w_e}.
\]
At this point we have a near-optimal oracle-represented dual point. We still have to turn the mass vector $\mu$ into a primal potential $x$.
Section~\ref{sec:dual-solver} gives the algorithmic version of the dual-to-lifted-flow reformulation; Theorem~\ref{thm:black-box-dual-solver} states the first-stage guarantee.

\paragraph{Step 3: characterizing feasible mass vectors.}
The lifted objective depends only on the quadratic-arc masses, not on the detailed routing inside each hyperedge gadget. We therefore define $\mathcal M(s)$ to be the set of mass vectors that arise from feasible lifted flows with demand $s$. In these variables, the dual optimum is exactly
\[
\mathcal D^\star=\min_{\mu\in\mathcal M(s)} q(\mu).
\]
This characterization is the first step in primal recovery.
It isolates the hyperedge-mass data that matter for the quadratic objective $q(\mu)$; the detailed routing is a feasibility witness showing that $\mu\in\mathcal M(s)$.
Geometrically, once the dual has been written as minimization of $q$ over $\mathcal M(s)$, an exact minimizer is characterized by a supporting hyperplane to $\mathcal M(s)$.
The normal vector of this hyperplane will become the range-budget vector used in the support problem.
This is also where the proof goes beyond a direct application of a general convex-flow theorem: that theorem returns an approximate point in the lifted flow polytope, whereas we project to the feasible set $\mathcal M(s)$ and then query $\mathcal M(s)$ to recover primal information.
Section~\ref{sec:mass-support} records this projection as part of Theorem~\ref{thm:mass-support}.

\paragraph{Step 4: support queries convert mass gradients into vertex potentials.}
For a budget vector $r\in\R^E_{\ge 0}$, define the support problem
\[
L_s(r):=\max\{\langle s,x\rangle:x\in\Xzero,\ R_e(x)\le r_e\ \forall e\in E\}.
\]
Before discussing how to solve this problem, let us explain why it is the right recovery problem.
If $x_r$ is a maximizer of $L_s(r)$, then $R_e(x_r)\le r_e$ for every $e$, and therefore
\begin{equation}
\label{eq:intro-support-upper-bound}
\mathcal P(x_r)
=
\frac12\sum_{e\in E}w_eR_e(x_r)^2-\langle s,x_r\rangle
\le
\frac12\sum_{e\in E}w_er_e^2-L_s(r).
\end{equation}
Thus, once a candidate budget vector $r$ has been chosen, the support problem is precisely the problem of choosing, among all potentials with those hyperedge ranges, the one that makes the linear term in the primal objective as large as possible.

The dual mass vector supplies the correct budget vector. To see this, first ignore approximation and suppose that $\mu^\star$ minimizes $q$ over $\mathcal M(s)$.
Set
\[
r_e^\star:=\frac{\mu_e^\star}{w_e}.
\]
Since $r^\star=\nabla q(\mu^\star)$, first-order optimality over $\mathcal M(s)$ gives
\[
\mu^\star\in
\arg\min_{\nu\in\mathcal M(s)}
\sum_{e\in E}r_e^\star\nu_e.
\]
The support identity, proved in Section~\ref{sec:mass-support}, is
\[
L_s(r)=\min_{\nu\in\mathcal M(s)}\sum_{e\in E}r_e\nu_e.
\]
Applying it to $r^\star$ yields
\[
L_s(r^\star)
=
\sum_{e\in E}r_e^\star\mu_e^\star
=
\sum_{e\in E}\frac{(\mu_e^\star)^2}{w_e}
=
2q(\mu^\star),
\]
where the last equality is the two-homogeneity relation
$\langle \nabla q(\mu^\star),\mu^\star\rangle=2q(\mu^\star)$.
On the other hand, \eqref{eq:intro-support-upper-bound} with $r=r^\star$ gives
\[
\mathcal P(x_{r^\star})
\le
\frac12\sum_{e\in E}w_e(r_e^\star)^2-L_s(r^\star)
=
q(\mu^\star)-2q(\mu^\star)
=
-q(\mu^\star).
\]
By strong duality, $\mathrm{OPT}=-\mathcal D^\star=-q(\mu^\star)$, so $x_{r^\star}$ is a primal minimizer.
This exact calculation is the reason for solving a support problem: the gradient of the dual mass objective gives a supporting hyperplane to the feasible mass set, and the support identity turns that hyperplane into a vertex potential.

Algorithmically, the support identity is realized as a min-cost-flow problem on the same lifted graph $G^\uparrow$: put cost $r_e$ on the quadratic arc of hyperedge $e$ and zero cost on the transport arcs.
The LP-dual variables of this min-cost flow are node potentials; the transport arcs and the quadratic arc impose potential inequalities that make the restriction to $V$ satisfy $R_e(x)\le r_e$ for every hyperedge.
Thus we use the same lifted $O(P)$-arc graph twice: first for a nonlinear convex-flow solve that returns $\mu$, and then for a linear min-cost-flow solve whose LP-dual node potentials give the primal certificate.
Section~\ref{sec:mass-support} proves this support reduction in Theorem~\ref{thm:mass-support}.

\paragraph{Step 5: approximate recovery via the support gap.}
We recover the primal point with the budget
\[
r_e:=\frac{\mu_e}{w_e},
\]
which is the gradient of the quadratic mass objective $q$ at the mass vector returned by the first stage.
If $\mu$ were an exact minimizer, the calculation in Step~4 would prove exact primal optimality.
In the algorithm, $\mu$ is only near-optimal, so the supporting-hyperplane identity is no longer exact.
The relevant error quantity is the support gap
\[
2q(\mu)-L_s(\mu/w).
\]
It vanishes in the exact calculation above.
Moreover, if $x_r$ maximizes $L_s(r)$ for $r=\mu/w$, then \eqref{eq:intro-support-upper-bound} gives
\[
\mathcal P(x_r)\le q(\mu)-L_s(r).
\]
Writing $\delta=q(\mu)-q(\mu^\star)$, we obtain
\[
\mathcal P(x_r)-\mathrm{OPT}
\le
q(\mu)-L_s(r)+q(\mu^\star)
=
2q(\mu)-L_s(r)-\delta.
\]
Thus primal recovery reduces to proving that the support gap is small.
The quantitative recovery lemma shows precisely this stability: when $\mu$ is near-optimal for $q$ over $\mathcal M(s)$, a maximizer of the support problem with budget $r=\mu/w$ is nearly optimal for the original Poisson objective.
Consequently the min-cost-flow solve for the support problem produces $x\in\Xzero$ with
\[
\mathcal P(x)\le \mathrm{OPT}+\exp(-\log^C P),
\]
while the convex-flow stage already supplies the matching dual bound.

\paragraph{Finite-precision certification.}
The conceptual recovery above is stated over real numbers. The implementation has to produce finite certificates.
The algorithm queries the masses on the quadratic arcs to high absolute precision and rounds the budgets upward to dyadic numbers. It then rounds and scales the demand vector, solves the resulting capacitated integral min-cost-flow instance exactly, rescales the returned node potentials back to the original units, and transfers them to the support maximizer. The integer magnitudes in this exact instance are bounded by $\exp(\log^{O(1)}P)$.
Acyclicity of an optimal support flow gives the needed capacity bound, so this rounding step does not introduce a hidden superlinear bottleneck.
This introductory overview is the reader's map; the full certified procedure, including the finite-precision steps, is stated in Section~\ref{sec:primal-recovery} as Algorithm~\ref{alg:certified-poisson}. The formal solver guarantee, Theorem~\ref{thm:primal-recovery}, states that this algorithm runs in \(P^{1+o(1)}\) time and returns an explicit primal point and an explicit dyadic dual certificate with exact dual feasibility and the additive bounds above.

\paragraph{Interpreting the primal objective gap.}
The theorem is stated as an additive objective guarantee because $\mathcal E_H$ is generally not strongly convex on $\Xzero$, and hence there need not be a unique solution whose distance can be measured in a canonical norm. Section~\ref{sec:primal-recovery} shows that this additive gap is the natural nonlinear analogue of the graph Laplacian energy-norm error: it is the Bregman divergence to the optimal set induced by $\mathcal E_H$.

\paragraph{Regularized problems.}
Once the Poisson solver is available, regularized objectives require no new flow machinery. The regularization term $\frac{\lambda}{2}\langle Dx,x\rangle$ is represented by adding a ground vertex $g$ and two-vertex hyperedges $\{v,g\}$ of weight $\lambda d_v$. Solving the resulting augmented Poisson instance and shifting the returned potential so that the ground coordinate is zero gives Theorem~\ref{thm:intro-shifted}.

\subsection{Related work\sectionpunct}

\paragraph{Graph Laplacian solvers.}
For ordinary graphs, the Poisson equation is a linear Laplacian system, and a long line of work established nearly-linear-time solvers for Laplacian and more general SDD systems, beginning with Spielman and Teng \cite{SpielmanTeng2014SDD} and followed by major refinements and simplifications due to Koutis, Miller, and Peng \cite{KoutisMillerPeng2014SDD}, Kelner, Orecchia, Sidford, and Zhu \cite{KelnerOrecchiaSidfordZhu2013}, and Kyng and Sachdeva \cite{KyngSachdeva2016}. This literature is the baseline complexity benchmark for the present work.
The present hypergraph problem is a higher-order nonlinear analogue of that graph primitive. What makes the extension nontrivial is precisely the loss of linearity: matrix inversion must be replaced by a dual transport formulation, and a primal potential must then be recovered from the resulting nonlinear certificate.

\paragraph{Cut-based hypergraph Laplacians.}
The operator $L_H=\partial \mathcal E_H$ studied here belongs to the cut-based / Lov\'asz-extension line of work on hypergraph Laplacians. The hyperedge term $R_e(x)$ is the Lov\'asz-extension contribution of the cut hyperedge $e$, and $\mathcal E_H$ is the cut-based hypergraph energy associated with that construction \cite{Hein2013TVH}. In this line of work, Hein, Setzer, Jost, and Rangapuram gave first-order algorithms for related regularized learning problems, including a primal--dual scheme in which the primal variable is updated throughout the computation rather than recovered afterward from a high-accuracy dual certificate \cite{Hein2013TVH}. Nonlinear hypergraph Laplacian operators of this max--min diffusion type were introduced and studied algorithmically by Chan, Louis, Tang, and Zhang \cite{ChanLouisTangZhang2018HypergraphLaplacian}. The focus here remains entirely on this cut-based setting and on an almost-linear-time solver for the Poisson problem on $\Xzero$ together with an explicit primal certificate.

\paragraph{Hypergraph sparsification.}
Hypergraph sparsification, and in particular spectral sparsification, is another standard algorithmic line for hypergraph Laplacians and cut objectives.
Soma and Yoshida introduced spectral sparsification for undirected and directed hypergraphs and gave polynomial-time constructions preserving the corresponding cut-based energy \(\mathcal E_H\) \cite{SomaYoshida2019}.
Bansal, Svensson, and Trevisan then gave an intermediate improvement in sparsifier size together with alternative additive sparsification notions for graphs and hypergraphs \cite{BansalSvenssonTrevisan2019}.
Kapralov, Krauthgamer, Tardos, and Yoshida next gave $O^*(nr)$-size spectral sparsifiers together with lower bounds on hypergraph compression \cite{KapralovKrauthgamerTardosYoshida2021TightBounds}. Here $n=|V|$, $r=\max_{e\in E}|e|$ is the maximum hyperedge size, and $O^*$ suppresses factors of $(\varepsilon^{-1}\log n)^{O(1)}$ for sparsification error $\varepsilon\in(0,1)$.
They later improved this to nearly-linear-size spectral sparsifiers \cite{KapralovKrauthgamerTardosYoshida2021}. Independently, Lee obtained the same asymptotic sparsifier-size bound by a chaining-based argument \cite{Lee2023SpectralHypergraphSparsificationViaChaining}, while Jambulapati, Liu, and Sidford obtained similar bounds together with a fast construction of spectral hypergraph sparsifiers \cite{JambulapatiLiuSidford2023}.
These works are complementary to the present goal: they compress a hypergraph while approximately preserving its global quadratic form and cut structure, whereas our goal is to solve the original hypergraph Poisson problem itself and recover an explicit primal witness on the given instance.

\paragraph{Submodular Laplacians.}
The cut-based Laplacian is a special case of the broader submodular-Laplacian framework introduced by Yoshida \cite{Yoshida2019CheegerSubmodular}. Fujii, Soma, and Yoshida gave polynomial-time algorithms for general submodular Laplacian systems and related regression problems \cite{FujiiSomaYoshida2018}. At the model level, our operator is therefore contained in that framework, and their results already solve both the cut-based Poisson and regularized Poisson problems considered here. The main difference is complexity and certification: by specializing to the cut-based setting, the present result gives a randomized solver running in $P^{1+o(1)}$ time together with explicit primal recovery and certified primal--dual objective gaps.
This improvement comes from the special structure of cut hyperedges: arbitrary submodular transformations retain the geometry of general base polytopes, whereas a cut hyperedge has a conjugate depending only on $\|\eta_e\|_1$. This scalarization is exactly what permits the $O(P)$-arc lifted flow, the characterization of feasible mass vectors, and the recovery of an explicit primal potential through the support flow.

\paragraph{Hypergraph diffusions, resolvents, and learning.}
Within the same cut-based / Lov\'asz-extension line, Takai, Miyauchi, Ikeda, and Yoshida study hypergraph PageRank and develop local and global clustering algorithms with conductance guarantees \cite{TakaiMiyauchiIkedaYoshida2020}. Ikeda, Miyauchi, Takai, and Yoshida also study hypergraph heat flow and use it to find Cheeger cuts in hypergraphs \cite{IkedaMiyauchiTakaiYoshida2022}. These works start from the same operator family studied here, but use it for downstream clustering via PageRank vectors or heat-flow trajectories. By contrast, Theorem~\ref{thm:intro-shifted} gives a global almost-linear-time solver with a certified primal--dual objective gap for the corresponding regularized Poisson / resolvent problem.

The distinction from regularized QDSFM is also important. Li, He, and Milenkovic study quadratic decomposable submodular function minimization (QDSFM), that is, a regularized framework built from sums of squared Lov\'asz-extension terms, where the regularizer makes the primal solution unique and gives an explicit recovery map from dual variables \cite{LiHeMilenkovic2020}. The present Poisson objective on $\Xzero$ has no such strong convexity: the optimal set may be non-singleton, and a near-optimal lifted dual flow does not by itself identify a primal potential. In the cut-based Poisson setting, the support-query stage precisely replaces the regularized recovery map.

Beyond those cut-based diffusion works, recent work has also studied related diffusion, resolvent, and local clustering objectives for broader hypergraph Laplacians built from local hyperedge norms rather than the cut-based square-range energy used here \cite{AmeranisEtAl2023Resolvents,liu2021strongly}. Ameranis, Chen, DePavia, Orecchia, and Tani give a nearly-linear-time resolvent algorithm with an explicit inverse dependence on the hypergraph spectral gap, while Liu, Veldt, Song, Li, and Gleich give strongly local algorithms whose runtime depends on localization parameters and output size. These results are therefore related in task type, but differ both in model and in output: they treat broader Laplacians and return diffusion/resolvent vectors or local clustering guarantees, whereas Theorem~\ref{thm:intro-shifted} gives a global certified solver for the cut-based regularized Poisson problem. On the learning side, hypergraph-regularized diffusion and energy-minimization procedures are also used in semi-supervised learning and representation learning \cite{FujiiSomaYoshida2018,prokopchik2022nonlinear,wang2023energy}. No downstream analysis of those pipelines is attempted here, but the result does provide a faster certified inner solver when that particular cut-based regularized objective is the one of interest.

\subsection{Organization\sectionpunct}
Section~\ref{sec:preliminaries} collects basic notation and convex-analytic definitions. Section~\ref{sec:fenchel-dual} carries out Step~1 by fixing the hyperedge-local notation, proving the support-function form of the hyperedge range, and deriving the Fenchel dual of the Poisson problem restricted to \(\Xzero\). Section~\ref{sec:dual-solver} carries out Step~2 by rewriting that dual as a lifted $O(P)$-arc flow problem and applying the framework of \cite{ChenKyngLiuPengProbstGutenbergSachdeva2025} to obtain a high-accuracy lifted flow from the first-stage oracle. Section~\ref{sec:mass-support} carries out Steps~3--4 by identifying the feasible mass vectors and reducing the needed support problems to linear min-cost flow on the same lifted graph. Section~\ref{sec:primal-recovery} carries out Step~5 by turning these ingredients into a certified solver with finite-precision output; this includes the repair that turns queryable first-stage coordinates into explicit dyadic dual certificates and the interpretation of the objective gap as a nonlinear energy distance. Section~\ref{sec:applications-limitations} reduces regularized objectives to the Poisson solver by adding a ground vertex and records direct consequences such as pairwise responses and resolvents, together with the corresponding certification limitations. Appendix~\ref{app:chen-black-box} collects the notation, epigraph, and incidence-convention details behind the black-box invocation of \cite{ChenKyngLiuPengProbstGutenbergSachdeva2025}.
 \section{Preliminaries\sectionpunct}
\label{sec:preliminaries}
We collect the notation and convex-analytic definitions used in the reductions. The input model and polynomial bounds are as in Definition~\ref{def:dyadic-input} and the accompanying assumptions; the constructions specific to the solver are introduced in the following sections.

\paragraph{Vectors and weighted norms.}
For a finite index set $S$, we write $\R^S$ for the space of real vectors indexed by $S$, $\1$ for the all-ones vector, and $\ee_v$ for the coordinate vector at $v\in S$. The pairing is $\ip{x}{y}=\sum_{v\in S}x_vy_v$. For $x\in\R^S$, its support is $\supp(x)=\{v\in S:x_v\ne0\}$, and its positive and negative parts are $(x_+)_v=\max\{x_v,0\}$ and $(x_-)_v=\max\{-x_v,0\}$. Thus $x=x_+-x_-$.
For the positive diagonal degree matrix $D$ from the introduction, set
\[
\normD{x}^2=\ip{Dx}{x},\qquad
\normDinv{x}^2=\ip{D^{-1}x}{x}.
\]
Vector inequalities and division of vectors of hyperedge weights or masses are interpreted coordinatewise. We use $\widetilde O$ to suppress factors polynomial in the logarithm of the input size; ``nearly linear'' refers to such a bound, whereas ``almost linear'' allows the factor $P^{o(1)}$ in $P^{1+o(1)}$.

\paragraph{Convex functions.}
For a proper convex function $F:\R^S\to\R\cup\{+\infty\}$ and a point $x$ with $F(x)<+\infty$, the subdifferential is
\[
\partial F(x)=\{g\in\R^S:F(y)\ge F(x)+\ip{g}{y-x}\text{ for all }y\in\R^S\}.
\]
In particular, $0\in\partial F(x)$ characterizes a minimizer.

The Fenchel conjugate, defined for $g\in\R^S$, is $F^*(g)=\sup_{x\in\R^S}\{\ip{g}{x}-F(x)\}$. Fenchel's inequality states that $F(x)+F^*(g)\ge\ip{g}{x}$, with equality precisely when $g\in\partial F(x)$. For a nonempty set $K\subseteq\R^S$, its support function is $\sigma_K(x)=\sup_{g\in K}\ip{g}{x}$, and its indicator function $\delta_K$ is zero on $K$ and $+\infty$ outside $K$. These definitions will express hyperedge ranges as support functions and turn the Poisson objective into a flow problem by conjugacy.

\paragraph{Directed flows.}
For a directed graph $G=(W,A)$, let $\delta^-(u)$ and $\delta^+(u)$ be the sets of arcs entering and leaving $u\in W$. Its node--arc incidence matrix $A_G\in\R^{W\times A}$ uses incoming minus outgoing flow:
\[
(A_Gf)_u=\sum_{a\in\delta^-(u)}f_a-\sum_{a\in\delta^+(u)}f_a.
\]
A nonnegative flow $f\in\R_{\ge0}^A$ with demand $b\in\R^W$ satisfies $A_Gf=b$; a circulation satisfies $A_Gf=0$. For costs $c\in\R^A$, the linear program minimizing $c^\top f$ subject to these constraints has dual
\[
\max\{b^\top\pi:\pi\in\R^W,\ A_G^\top\pi\le c\}.
\]
Here $\pi$ is a vector of node potentials. In particular, an arc $(u,v)$ imposes $\pi_v-\pi_u\le c_{(u,v)}$. We will use such inequalities to recover a vertex potential from a min-cost-flow dual solution.
 \section{From Hyperedge Ranges to the Fenchel Dual\sectionpunct}
\label{sec:fenchel-dual}

This section carries out Step~1 of the roadmap. Starting from the support-function form of the hyperedge range, it derives the Fenchel dual of the Poisson problem restricted to $\Xzero$. The output is a dual with hyperedge-local variables $\eta_e\in U_e$ and the global balance constraint $B\eta=s$.

\subsection{Fixing the additive constant and hyperedge-local notation\sectionpunct}

Except where explicitly stated otherwise, we assume throughout that $H=(V,E,w)$ is a connected weighted hypergraph with strictly positive hyperedge weights.
The connectedness assumption is made only to keep notation light.
Every statement in the remainder of this paper extends componentwise to disconnected hypergraphs by fixing the additive constant separately on each connected component and requiring the demand vector to have zero sum on each component.
We recall the weighted-degree matrix and the weighted-mean-zero subspace
\[
d_v:=\sum_{e\ni v} w_e,
\qquad
D=\operatorname{diag}(d_v),
\qquad
\Xzero:=\{x\in\R^V:\ip{Dx}{\1}=0\}.
\]

\begin{definition}[The hypergraph Laplacian]
\label{def:hypergraph-laplacian}
The hypergraph Laplacian is the cut-based nonlinear operator induced by the energy
\[
\mathcal E_H(x):=\frac12\sum_{e\in E} w_e R_e(x)^2.
\]
We write
\[
L_H:=\partial \mathcal E_H
\]
for the subdifferential of $\mathcal E_H$, viewed as this set-valued operator.
\end{definition}

\begin{remark}[The Poisson problem on $\Xzero$]
\label{rem:xzero-poisson}
The optimization problem \eqref{eq:poisson-primal} with $x$ restricted to $\Xzero$ is the corresponding hypergraph Laplacian / Poisson problem after fixing the additive constant: its optimal solutions are the points $x\in\Xzero$ satisfying
\[
s\in L_H(x)+\operatorname{span}(D\1).
\]
\end{remark}

For each hyperedge $e\in E$, define the hyperedge-local zero-sum subspace
\[
U_e:=\{\eta\in\R^V:\supp(\eta)\subseteq e,\ \1^\top \eta=0\}.
\]
Also define
\[
K_e:=\operatorname{conv}\{\ee_u-\ee_v:\ u,v\in e,\ \text{not necessarily distinct}\}\subseteq\R^V.
\]

\subsection{Hyperedge range as a support function\sectionpunct}

Recall that for each hyperedge $e\in E$ and vector $x\in\R^V$,
\[
R_e(x):=\max_{u\in e} x_u-\min_{v\in e} x_v
\]
is the hyperedge range of $x$ on $e$. The key point is that $R_e$ is the support function of $K_e$.
\begin{lemma}
\label{lem:hyperedge-support}
For every hyperedge $e\in E$,
\[
K_e=\{\eta\in U_e:\norm{\eta}_1\le 2\}.
\]
Consequently,
\[
R_e(x)=\max_{\eta\in K_e}\ip{\eta}{x}
=\max_{\eta\in U_e,\;\norm{\eta}_1\le 2}\ip{\eta}{x}.
\]
\end{lemma}

\begin{proof}
Every generator $\ee_u-\ee_v$ lies in $U_e$ and has $\ell_1$ norm at most $2$, so
\[
K_e\subseteq \{\eta\in U_e:\norm{\eta}_1\le 2\}.
\]
For the reverse inclusion, take $\eta\in U_e$ with $\norm{\eta}_1\le 2$ and write
\[
\eta=\eta^+-\eta^-,
\qquad
\eta^+,\eta^-\ge 0.
\]
Since $\1^\top \eta=0$,
\[
\1^\top \eta^+=\1^\top \eta^-=:m=\frac12\norm{\eta}_1\le 1.
\]
Choose any transport plan $\alpha_{uv}\ge 0$ from $\eta^+$ to $\eta^-$ so that
\[
\sum_{v\in e}\alpha_{uv}=\eta_u^+,
\qquad
\sum_{u\in e}\alpha_{uv}=\eta_v^-.
\]
Then
\[
\eta=\sum_{u,v\in e}\alpha_{uv}(\ee_u-\ee_v).
\]
If $m<1$, pad with weight $1-m$ on the zero vector, which belongs to $K_e$ by the not-necessarily-distinct convention. This shows $\eta\in K_e$.

For the support-function formula, let $M:=\max_{u\in e} x_u$ and $m:=\min_{v\in e} x_v$. Every $\eta\in K_e$ is a convex combination of vectors $\ee_u-\ee_v$, so
\[
\ip{\eta}{x}\le M-m=R_e(x).
\]
Equality is attained by choosing $\eta=\ee_{u^\star}-\ee_{v^\star}$ for any maximizer $u^\star$ and minimizer $v^\star$. Hence $R_e=\sigma_{K_e}$.
\end{proof}
 \subsection{One-hyperedge conjugate and aggregate feasibility\sectionpunct}

With the support-function description of $R_e$ in hand, we begin by computing the conjugate of a single hyperedge term $\psi_e$ and by recording which right-hand sides can be obtained by summing hyperedge-local zero-sum flows.

For each hyperedge $e\in E$, define
\[
\psi_e(x):=\frac{w_e}{2}R_e(x)^2.
\]
We begin with the one-hyperedge conjugate formula that underlies the whole paper.

\begin{proposition}[One-hyperedge conjugate]
\label{prop:hyperedge-conjugate}
For every hyperedge $e\in E$ and every $x\in\R^V$,
\[
\psi_e(x)
=
\sup_{\eta_e\in U_e}
\left\{
\ip{\eta_e}{x}-\frac{\norm{\eta_e}_1^2}{8w_e}
\right\}.
\]
Equivalently,
\[
\psi_e^\ast(\eta_e)
:=
\sup_{x\in\R^V}
\left\{
\ip{\eta_e}{x}-\psi_e(x)
\right\}
=
\begin{cases}
\dfrac{\norm{\eta_e}_1^2}{8w_e}, & \eta_e\in U_e,\\[1ex]
+\infty, & \text{otherwise.}
\end{cases}
\]
\end{proposition}

\begin{proof}
By Lemma~\ref{lem:hyperedge-support},
\[
R_e(x)=\max_{\kappa\in K_e}\ip{\kappa}{x}.
\]
For every scalar $\rho\ge 0$,
\[
\frac{w_e}{2}R_e(x)^2
=
\sup_{\rho\ge 0}
\left\{
\rho R_e(x)-\frac{\rho^2}{2w_e}
\right\}
=
\sup_{\rho\ge 0,\;\kappa\in K_e}
\left\{
\ip{\rho\kappa}{x}-\frac{\rho^2}{2w_e}
\right\}.
\]
Now $\eta=\rho\kappa$ with $\rho\ge 0$ and $\kappa\in K_e$ if and only if $\eta\in U_e$ and
\[
\norm{\eta}_1\le 2\rho.
\]
For fixed $\eta\in U_e$, the smallest admissible $\rho$ is $\norm{\eta}_1/2$, and this minimizes the penalty $\rho^2/(2w_e)$. Therefore
\[
\frac{w_e}{2}R_e(x)^2
=
\sup_{\eta_e\in U_e}
\left\{
\ip{\eta_e}{x}-\frac{\norm{\eta_e}_1^2}{8w_e}
\right\},
\]
which is the claimed conjugate formula.
\end{proof}

Let
\[
\mathcal Y:=\prod_{e\in E} U_e
\]
and define the aggregation map
\[
B:\mathcal Y\to\R^V,
\qquad
B\eta:=\sum_{e\in E} \eta_e.
\]

\begin{lemma}[Range of the aggregation map]
\label{lem:B-range}
If $H$ is connected, then
\[
B(\mathcal Y)=\{z\in\R^V:\1^\top z=0\}.
\]
In particular, for every $s\in\R^V$ with $\1^\top s=0$, the affine set
\[
\{\eta\in\mathcal Y:\ B\eta=s\}
\]
is nonempty.
\end{lemma}

\begin{proof}
Every $\eta\in\mathcal Y$ satisfies
\[
\1^\top B\eta=\sum_{e\in E}\1^\top \eta_e=0,
\]
so
\[
B(\mathcal Y)\subseteq \{z\in\R^V:\1^\top z=0\}.
\]
For the reverse inclusion, it suffices to show that
\[
\ee_u-\ee_v\in B(\mathcal Y)
\qquad
\text{for all }u,v\in V,
\]
because the zero-sum subspace is spanned by these differences. Fix $u,v\in V$. Since $H$ is connected, there exist vertices
\[
p_0=u,\ p_1,\ \dots,\ p_k=v
\]
and hyperedges $e_1,\dots,e_k\in E$ such that
\[
p_{i-1},p_i\in e_i
\qquad
(1\le i\le k).
\]
For each $i=1,\dots,k$, let $\xi^{(i)}\in\mathcal Y$ be the element whose $e_i$-coordinate is
\[
\ee_{p_{i-1}}-\ee_{p_i}\in U_{e_i},
\]
and whose other coordinates are zero. Then
\[
B\Bigl(\sum_{i=1}^k \xi^{(i)}\Bigr)
=
\sum_{i=1}^k (\ee_{p_{i-1}}-\ee_{p_i})
=
\ee_u-\ee_v.
\]
Hence $\ee_u-\ee_v\in B(\mathcal Y)$, proving the reverse inclusion.
\end{proof}

\subsection{Fenchel dual of the Poisson problem\sectionpunct}

We now combine the one-hyperedge conjugate formula and the aggregation lemma into a Fenchel dual of the Poisson problem with $x$ restricted to $\Xzero$.

\begin{theorem}[Fenchel dual of the hypergraph Poisson problem with $x$ restricted to $\Xzero$]
\label{thm:poisson-fenchel-dual}
Assume $\1^\top s=0$. Then the affine set
\[
\{\eta\in\mathcal Y:\ B\eta=s\}
\]
is nonempty, and
\begin{equation}
\label{eq:poisson-duality}
\mathrm{OPT}
=
\min_{x\in\Xzero} \mathcal P(x)
=
-\min_{\eta\in\mathcal Y,\;B\eta=s}
\sum_{e\in E}\frac{\norm{\eta_e}_1^2}{8w_e}.
\end{equation}
In particular, the dual problem is
\begin{equation}
\label{eq:dual-problem}
\mathcal D^\star
:=
\min_{\eta\in\mathcal Y,\;B\eta=s}
\mathcal D(\eta),
\qquad
\mathcal D(\eta):=\sum_{e\in E}\frac{\norm{\eta_e}_1^2}{8w_e},
\end{equation}
and $\mathrm{OPT}=-\mathcal D^\star$.
\end{theorem}

\begin{proof}
Set
\[
Z:=\prod_{e\in E}\R^V
\]
with the product pairing
\[
\ip{\eta}{\zeta}_Z:=\sum_{e\in E}\ip{\eta_e}{\zeta_e},
\qquad
\eta,\zeta\in Z.
\]
We view $\mathcal Y$ as a linear subspace of $Z$ by zero-padding each hyperedge-local vector.
Define the linear map
\[
A:\R^V\to Z,
\qquad
Ax:=(x)_{e\in E}.
\]
Its adjoint is
\[
A^\ast \eta=\sum_{e\in E}\eta_e,
\]
so $A^\ast \eta=B\eta$ for every $\eta\in\mathcal Y$.

Define
\[
f(x):=\delta_{\Xzero}(x)-\ip{s}{x},
\qquad
\Phi(z):=\sum_{e\in E}\psi_e(z_e)
\qquad
(z=(z_e)_{e\in E}\in Z).
\]
Then
\[
\mathrm{OPT}
=
\inf_{x\in\R^V}\{f(x)+\Phi(Ax)\}.
\]
By Proposition~\ref{prop:hyperedge-conjugate},
\[
\Phi^\ast(\eta)
=
\begin{cases}
\sum_{e\in E}\dfrac{\norm{\eta_e}_1^2}{8w_e}, & \eta\in\mathcal Y,\\[1ex]
+\infty, & \text{otherwise.}
\end{cases}
\]
Since $0\in\Xzero$ and $\Phi$ is finite and continuous on $Z$, the Fenchel--Rockafellar qualification condition holds, so the standard duality theorem applies; see, for example, \cite{Rockafellar1966FenchelExtension,Rockafellar1970ConvexAnalysis}. Therefore
\[
\mathrm{OPT}
=
\sup_{\eta\in Z}
\left\{
-f^\ast(-A^\ast\eta)-\Phi^\ast(\eta)
\right\}.
\]
Moreover, by the standard identities for conjugates of indicator functions and support functions, again see \cite{Rockafellar1970ConvexAnalysis},
\[
f^\ast(y)
=
\sup_{x\in\Xzero}\ip{y+s}{x}
=
\begin{cases}
0, & y+s\in\Xzero^\perp=\operatorname{span}(D\1),\\
+\infty, & \text{otherwise.}
\end{cases}
\]
Hence
\[
\mathrm{OPT}
=
\sup_{\eta\in\mathcal Y,\;B\eta-s\in\operatorname{span}(D\1)}
\bigl(-\mathcal D(\eta)\bigr).
\]

Now every $\eta\in\mathcal Y$ satisfies
\[
\1^\top B\eta=\sum_{e\in E}\1^\top \eta_e=0,
\]
and by assumption $\1^\top s=0$. Thus every dual-feasible $\eta$ satisfies
\[
\1^\top(B\eta-s)=0.
\]
If also $B\eta-s\in\operatorname{span}(D\1)$, then
\[
B\eta-s=cD\1
\]
for some scalar $c$. Taking the coordinate sum gives
\[
0=\1^\top(B\eta-s)=c\,\1^\top D\1.
\]
Since $D\1$ has strictly positive entries, we have $\1^\top D\1>0$, so $c=0$. Therefore the compatibility condition reduces to
\[
B\eta=s.
\]
Thus
\[
\mathrm{OPT}
=
\sup_{\eta\in\mathcal Y,\;B\eta=s}
\bigl(-\mathcal D(\eta)\bigr).
\]

By Lemma~\ref{lem:B-range}, the feasible set $\{\eta\in\mathcal Y:\ B\eta=s\}$ is nonempty. Also,
\[
\mathcal D(\eta)
\ge
\frac{1}{8\max_{e\in E} w_e}\sum_{e\in E}\norm{\eta_e}_1^2
\ge
\frac{1}{8|E|\max_{e\in E} w_e}\left(\sum_{e\in E}\norm{\eta_e}_1\right)^2,
\]
so $\mathcal D$ is coercive on the finite-dimensional space $\mathcal Y$. Hence $\mathcal D$ attains a minimum on the closed affine set $\{\eta\in\mathcal Y:\ B\eta=s\}$. Consequently,
\[
\mathrm{OPT}
=
-\min_{\eta\in\mathcal Y,\;B\eta=s}\mathcal D(\eta)
=
-\mathcal D^\star,
\]
which yields \eqref{eq:poisson-duality} and \eqref{eq:dual-problem}.
\end{proof}

\begin{remark}[Recovery from an exact dual optimum]
\label{rem:exact-dual-recovery}
Let $\eta^\star\in\mathcal Y$ be a dual optimum with $B\eta^\star=s$, and set $\mu_e^\star=\|\eta_e^\star\|_1/2$. For $x\in\Xzero$, its gap against this dual optimum decomposes as
\begin{align*}
\mathcal P(x)+\mathcal D(\eta^\star)
&=\sum_{e\in E}\left(\frac{w_e}{2}R_e(x)^2+
\frac{(\mu_e^\star)^2}{2w_e}-\ip{\eta_e^\star}{x}\right)\\
&=\sum_{e\in E}\left[\frac{w_e}{2}\left(R_e(x)-\frac{\mu_e^\star}{w_e}\right)^2
+\mu_e^\star R_e(x)-\ip{\eta_e^\star}{x}\right].
\end{align*}
Each summand is nonnegative, since $\ip{\eta_e^\star}{x}\le\mu_e^\star R_e(x)$. Consequently, $x$ is primal optimal if and only if, for every $e\in E$,
\begin{equation}
\label{eq:exact-dual-complementarity}
\begin{aligned}
R_e(x)&=\mu_e^\star/w_e,\\
\supp((\eta_e^\star)_+)&\subseteq\{v\in e:x_v=\max_{u\in e}x_u\},\\
\supp((\eta_e^\star)_-)&\subseteq\{v\in e:x_v=\min_{u\in e}x_u\}.
\end{aligned}
\end{equation}
Indeed, equality in the second term requires every vertex carrying positive mass to attain the maximum and every vertex carrying negative mass to attain the minimum. If $\mu_e^\star=0$, both supports are empty and the remaining condition is $R_e(x)=0$.
Thus the supports and values of an exact dual optimum permit recovery by solving for consistent vertex potentials. More explicitly, one may introduce bounds $\ell_e,u_e$ with $\ell_e\le x_v\le u_e$ for $v\in e$, impose $u_e-\ell_e=\mu_e^\star/w_e$, set $x_v=u_e$ on the positive support and $x_v=\ell_e$ on the negative support, and require $x\in\Xzero$. This linear feasibility formulation explains exact recovery; it does not by itself give a robust algorithm for approximately optimal dual inputs.
\end{remark}
 \section{A Lifted Flow Formulation and the Solver\sectionpunct}
\label{sec:dual-solver}

This section carries out Step~2 of the roadmap. Its input is the hyperedge-local Fenchel dual from Section~\ref{sec:fenchel-dual}; it replaces the signed variables $\eta_e$ by nonnegative arc flows on a lifted directed graph $G^\uparrow$ with $O(P)$ arcs. The output, stated in Theorem~\ref{thm:black-box-dual-solver}, is a queryable feasible lifted flow $f$, its mass vector $\mu_e=f_{(e^-,e^+)}$, and the induced dual point.
Subsection~\ref{sec:lifted-flow} gives the exact lift, Subsection~\ref{sec:convex-flow} verifies the hypotheses of the convex-flow framework of \cite{ChenKyngLiuPengProbstGutenbergSachdeva2025}, and Subsection~\ref{sec:main-theorem} states the black-box first-stage guarantee.

\subsection{A lifted flow formulation\sectionpunct}
\label{sec:lifted-flow}

Section~\ref{sec:fenchel-dual} gave a dual with one signed zero-sum vector $\eta_e$ on each hyperedge. That formulation is still not in a form that a fast flow algorithm can use directly, because its variables are signed hyperedge-local vectors rather than nonnegative arc flows on a sparse graph.

The key step of this section is an equivalent lifted flow formulation of that dual.
For each hyperedge $e$, we split the signed vector $\eta_e$ into its positive part, its negative part, and their common total mass:
\[
p_{ev}:=((\eta_e)_v)_+,
\qquad
n_{ev}:=((\eta_e)_v)_-,
\qquad
\mu_e:=\frac12\norm{\eta_e}_1.
\]
Because $\eta_e\in U_e$ has zero total mass, its positive and negative parts have the same total mass, namely $\mu_e$. The proposition below shows that these variables rewrite the dual using only nonnegative variables and conservation-type constraints. Throughout this section, when we focus on a single hyperedge $e$ we freely identify $\eta_e\in U_e$ with its local coordinate vector in $\R^e$.

For each hyperedge $e$, the variables $p_{ev}$ and $n_{ev}$ below record how much positive and negative mass of $\eta_e$ is assigned to the vertex $v$, while $\mu_e$ records their common total mass.

\begin{proposition}[Equivalent lifted flow formulation]
\label{prop:lifted-flow-equivalence}
The dual \eqref{eq:dual-problem} is equivalent to
\begin{equation}
\label{eq:lifted-flow}
\min_{p,n,\mu}
\left\{
\sum_{e\in E}\frac{\mu_e^2}{2w_e}
:
\begin{array}{l}
p_{ev}\ge 0,\ n_{ev}\ge 0,\ \mu_e\ge 0,\\[0.3em]
\sum_{e\ni v}(p_{ev}-n_{ev})=s_v \qquad (\forall v\in V),\\[0.3em]
\sum_{v\in e} p_{ev}=\mu_e \qquad (\forall e\in E),\\[0.3em]
\sum_{v\in e} n_{ev}=\mu_e \qquad (\forall e\in E).
\end{array}
\right\}.
\end{equation}
More precisely:
\begin{enumerate}[label=(\roman*),leftmargin=1.5em]
\item
Given any feasible $(p,n,\mu)$ for \eqref{eq:lifted-flow}, define
\[
(\eta_e)_v:=p_{ev}-n_{ev}\qquad (v\in e),
\]
and zero-pad outside $e$. Then $\eta_e\in U_e$ for every $e$, $B\eta=s$, and
\[
\mathcal D(\eta)\le \sum_{e\in E}\frac{\mu_e^2}{2w_e}.
\]

\item
Given any feasible $\eta$ for \eqref{eq:dual-problem}, define
\[
p_{ev}:=((\eta_e)_v)_+,
\qquad
n_{ev}:=((\eta_e)_v)_-,
\qquad
\mu_e:=\frac12\norm{\eta_e}_1.
\]
Then $(p,n,\mu)$ is feasible for \eqref{eq:lifted-flow}, and
\[
\sum_{e\in E}\frac{\mu_e^2}{2w_e}
=
\mathcal D(\eta).
\]
\end{enumerate}
Consequently, \eqref{eq:lifted-flow} and \eqref{eq:dual-problem} have the same optimal value.
\end{proposition}

\begin{proof}
For part (i), define $\eta$ from $(p,n,\mu)$ as above. For every hyperedge,
\[
\sum_{v\in e} (\eta_e)_v
=
\sum_{v\in e} p_{ev}-\sum_{v\in e} n_{ev}
=
\mu_e-\mu_e
=
0,
\]
so $\eta_e\in U_e$. Also,
\[
(B\eta)_v
=
\sum_{e\ni v} (\eta_e)_v
=
\sum_{e\ni v}(p_{ev}-n_{ev})
=
s_v.
\]
Finally,
\[
\norm{\eta_e}_1
=
\sum_{v\in e}|p_{ev}-n_{ev}|
\le
\sum_{v\in e}(p_{ev}+n_{ev})
=
2\mu_e,
\]
which implies
\[
\frac{\norm{\eta_e}_1^2}{8w_e}
\le
\frac{\mu_e^2}{2w_e}.
\]
Summing over $e$ proves the claim.

For part (ii), since $\eta_e\in U_e$, its positive and negative parts have the same total mass:
\[
\sum_{v\in e} ((\eta_e)_v)_+
=
\sum_{v\in e} ((\eta_e)_v)_-
=
\frac12\norm{\eta_e}_1.
\]
Therefore
\[
\sum_{v\in e} p_{ev}=\mu_e,
\qquad
\sum_{v\in e} n_{ev}=\mu_e.
\]
Also,
\[
\sum_{e\ni v}(p_{ev}-n_{ev})
=
\sum_{e\ni v} (\eta_e)_v
=
(B\eta)_v
=
s_v.
\]
Thus $(p,n,\mu)$ is feasible. The identity in the objective is immediate from $\mu_e=\norm{\eta_e}_1/2$.
\end{proof}

The proposition leaves us with three kinds of nonnegative variables and only conservation-type constraints. We now package them as arc flows on a single directed graph with $O(P)$ arcs.

\paragraph{Lifted directed graph.}
For each hyperedge $e$, we create two lifted nodes $e^+$ and $e^-$ associated with $e$; these are nodes of the lifted graph, not hyperedges.
Define the directed graph
\[
G^\uparrow=(V^\uparrow,A^\uparrow)
\]
with node set
\[
V^\uparrow
:=
V\cup \{e^+,e^-:e\in E\}
\]
and arc set
\[
A^\uparrow
:=
\{(e^+,v):v\in e\}
\cup
\{(v,e^-):v\in e\}
\cup
\{(e^-,e^+):e\in E\}.
\]

\begin{definition}[Transport and quadratic arcs]
\label{def:transport-quadratic-arcs}
For the lifted graph \(G^\uparrow\), define
\[
A^\uparrow_{\mathrm{tr}}
:=
\{(e^+,v):v\in e\}
\cup
\{(v,e^-):v\in e\},
\qquad
A^\uparrow_{\mathrm{quad}}
:=
\{(e^-,e^+):e\in E\}.
\]
We call arcs in \(A^\uparrow_{\mathrm{tr}}\) \emph{transport arcs} and arcs in \(A^\uparrow_{\mathrm{quad}}\) \emph{quadratic arcs}.
\end{definition}

We interpret
\[
p_{ev}=f_{(e^+,v)},
\qquad
n_{ev}=f_{(v,e^-)},
\qquad
\mu_e=f_{(e^-,e^+)}.
\]
Let $A^\uparrow$ also denote the signed node-arc incidence matrix with convention
\[
(A^\uparrow f)_u
:=
\sum_{a\in\delta^-(u)} f_a-\sum_{a\in\delta^+(u)} f_a.
\]
Define the demand vector
\[
b^\uparrow_u
:=
\begin{cases}
s_u, & u\in V,\\
0, & u=e^+\text{ or }u=e^-.
\end{cases}
\]
Then \eqref{eq:lifted-flow} can be written as
\begin{equation}
\label{eq:lifted-graph-flow}
\min_{f\in\R^{A^\uparrow}}
\left\{
\mathcal H^\uparrow(f)
:=
\sum_{e\in E}\frac{f_{(e^-,e^+)}^2}{2w_e}
\;\middle|\;
A^\uparrow f=b^\uparrow,\ f_a\ge 0\ \forall a\in A^\uparrow
\right\}.
\end{equation}
The lifted graph has
\[
m^\uparrow:=|A^\uparrow|=2P+|E|=O(P)
\]
arcs and
\[
|V^\uparrow|=|V|+2|E|
\]
vertices.

\begin{figure}[t]
  \centering
  \begin{tikzpicture}[
      >=Stealth,
      scale=0.88,
      transform shape,
      every node/.style={font=\small},
      vtx/.style={circle, draw, thick, fill=white, minimum size=6.7mm, inner sep=0pt},
      liftednode/.style={circle, draw, thick, fill=black!7, minimum size=8.0mm, inner sep=0pt},
      transport/.style={-{Stealth[length=2.0mm]}, line width=0.55pt},
      quadratic/.style={-{Stealth[length=2.4mm]}, line width=1.05pt},
      note/.style={font=\scriptsize, fill=none, inner sep=1.1pt},
      box/.style={draw, rounded corners=2pt, densely dashed, inner sep=3.5pt}
    ]

\node[liftednode] (ep) at (0, 1.95) {$e^+$};
    \node[liftednode] (em) at (0,-1.95) {$e^-$};

\node[vtx] (v1) at (-2.55,0) {$v_1$};
    \node[vtx] (v2) at (-1.25,0) {$v_2$};
    \node      (vd) at ( 0.00,0) {$\cdots$};
    \node[vtx] (vk1) at ( 1.25,0) {$v_{k-1}$};
    \node[vtx] (vk) at ( 2.55,0) {$v_k$};

    \begin{scope}[on background layer]
      \node[box, fit=(v1)(v2)(vd)(vk1)(vk)] (ebox) {};
    \end{scope}
    \node[note, anchor=east, align=right] at (-3.15, 0.42) {vertices of\\the hyperedge $e$};

\foreach \v in {v1,v2,vk1,vk} {
      \draw[transport] (ep) -- (\v);
      \draw[transport] (\v) -- (em);
    }

\draw[quadratic] (em) to[out=-18,in=18,looseness=1.15] (ep);

\node[note, anchor=east] at (-3.15, 1.04) {$p^e_v=f_{(e^+,v)}$};
    \node[note, anchor=east] at (-3.15,-1.04) {$n^e_v=f_{(v,e^-)}$};
    \node[note, anchor=west, align=left] at (3.20, 0.00) {$f_{(e^-,e^+)}=\mu_e$\\cost $\mu_e^2/(2w_e)$};

\node[note] at (0, 2.65) {$\displaystyle \sum_{v\in e} p^e_v=\mu_e$};
    \node[note] at (0,-2.82) {$\displaystyle \sum_{v\in e} n^e_v=\mu_e$};
    \node[note, anchor=west, align=left] at (1.35,-3.16) {vertex contribution:\\$(\eta_e)_v=p^e_v-n^e_v$};
  \end{tikzpicture}
  \caption{A hyperedge gadget in the lifted graph $G^\uparrow$.  For each
  hyperedge $e$, the transport arcs $e^+\to v$ and $v\to e^-$ distribute the
  positive and negative parts of the hyperedge-local dual vector, while the single
  quadratic arc $e^-\to e^+$ carries the total mass $\mu_e$ and has cost
  $\mu_e^2/(2w_e)$.}
  \label{fig:hyperedge-gadget}
\end{figure}
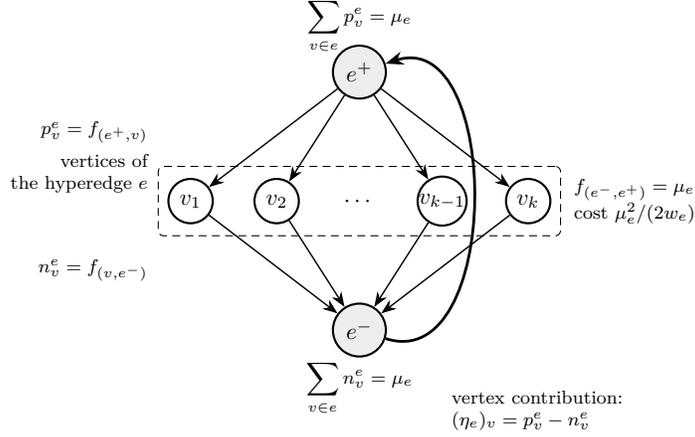
 
For each hyperedge $e$, the arcs incident to $e^+$ and $e^-$ form a small directed gadget: the two transport stars encode the positive and negative parts of the hyperedge-local dual vector, and the single quadratic arc $(e^-,e^+)$ stores the total mass $\mu_e$. Gluing these gadgets together yields an $O(P)$-arc directed graph whose global coupling appears only through incidence constraints. Figure~\ref{fig:hyperedge-gadget} shows the local structure of one gadget.

\begin{remark}[Arc support need not identify an approximate optimum]
\label{rem:approximate-flow-support}
Consider the uncapacitated lifted problem \eqref{eq:lifted-graph-flow}, which has the same optimum as the Fenchel dual. Let $f^\star$ be an optimal lifted flow, and let $c$ be a circulation with $c_a>0$ on every arc. Such a circulation is obtained by summing, over all incidences $v\in e$, one unit of flow around $e^+\to v\to e^-\to e^+$; see Lemma~\ref{lem:positive-circulation}. For every $t>0$, the flow $f^\star+tc$ remains feasible and has full arc support. By continuity of the quadratic objective,
\[
\mathcal H^\uparrow(f^\star+tc)\longrightarrow\mathcal H^\uparrow(f^\star)
\qquad\text{as }t\downarrow0.
\]
Hence even arbitrarily small objective error does not make the arc support informative. This example concerns the uncapacitated problem, not a perturbation asserted to respect the artificial capacities introduced below. It also concerns the lifted arc support, not the support of the induced signed vector $\eta_e$: positive transport flows can cancel in $p_{ev}-n_{ev}$. Objective accuracy alone supplies no test for deciding which small signed coordinates should be treated as zero in the exact recovery conditions of Remark~\ref{rem:exact-dual-recovery}. The recovery in Sections~\ref{sec:mass-support}--\ref{sec:primal-recovery} avoids support decisions and instead bounds the primal error using only the hyperedge masses.
\end{remark}

\subsection{Compatibility with the general convex-flow framework\sectionpunct}
\label{sec:convex-flow}

The goal of this subsection is to invoke Theorem~8.13 of
\cite{ChenKyngLiuPengProbstGutenbergSachdeva2025} as a black box, with all
notational conversions and all items of their Assumption~8.2 made explicit.
We use the incidence convention of
\cite{ChenKyngLiuPengProbstGutenbergSachdeva2025} only inside the following theorem:
if \(B_{\rm Ch}\in\R^{A\times V}\) has \(+1\) at the tail and \(-1\) at the
head of an arc, their conservation constraint is \(B_{\rm Ch}^{\top}f=d\).
Our lifted incidence matrix \(A^\uparrow\) uses the opposite sign convention,
namely incoming minus outgoing.  Thus, when applying the theorem of
\cite{ChenKyngLiuPengProbstGutenbergSachdeva2025} to the lifted
instance below, we set
\[
    B_{\rm Ch}^{\top}=-A^\uparrow,
    \qquad
    d_{\rm Ch}=-b^\uparrow .
\]

We also state the theorem in the open-epigraph orientation \(y>h(f)\).  This
is the orientation in which minimizing \(\sum_a y_a\) represents minimizing
\(\sum_a h_a(f_a)\), and it is the orientation used by the explicit \(p\)-norm
and entropy barriers in the applications in
\cite{ChenKyngLiuPengProbstGutenbergSachdeva2025}.  The precise relation
between this black-box statement and the theorem in
\cite{ChenKyngLiuPengProbstGutenbergSachdeva2025}, including the
epigraph orientation, incidence convention, oracle model, and the specialized
finite-precision output convention used here, is given in
Appendix~\ref{app:chen-black-box}.

\begin{theorem}[High-accuracy convex-flow black box from \cite{ChenKyngLiuPengProbstGutenbergSachdeva2025}, epigraph form]
\label{thm:chen-convex-flow}
Let $G=(V,A)$ be a directed graph with $m$ arcs and demand vector $d\in\R^V$. For every arc $a\in A$, let $h_a:\R\to\R\cup\{+\infty\}$ be convex, and let
\[
X_a:=\{(x,y):x\in\operatorname{int}(\operatorname{dom}h_a),\ y>h_a(x)\}
\]
be the corresponding open epigraph over the interior of the finite domain. Suppose that each $X_a$ is equipped with a $\nu$-self-concordant barrier $\psi_a:X_a\to\R$, and that the following six conditions hold for the oracle model and parameters of \cite{ChenKyngLiuPengProbstGutenbergSachdeva2025}, for some $K=\log^{O(1)}m$.
\begin{enumerate}[label=(\arabic*),leftmargin=1.5em]
\item One has $\widetilde O(1)$-time oracle access to $\nabla\psi_a(x,y)$ and $\nabla^2\psi_a(x,y)$.
\item The finite capacity range, demands, and finite costs are polynomially bounded: $|x|\le m^K$ whenever $h_a(x)<+\infty$, $\|d\|_\infty\le m^K$, and $|h_a(x)|\le O(m^K+|x|^K)$ on the finite part of the domain.
\item Each barrier is shifted by an additive constant so that
\[
\inf\{\psi_a(x,y):(x,y)\in X_a,\ |x|,|y|\le m^K\}=0.
\]
\item There are feasible variables $(f^{(0)},y^{(0)})$ with $B_{\rm Ch}^{\top}f^{(0)}=d$, $(f^{(0)}_a,y^{(0)}_a)\in X_a$, $|y^{(0)}_a|\le m^K$, and
\[
m^{-K}I\preceq \nabla^2\psi_a(f^{(0)}_a,y^{(0)}_a)\preceq m^K I,
\qquad
\psi_a(f^{(0)}_a,y^{(0)}_a)\le K
\]
for every arc $a$.
\item The algorithmic parameters $\alpha,\varepsilon,\kappa$ used in \cite{ChenKyngLiuPengProbstGutenbergSachdeva2025} are all smaller than $1/(1000\nu)$.
\item On the polynomial box $|x|,|y|\le m^K$, every point with $\psi_a(x,y)\le \widetilde O(1)$ satisfies
\[
\nabla^2\psi_a(x,y)\preceq \exp(\log^{O(1)}m)I .
\]
\end{enumerate}
Then, for every fixed constant $C>0$, the high-accuracy implementation from
\cite{ChenKyngLiuPengProbstGutenbergSachdeva2025} can be used as follows.  It runs in $m^{1+o(1)}$ time and, with
high probability, returns a finite string $S$ together with a query routine
$Q_S$.  There is an underlying real flow $f$ satisfying
$B_{\rm Ch}^{\top}f=d$ such that
\[
\sum_{a\in A} h_a(f_a)
\le
\min_{B_{\rm Ch}^{\top}f'=d}\sum_{a\in A} h_a(f'_a)
+
\exp(-\log^C m).
\]
For every arc $a\in A$ and every requested absolute precision $\tau$ with
$\log(1/\tau)=\log^{O(1)}m$, the query routine returns a dyadic number
$Q_S(a,\tau)=\widetilde f_a$ satisfying
\[
|\widetilde f_a-f_a|\le \tau.
\]
The work for $O(m)$ such coordinate queries at these precisions is included in
the $m^{1+o(1)}$ running time in the applications below.  Individual dyadic
coordinate values read from $S$ are approximate and need not themselves satisfy
the demand equation exactly.
\end{theorem}

We now put the lifted problem into exactly this form. The only technical addition is a redundant polynomial upper capacity on every lifted arc. This cap is useful because item~(2) of Assumption~8.2 is a bounded-capacity assumption. Fix a constant $K_{\rm cap}$, to be chosen in Lemma~\ref{lem:lift-interiorization} as a function of the input-bound exponent $K_0$, and set
\begin{equation}
\label{eq:lift-cap}
\Lambda:=P^{K_{\rm cap}} .
\end{equation}
The proof below shows that, under the assumptions of Theorem~\ref{thm:black-box-dual-solver}, this cap does not change the optimum value of the lifted problem.
Recall the arc classes \(A^\uparrow_{\mathrm{tr}}\) and \(A^\uparrow_{\mathrm{quad}}\) from Definition~\ref{def:transport-quadratic-arcs}.
For each transport arc $a\in A^\uparrow_{\mathrm{tr}}$, define the capped transport cost
\[
h_a^{\mathrm{tr}}(f)
:=
\begin{cases}
0, & 0\le f\le \Lambda,\\
+\infty, & \text{otherwise}.
\end{cases}
\]
For each quadratic arc $(e^-,e^+)\in A^\uparrow_{\mathrm{quad}}$, define
\[
h_e^{\mathrm{quad}}(f)
:=
\begin{cases}
\dfrac{f^2}{2w_e}, & 0\le f\le \Lambda,\\
+\infty, & \text{otherwise}.
\end{cases}
\]
The capped lifted convex-flow instance is
\begin{equation}
\label{eq:chen-lifted-flow}
\min_{A^\uparrow f=b^\uparrow}
\left\{
\sum_{a\in A^\uparrow_{\mathrm{tr}}} h_a^{\mathrm{tr}}(f_a)
+
\sum_{e\in E} h_e^{\mathrm{quad}}(f_{(e^-,e^+)})
\right\}.
\end{equation}

\paragraph{Transport arcs.}
For a transport arc $a=(e^+,v)$ or $a=(v,e^-)$, the open epigraph domain is
\[
X_a^{\mathrm{tr}}
=
\{(f,y):0<f<\Lambda,\ y>0\}.
\]
We use the barrier
\begin{equation}
\label{eq:transport-barrier}
\psi_a^{\mathrm{tr}}(f,y):=-\log f-\log(\Lambda-f)-\log y .
\end{equation}

\paragraph{Quadratic arcs.}
For the quadratic arc $a=(e^-,e^+)$ associated with hyperedge $e$, the open epigraph domain is
\[
X_e^{\mathrm{quad}}
=
\left\{
(f,y):0<f<\Lambda,\ y>\frac{f^2}{2w_e}
\right\}.
\]
We use the barrier
\begin{equation}
\label{eq:quadratic-barrier}
\psi_e^{\mathrm{quad}}(f,y)
:=
-\log f-\log(\Lambda-f)-2\log y
-\log\!\left(y-\frac{f^2}{2w_e}\right).
\end{equation}

\begin{lemma}[Uniform barriers for the lifted arc epigraphs]
\label{lem:lifted-arc-barriers}
For every transport arc \(a\in A^\uparrow_{\mathrm{tr}}\), the function
\[
\psi_a^{\mathrm{tr}}(f,y)
=
-\log f-\log(\Lambda-f)-\log y
\]
is a \(3\)-self-concordant barrier for
\[
X_a^{\mathrm{tr}}=\{(f,y):0<f<\Lambda,\ y>0\}.
\]
For every quadratic arc \(a=(e^-,e^+)\in A^\uparrow_{\mathrm{quad}}\), the function
\[
\psi_e^{\mathrm{quad}}(f,y)
=
-\log f-\log(\Lambda-f)-2\log y
-\log\!\left(y-\frac{f^2}{2w_e}\right)
\]
is a \(6\)-self-concordant barrier for
\[
X_e^{\mathrm{quad}}
=
\left\{(f,y):0<f<\Lambda,\ y>\frac{f^2}{2w_e}\right\}.
\]
The constants are independent of \(e\), \(w_e\), and \(\Lambda\).  Hence all lifted arcs can be handled with the uniform parameter
\[
\nu:=6.
\]
\end{lemma}

\begin{proof}
For a transport arc, the domain is cut out by the three affine inequalities
\(f>0\), \(\Lambda-f>0\), and \(y>0\).  The displayed barrier is the sum of the
three corresponding logarithmic barriers, so its self-concordance parameter is
\(3\).

For a quadratic arc, let
\[
Q_2:=\{(x,y):y>x^2\}.
\]
The \(p=2\) case of the \(p\)-norm epigraph barrier gives that
\[
\Phi_2(x,y):=-2\log y-\log(y-x^2)
\]
is a \(4\)-self-concordant barrier for \(Q_2\)
\cite[proof of Theorem~8.14]{ChenKyngLiuPengProbstGutenbergSachdeva2025};
see also Nemirovski \cite[Example~9.2.1]{Nemirovski2004IPM}.  Define the
nonsingular affine scaling
\[
T_e(f,y):=\left(\frac{f}{\sqrt{2w_e}},y\right).
\]
Since \(w_e>0\),
\[
T_e^{-1}(Q_2)
=
\left\{(f,y):y>\frac{f^2}{2w_e}\right\},
\]
and the pullback barrier is
\[
\Phi_2(T_e(f,y))
=
-2\log y-\log\!\left(y-\frac{f^2}{2w_e}\right).
\]
Self-concordance and the barrier parameter are invariant under nonsingular
affine changes of variables.  The capped quadratic epigraph is exactly
\[
X_e^{\mathrm{quad}}
=
T_e^{-1}(Q_2)\cap\{f>0\}\cap\{\Lambda-f>0\},
\]
which is
\[
\left\{(f,y):0<f<\Lambda,\ y>\frac{f^2}{2w_e}\right\}.
\]
Adding the logarithmic barriers for the two affine halfspaces gives
\[
-\log f-\log(\Lambda-f)-2\log y
-\log\!\left(y-\frac{f^2}{2w_e}\right),
\]
with parameter at most \(4+1+1=6\), uniformly over all quadratic arcs and
independently of \(w_e\) and \(\Lambda\).
\end{proof}

\paragraph{Open epigraph formulation.}
Introducing epigraph variables $y_a$, the instance supplied to Theorem~\ref{thm:chen-convex-flow} is
\begin{equation}
\label{eq:open-epigraph-lift}
\min_{f,y}
\left\{
\sum_{a\in A^\uparrow} y_a
\;\middle|\;
A^\uparrow f=b^\uparrow,
\ (f_a,y_a)\in X_a\ \forall a\in A^\uparrow
\right\},
\end{equation}
where $X_a=X_a^{\mathrm{tr}}$ for transport arcs and $X_{(e^-,e^+)}=X_e^{\mathrm{quad}}$ for quadratic arcs. In the notation of \cite{ChenKyngLiuPengProbstGutenbergSachdeva2025}, this is the same as using
\[
G=G^\uparrow,
\qquad
m=m^\uparrow,
\qquad
B_{\rm Ch}^{\top}=-A^\uparrow,
\qquad
 d_{\rm Ch}=-b^\uparrow .
\]
We emphasize that the epigraph coordinates are local epigraph variables, not
additional network-flow variables.  In the instance below, the
conserved variable is only the arc-flow vector \(f\).  The coordinates \(y_a\)
belong to the constant-dimensional local epigraph domain of arc \(a\), and no
incidence or conservation constraint is imposed on them.

\begin{lemma}[Exact compatibility with the epigraph model of \cite{ChenKyngLiuPengProbstGutenbergSachdeva2025}]
\label{lem:chen-epigraph-compatibility}
Let \(h_a=h_a^{\mathrm{tr}}\) for \(a\in A^\uparrow_{\mathrm{tr}}\), and let
\(h_{(e^-,e^+)}=h_e^{\mathrm{quad}}\) for \(e\in E\).  Then
\eqref{eq:open-epigraph-lift}, together with the domains
\(X_a^{\mathrm{tr}}\) and \(X_e^{\mathrm{quad}}\), is exactly the open-epigraph
instance of Theorem~\ref{thm:chen-convex-flow} for the arc-separable cost
\[
\sum_{a\in A^\uparrow_{\mathrm{tr}}}h_a^{\mathrm{tr}}(f_a)
+
\sum_{e\in E}h_e^{\mathrm{quad}}(f_{(e^-,e^+)})
\]
on the graph \(G^\uparrow\), with
\[
B_{\rm Ch}^{\top}=-A^\uparrow,
\qquad
d_{\rm Ch}=-b^\uparrow .
\]
In particular, the variables \(y_a\) are not flow variables and are not subject
to any conservation law.
\end{lemma}

\begin{proof}
	For a transport arc, the interior of the finite domain of
	\(h_a^{\mathrm{tr}}\) is \(0<f<\Lambda\), and the inequality
	\(y>h_a^{\mathrm{tr}}(f)\) is simply \(y>0\).  Hence the open epigraph is
	\(X_a^{\mathrm{tr}}\).  For a quadratic arc \(a=(e^-,e^+)\), the interior of
the finite domain of \(h_e^{\mathrm{quad}}\) is again \(0<f<\Lambda\), and
\(y>h_e^{\mathrm{quad}}(f)\) is precisely
	\[
	y>\frac{f^2}{2w_e},
	\]
	so the open epigraph is \(X_e^{\mathrm{quad}}\).

	The objective of the epigraph formulation is \(\sum_a y_a\).  Taking the
	infimum over strict epigraph variables therefore gives the same value as the
	closed separable-cost problem, because every finite-cost point can be
	approached from the strict epigraph by decreasing the slack to zero.  The
	only incidence conversion is the sign convention
	\(B_{\rm Ch}^{\top}=-A^\uparrow\), which turns
	\(B_{\rm Ch}^{\top}f=d_{\rm Ch}\) into \(A^\uparrow f=b^\uparrow\).
\end{proof}

The next two lemmas show that the cap is harmless and provide the interior
point used in the assumptions of \cite{ChenKyngLiuPengProbstGutenbergSachdeva2025}.  Proposition~\ref{prop:convex-flow-assumptions}
then checks the six items of Assumption~8.2 one by one, using
Lemmas~\ref{lem:chen-epigraph-compatibility} and
\ref{lem:lifted-arc-barriers}.

\begin{lemma}[A strictly positive circulation in the lift]
\label{lem:positive-circulation}
Define $c\in\R^{A^\uparrow}$ by
\[
c_a:=1
\qquad\text{for all }a\in A^\uparrow_{\mathrm{tr}},
\]
and
\[
c_{(e^-,e^+)}:=|e|
\qquad (\forall e\in E).
\]
Then $c_a>0$ for every arc $a\in A^\uparrow$, and
\[
A^\uparrow c=0.
\]
\end{lemma}

\begin{proof}
For an original vertex $v\in V$, there is one incoming transport arc $(e^+,v)$ and one outgoing transport arc $(v,e^-)$ for each incident hyperedge $e\ni v$, both carrying value $1$, so $(A^\uparrow c)_v=0$. For the lifted node $e^+$ associated with a hyperedge $e$, the incoming quadratic arc $(e^-,e^+)$ carries $|e|$, while the outgoing transport arcs $(e^+,v)$ with $v\in e$ contribute total $|e|$, so $(A^\uparrow c)_{e^+}=0$. Similarly, at the lifted node $e^-$, the incoming transport arcs $(v,e^-)$ with $v\in e$ contribute total $|e|$, which matches the outgoing quadratic arc $(e^-,e^+)$, so $(A^\uparrow c)_{e^-}=0$. Hence $A^\uparrow c=0$.
\end{proof}

\begin{lemma}[The polynomial cap and open epigraph do not change the value]
\label{lem:lift-interiorization}
Assume that $H$ is connected, $\1^\top s=0$, and
\[
\norm{s}_\infty\le P^{K_0},
\qquad
P^{-K_0}\le w_e\le P^{K_0}\qquad(\forall e\in E).
\]
If $K_{\rm cap}$ is chosen sufficiently large as a function of $K_0$, then the capped problem \eqref{eq:chen-lifted-flow}, the uncapped lifted problem \eqref{eq:lifted-graph-flow}, and the open epigraph problem \eqref{eq:open-epigraph-lift} all have the same infimum. Moreover, the capped epigraph problem has a strictly feasible point whose coordinates are polynomially bounded in $P$.
\end{lemma}

\begin{proof}
We first construct a polynomially bounded feasible flow for the uncapped closed lifted problem. Because $H$ is connected, for any two original vertices $u,v\in V$ there is a path in the lifted directed graph from $u$ to $v$: if $u$ and $v$ lie in a common hyperedge $e$, then
\[
u\to e^-\to e^+\to v
\]
is a directed path, and these $3$-arc gadgets concatenate along a hypergraph path. Fix a root vertex $r\in V$. For every $v\in V$, choose a directed path $Q_v^{\rm in}$ from $v$ to $r$ and a directed path $Q_v^{\rm out}$ from $r$ to $v$. If $s_v<0$, send $-s_v$ units along $Q_v^{\rm in}$; if $s_v>0$, send $s_v$ units along $Q_v^{\rm out}$. Since $\1^\top s=0$, the imbalance at $r$ cancels. The resulting flow $f^{\rm path}$ satisfies
\[
A^\uparrow f^{\rm path}=b^\uparrow,
\qquad
f_a^{\rm path}\ge 0,
\qquad
\|f^{\rm path}\|_\infty\le \|s\|_1\le P^{K_0+1}.
\]
Its objective value is at most $P^{O_{K_0}(1)}$.

Now take any $\delta$-optimal feasible flow $f$ for the uncapped lifted problem. Since the objective is nonnegative and $f^{\rm path}$ is feasible, we may assume
\[
\mathcal H^\uparrow(f)\le P^{O_{K_0}(1)}.
\]
For a quadratic arc corresponding to $e$, this implies
\[
f_{(e^-,e^+)}\le \sqrt{2w_e\mathcal H^\uparrow(f)}\le P^{O_{K_0}(1)}.
\]
At the nodes $e^+$ and $e^-$, conservation gives
\[
\sum_{v\in e} f_{(e^+,v)}=f_{(e^-,e^+)},
\qquad
\sum_{v\in e} f_{(v,e^-)}=f_{(e^-,e^+)},
\]
so every transport arc flow is also at most $P^{O_{K_0}(1)}$. Choosing $K_{\rm cap}$ larger than the exponent in this bound shows that every $\delta$-optimal solution of the uncapped problem can be taken inside the cap $0\le f_a\le\Lambda$. Letting $\delta\downarrow0$, the capped and uncapped closed problems have the same infimum.

It remains to compare the capped closed problem to the open epigraph formulation. Let $c$ be the positive circulation from Lemma~\ref{lem:positive-circulation}. For the above choice of $K_{\rm cap}$,
\begin{equation}
\label{eq:lift-strict-start-flow}
f^{\rm int}:=f^{\rm path}+c
\end{equation}
By increasing $K_{\rm cap}$ if necessary, we may also ensure
\[
\Lambda-f_a^{\rm int}\ge 1
\qquad(\forall a\in A^\uparrow).
\]
Thus $A^\uparrow f^{\rm int}=b^\uparrow$ and all denominators
$f_a^{\rm int}$ and $\Lambda-f_a^{\rm int}$ are bounded below by
$P^{-O_{K_0}(1)}$; in particular $0<f_a^{\rm int}<\Lambda$ for every arc.
Choose
\begin{equation}
\label{eq:lift-strict-start-epi}
y_a^{\rm int}:=1 \quad(a\in A^\uparrow_{\mathrm{tr}}),
\qquad
 y_{(e^-,e^+)}^{\rm int}:=\frac{(f^{\rm int}_{(e^-,e^+)})^2}{2w_e}+1.
\end{equation}
Then $(f^{\rm int},y^{\rm int})$ is strictly feasible for \eqref{eq:open-epigraph-lift}, and all its coordinates are $P^{O_{K_0}(1)}$.

For any feasible $f$ of the capped closed problem and any $\theta\in(0,1)$, the convex combination
\[
f^{(\theta)}:=(1-\theta)f+\theta f^{\rm int}
\]
satisfies $A^\uparrow f^{(\theta)}=b^\uparrow$ and $0<f^{(\theta)}_a<\Lambda$ for all arcs. By choosing epigraph variables $y^{(\theta)}$ with arbitrarily small positive slack above the arc costs, one obtains feasible points of \eqref{eq:open-epigraph-lift} whose objectives converge to the capped closed objective of $f$ as $\theta\downarrow0$. This proves that the open epigraph infimum is at most the capped closed optimum. The reverse inequality is immediate because every open-epigraph feasible point has $0<f_a<\Lambda$ and $y_a>h_a(f_a)$, hence its objective is at least the capped closed objective of the same flow. Thus the two infima are equal.
\end{proof}

\begin{proposition}[The lifted instance satisfies Assumption~8.2 of \cite{ChenKyngLiuPengProbstGutenbergSachdeva2025}]
\label{prop:convex-flow-assumptions}
Assume that $H$ is connected, that $\1^\top s=0$, and that for some constant $K_0\ge 1$,
\[
\norm{s}_\infty\le P^{K_0},
\qquad
P^{-K_0}\le w_e\le P^{K_0}\qquad (\forall e\in E).
\]
Choose the cap $\Lambda=P^{K_{\rm cap}}$ as in
Lemma~\ref{lem:lift-interiorization}. Then the following hold:
\[
P\le m^\uparrow=2P+|E|\le 3P.
\]
By Lemma~\ref{lem:chen-epigraph-compatibility},
\eqref{eq:open-epigraph-lift} is a direct instance of
Theorem~\ref{thm:chen-convex-flow}. By Lemma~\ref{lem:lifted-arc-barriers},
the uniform barrier parameter is $\nu=6$. The capped extended-real lifted
instance \eqref{eq:chen-lifted-flow}, together with the open epigraph domains
and barriers above, satisfies all six items of Assumption~8.2 of \cite{ChenKyngLiuPengProbstGutenbergSachdeva2025}.
Therefore Theorem~\ref{thm:chen-convex-flow} applies to the lifted instance.
\end{proposition}

\begin{proof}
		Let $M:=m^\uparrow$. Since $P\le M\le3P$, polynomial and quasipolynomial bounds in $P$ and in $M$ are interchangeable. The verification below checks Assumption~8.2 of \cite{ChenKyngLiuPengProbstGutenbergSachdeva2025} for the lifted instance. The translation from the theorem in \cite{ChenKyngLiuPengProbstGutenbergSachdeva2025} to the epigraph form used here, including notation and finite-precision conventions, is given in Appendix~\ref{app:chen-black-box}. We verify the six items in the order used in Theorem~\ref{thm:chen-convex-flow}.
		Lemma~\ref{lem:chen-epigraph-compatibility} identifies the exact
		instance, and Lemma~\ref{lem:lifted-arc-barriers} gives the uniform barrier
		parameter \(\nu=6\).  It remains only to verify the six quantitative items of
		Assumption~8.2, summarized in Table~\ref{tab:convex-flow-assumptions}.

\begin{table}[tbp]
\caption{Verification of the six conditions in Assumption~8.2 of \cite{ChenKyngLiuPengProbstGutenbergSachdeva2025}. Items (1)--(6) in the proof give the quantitative details.}
\label{tab:convex-flow-assumptions}
\centering
\begin{tabularx}{\linewidth}{@{}p{0.29\linewidth}X@{}}
\toprule
Condition & Verification and references \\
\midrule
(1) Oracle access & Explicit gradients and Hessians in \eqref{eq:transport-oracles} and \eqref{eq:quadratic-oracles}; evaluation to the required precision. \\
(2) Bounded input data & Polynomial capacity, demand, and cost bounds; see \eqref{eq:lift-cap} and item (2). \\
(3) Barrier normalization & Subtract a finite constant without changing derivatives or self-concordance; see item (3). \\
(4) Strict feasible start & Lemmas~\ref{lem:positive-circulation} and \ref{lem:lift-interiorization} construct the start; item (4) bounds its barrier value and Hessian. \\
(5) Algorithmic parameters & Barrier parameter $\nu=6$ and $\alpha,\varepsilon,\kappa<1/6000$; see Lemma~\ref{lem:lifted-arc-barriers} and item (5). \\
(6) Hessian bound & Bounded barrier level sets keep all logarithmic denominators away from zero; see item (6). \\
\bottomrule
\end{tabularx}
\end{table}

	\smallskip
	\noindent
		\textbf{(1) Oracle access.}
		For transport arcs,
		\begin{equation}
		\label{eq:transport-oracles}
		\nabla\psi_a^{\mathrm{tr}}(f,y)
		=
		\begin{pmatrix}
		-f^{-1}+(\Lambda-f)^{-1}\\
		-y^{-1}
		\end{pmatrix},
		\qquad
		\nabla^2\psi_a^{\mathrm{tr}}(f,y)
		=
		\begin{pmatrix}
		f^{-2}+(\Lambda-f)^{-2} & 0\\
		0 & y^{-2}
	\end{pmatrix}.
	\end{equation}
	For quadratic arcs, abbreviate
	\[
	g:=y-\frac{f^2}{2w_e}.
	\]
	Then
	\begin{equation}
	\label{eq:quadratic-oracles}
	\begin{aligned}
		\nabla\psi_e^{\mathrm{quad}}(f,y)
		&=
		\begin{pmatrix}
		-f^{-1}+(\Lambda-f)^{-1}+\dfrac{f}{w_eg}\\[0.4em]
		-2y^{-1}-g^{-1}
		\end{pmatrix},
		\\
		\nabla^2\psi_e^{\mathrm{quad}}(f,y)
		&=
		\begin{pmatrix}
		f^{-2}+(\Lambda-f)^{-2}+(w_eg)^{-1}+f^2/(w_e^2g^2) & -f/(w_eg^2)\\[0.4em]
		-f/(w_eg^2) & 2y^{-2}+g^{-2}
	\end{pmatrix}.
	\end{aligned}
	\end{equation}
	These are constant-dimensional rational formulas, so the required gradient and Hessian oracle calls take $\widetilde O(1)$ time. On the level sets used in item~(6) below, all denominators are at least $\exp(-\log^{O(1)}M)$; hence evaluating these formulas with $\log^{O(1)}M$ bits of precision gives the numerical accuracy required by the oracle model of \cite{ChenKyngLiuPengProbstGutenbergSachdeva2025}. Additive barrier shifts used in item~(3) do not affect these oracle outputs.

	\smallskip
	\noindent
		\textbf{(2) Polynomially bounded capacities, demands, and costs.}
	The cap \eqref{eq:lift-cap} gives $|f_a|\le\Lambda=P^{K_{\rm cap}}\le M^{O_{K_0}(1)}$ on the finite domain of every arc cost. In the notation of \cite{ChenKyngLiuPengProbstGutenbergSachdeva2025}, the demand is $d_{\rm Ch}=-b^\uparrow$, so
	\[
	\|d_{\rm Ch}\|_\infty=\|b^\uparrow\|_\infty=\|s\|_\infty\le P^{K_0}\le M^{O_{K_0}(1)}.
	\]
Transport costs are identically zero on their finite domain. Quadratic costs satisfy
\[
0\le h_e^{\mathrm{quad}}(f)=\frac{f^2}{2w_e}
\le \frac12 P^{K_0}\Lambda^2
\le M^{O_{K_0}(1)}
\]
whenever they are finite. This verifies the polynomial boundedness condition.

	\smallskip
	\noindent
		\textbf{(3) Barrier normalization.}
	For each arc, define
	\[
\gamma_a
:=
\inf\{\psi_a(f,y):(f,y)\in X_a,\ |f|,|y|\le M^K\},
\qquad
\widehat\psi_a:=\psi_a-\gamma_a,
\]
Choose $K=\log^{O(1)}M$ large enough that the polynomial box contains the cap
and the starting point from Lemma~\ref{lem:lift-interiorization}.  On this box,
all logarithmic arguments that can appear in the barriers are at most
$M^{\log^{O(1)}M}$, while the starting point has all logarithmic denominators
bounded below by $M^{-O_{K_0}(1)}$.  Therefore
\[
|\gamma_a|\le \log^{O(1)}M .
\]
The shift from $\psi_a$ to $\widehat\psi_a$ does not change
self-concordance, gradients, Hessians, or oracle complexity, and gives exactly
the normalization required by item~(3) of Assumption~8.2.

	\smallskip
	\noindent
		\textbf{(4) Strictly feasible start with controlled Hessian.}
	Use the strictly feasible point $(f^{\rm int},y^{\rm int})$ constructed in Lemma~\ref{lem:lift-interiorization}, and set
\[
f^{(0)}:=f^{\rm int},
\qquad
 y^{(0)}:=y^{\rm int}.
	\]
	All coordinates are bounded by $M^{O_{K_0}(1)}$, and $B_{\rm Ch}^{\top}f^{(0)}=-A^\uparrow f^{(0)}=-b^\uparrow=d_{\rm Ch}$.

	At this point, every transport denominator $f^{(0)}$, $\Lambda-f^{(0)}$, and $y^{(0)}$ lies between $M^{-O_{K_0}(1)}$ and $M^{O_{K_0}(1)}$. Substituting these bounds into \eqref{eq:transport-oracles} gives
	\[
	M^{-O_{K_0}(1)}I
	\preceq
	\nabla^2\psi_a^{\mathrm{tr}}(f_a^{(0)},y_a^{(0)})
	\preceq
	M^{O_{K_0}(1)}I
	\]
	for every transport arc.

	For a quadratic arc, we have
	\[
	g_e^{(0)}:=y_e^{(0)}-\frac{(f_{(e^-,e^+)}^{(0)})^2}{2w_e}=1,
	\]
	and again all quantities $f^{(0)}$, $\Lambda-f^{(0)}$, $y^{(0)}$, $g_e^{(0)}$, $w_e$, and $w_e^{-1}$ are polynomially bounded in $M$. The explicit Hessian formula \eqref{eq:quadratic-oracles} gives the upper bound $\nabla^2\psi_e^{\mathrm{quad}}(f^{(0)},y^{(0)})\preceq M^{O_{K_0}(1)}I$. For the lower bound, dropping the nonnegative cap term from the $(1,1)$ entry only decreases the determinant, and the remaining two-dimensional Hessian satisfies
	\[
	\det \nabla^2 \psi_e^{\mathrm{quad}}(f,y)
	\ge
\Bigl(f^{-2}+(w_eg)^{-1}\Bigr)\Bigl(2y^{-2}+g^{-2}\Bigr)
+
\frac{2f^2}{w_e^2g^2y^2}
\ge
\frac{2}{f^2y^2}.
\]
At the starting point, this determinant is at least $M^{-O_{K_0}(1)}$, while the trace is at most $M^{O_{K_0}(1)}$. Thus the smallest eigenvalue is at least $M^{-O_{K_0}(1)}$. Therefore, for all arcs,
	\[
	M^{-O_{K_0}(1)}I
	\preceq
	\nabla^2\psi_a(f_a^{(0)},y_a^{(0)})
	\preceq
	M^{O_{K_0}(1)}I.
	\]
	Finally, because the unshifted logarithmic barriers have value at most $\log^{O(1)}M$ at $(f^{(0)},y^{(0)})$ and the normalizing constants satisfy $|\gamma_a|\le\log^{O(1)}M$ on the polynomial box, the normalized values obey
\[
\widehat\psi_a(f_a^{(0)},y_a^{(0)})\le\log^{O(1)}M.
\]
This is item~(4) with the parameter $K=\log^{O(1)}M$ from \cite{ChenKyngLiuPengProbstGutenbergSachdeva2025}.

\smallskip
\noindent
	\textbf{(5) Small algorithmic parameters.}
The uniform barrier parameter is $\nu=6$. We use the parameter choices of \cite{ChenKyngLiuPengProbstGutenbergSachdeva2025} with
\[
\alpha,\varepsilon,\kappa<\frac{1}{1000\nu}=\frac{1}{6000},
\]
which is exactly item~(5).

	\smallskip
	\noindent
		\textbf{(6) Quasipolynomial Hessian bounds on bounded barrier level sets.}
	Consider first a transport arc. Suppose $|f|,|y|\le M^K$ and
	$\widehat\psi_a^{\mathrm{tr}}(f,y)\le\widetilde O(1)$. Since
	$|\gamma_a|\le\log^{O(1)}M$, the unshifted barrier value is also at most
	$\log^{O(1)}M$. On the polynomial box, no positive logarithmic argument can
	exceed $M^{\log^{O(1)}M}$. Hence no term $-\log f$, $-\log(\Lambda-f)$, or
	$-\log y$ can be larger than $\log^{O(1)}M$, and therefore
\[
f,\ \Lambda-f,\ y\ge\exp(-\log^{O(1)}M).
\]
Substitution into the transport Hessian gives
	\[
	\nabla^2\psi_a^{\mathrm{tr}}(f,y)
	\preceq
	\exp(\log^{O(1)}M)I.
	\]

	For a quadratic arc, the same argument applied to the four logarithmic terms in
	$\psi_e^{\mathrm{quad}}$ gives
	\[
	f,\ \Lambda-f,\ y,
	\quad
	g:=y-\frac{f^2}{2w_e}
	\ge
	\exp(-\log^{O(1)}M).
	\]
	Together with $|f|,|y|\le M^K$ and
	$w_e,w_e^{-1}\le M^{O_{K_0}(1)}$, the displayed quadratic Hessian formula
	implies
	\[
	\nabla^2\psi_e^{\mathrm{quad}}(f,y)
	\preceq
	\exp(\log^{O(1)}M)I.
	\]
	This is the quasipolynomial Hessian bound in item~(6).

	All six items of Assumption~8.2 have now been matched to explicit formulas, bounds, and constructions for the lifted instance, so Theorem~\ref{thm:chen-convex-flow} applies.
\end{proof}

\subsection{Black-box theorem for the solver's first stage and consequences\sectionpunct}
\label{sec:main-theorem}

Having verified the required hypotheses, we can now state the guarantee for the first stage of the solver, which computes a high-accuracy lifted dual flow before primal recovery.

\begin{definition}[Oracle representation for the first stage]
\label{def:first-stage-coordinate-oracle-representation}
This specializes the output model in Theorem~\ref{thm:chen-convex-flow} to the
lifted graph \(G^\uparrow\). An oracle representation for the first stage
consists of a finite string together with a query routine. It represents a real
lifted flow $f\in\R^{A^\uparrow}$ if, for every arc $a\in A^\uparrow$ and every
requested absolute precision $\tau$ with $\log(1/\tau)=\log^{O(1)}P$, the query
routine returns a dyadic number $\widetilde f_a$ satisfying
\[
|\widetilde f_a-f_a|\le \tau.
\]
When the underlying flow $f$ defined by the oracle is used to define
\[
\mu_e:=f_{(e^-,e^+)},
\qquad
(\eta_e)_v:=f_{(e^+,v)}-f_{(v,e^-)}
\quad (v\in e),
\]
the same oracle gives coordinate access to the vectors $\mu$ and $\eta$. The feasibility statements $A^\uparrow f=b^\uparrow$ and $B\eta=s$ are statements about the underlying real vectors, not about the raw dyadic outputs returned by finite-precision queries. Whenever the final algorithm needs a finite dual certificate, it explicitly materializes and repairs the queried coordinates as in Lemma~\ref{lem:dual-certificate-repair}.
\end{definition}

The query precisions used below are all of the form $\exp(-\log^{O(1)}P)$, and only $O(P)$ coordinates are queried. These queries are included in the stated $P^{1+o(1)}$ running time. Under the dyadic input model of Definition~\ref{def:dyadic-input}, these oracle evaluations are implemented with $\log^{O(1)}P$-bit dyadic arithmetic.

For later use, write
\[
q(z):=\frac12\sum_{e\in E}\frac{z_e^2}{w_e}.
\]
Call $\mu\in\R_{\ge0}^E$ a feasible hyperedge-mass vector for demand $s$ if it is the $\mu$-component of a feasible triple in \eqref{eq:lifted-flow}; denote the set of all such vectors by $\mathcal M(s)$.

\begin{theorem}[High-accuracy first stage]
\label{thm:black-box-dual-solver}
Assume that $H$ is connected, that $\1^\top s=0$, and that
\[
\norm{s}_\infty\le P^{K_0},
\qquad
P^{-K_0}\le w_e\le P^{K_0}\qquad (\forall e\in E)
\]
for some constant $K_0\ge 1$. Then for every fixed constant $C>0$, there is a randomized algorithm that runs in
\[
P^{1+o(1)}
\]
time and, with high probability, returns an oracle representation for a real feasible lifted flow $f$ on $G^\uparrow$. The mass vector $\mu$ on the quadratic arcs and the induced dual vector $\eta$, accessible through the same oracle, are defined by
\[
\mu_e:=f_{(e^-,e^+)}
\qquad (e\in E),
\]
and
\[
(\eta_e)_v:=f_{(e^+,v)}-f_{(v,e^-)}
\qquad (v\in e),
\]
with $\eta_e$ zero-padded outside $e$. Then the underlying vectors satisfy $A^\uparrow f=b^\uparrow$, $B\eta=s$, $\mu\in\mathcal M(s)$, and
\[
\mathcal D(\eta)
\le
q(\mu)=\mathcal H^\uparrow(f)
\le
\mathcal D^\star+
\exp(-\log^C P).
\]
\end{theorem}

\begin{proof}
By Lemmas~\ref{lem:chen-epigraph-compatibility} and
\ref{lem:lifted-arc-barriers}, and by
Proposition~\ref{prop:convex-flow-assumptions}, the capped lifted instance
satisfies the hypotheses of Theorem~\ref{thm:chen-convex-flow} with
$G=G^\uparrow$, $m=m^\uparrow$,
$B_{\rm Ch}^{\top}=-A^\uparrow$, and $d_{\rm Ch}=-b^\uparrow$. Applying
Theorem~\ref{thm:chen-convex-flow} in the finite-precision oracle model of
\cite{ChenKyngLiuPengProbstGutenbergSachdeva2025}, we obtain, with high probability, a finite string and query routine.
Together they form an oracle representation in the sense
of Definition~\ref{def:first-stage-coordinate-oracle-representation} for an
underlying real flow $f$ on the lifted graph. The raw dyadic answers returned by
coordinate queries need not satisfy the lifted demand equation exactly; the
feasibility statement is about the underlying real flow. The call, including
the coordinate-query budget used below, takes time
\[
(m^\uparrow)^{1+o(1)}=P^{1+o(1)}.
\]
The underlying real flow satisfies
\[
A^\uparrow f=b^\uparrow
\]
and
\[
\sum_{a\in A^\uparrow_{\mathrm{tr}}} h_a^{\mathrm{tr}}(f_a)
+
\sum_{e\in E} h_e^{\mathrm{quad}}(f_{(e^-,e^+)})
\le
\min_{A^\uparrow f'=b^\uparrow}
\left\{
\sum_{a\in A^\uparrow_{\mathrm{tr}}} h_a^{\mathrm{tr}}(f'_a)
+
\sum_{e\in E} h_e^{\mathrm{quad}}(f'_{(e^-,e^+)})
\right\}
+
\exp(-\log^C m^\uparrow).
\]
All remaining equalities and inequalities in the proof are statements about this underlying real flow and its induced vectors, not about the raw dyadic outputs returned by finite-precision queries.
Since $m^\uparrow\ge P$, the additive term is at most $\exp(-\log^C P)$. Because the left-hand side is finite, every arc flow of $f$ satisfies the cap and is nonnegative, and by the definition of the arc costs the left-hand side equals $\mathcal H^\uparrow(f)$. By Lemma~\ref{lem:lift-interiorization}, the cap does not change the optimum value of the original lifted problem \eqref{eq:lifted-graph-flow}. Therefore
\[
\mathcal H^\uparrow(f)
\le
\min_{A^\uparrow f'=b^\uparrow,\;f'\ge0}\mathcal H^\uparrow(f')
+
\exp(-\log^C P).
\]
Define
\[
\mu_e:=f_{(e^-,e^+)}
\qquad (e\in E).
\]
Because the transport arcs have zero cost,
\[
q(\mu)=\sum_{e\in E}\frac{\mu_e^2}{2w_e}=\mathcal H^\uparrow(f).
\]
Writing
\[
p_{ev}:=f_{(e^+,v)},
\qquad
n_{ev}:=f_{(v,e^-)},
\]
the conservation constraints $A^\uparrow f=b^\uparrow$ give
\[
\sum_{e\ni v}(p_{ev}-n_{ev})=s_v,
\qquad
\sum_{v\in e}p_{ev}=\mu_e,
\qquad
\sum_{v\in e}n_{ev}=\mu_e.
\]
These are exactly the defining constraints of $\mathcal M(s)$, so $\mu\in\mathcal M(s)$. Map this underlying flow back to a hypergraph dual point by
\[
(\eta_e)_v:=f_{(e^+,v)}-f_{(v,e^-)}
\qquad (v\in e),
\]
and zero-pad outside $e$. By Proposition~\ref{prop:lifted-flow-equivalence}, this $\eta$ is feasible for \eqref{eq:dual-problem} and satisfies
\[
\mathcal D(\eta)\le \mathcal H^\uparrow(f)=q(\mu).
\]
Taking minima on the lifted side and using the value equivalence from Proposition~\ref{prop:lifted-flow-equivalence} yields
\[
q(\mu)=\mathcal H^\uparrow(f)\le \mathcal D^\star+\exp(-\log^C P).
\]
The probability of success is the high-probability guarantee inherited from Theorem~\ref{thm:chen-convex-flow}.
\end{proof}

\begin{corollary}[Additive estimate of the primal optimum value]
\label{cor:primal-value}
Under the assumptions of Theorem~\ref{thm:black-box-dual-solver}, on the same high-probability event, the quantity
\[
-q(\mu)
\]
associated with the returned mass vector satisfies
\[
\mathrm{OPT}-\exp(-\log^C P)
\le
-q(\mu)
\le
\mathrm{OPT}.
\]
\end{corollary}

\begin{proof}
By Theorem~\ref{thm:poisson-fenchel-dual},
\[
\mathrm{OPT}=-\mathcal D^\star.
\]
On the success event of Theorem~\ref{thm:black-box-dual-solver}, the induced vector $\eta$ is feasible for the dual and satisfies
\[
\mathcal D^\star\le \mathcal D(\eta)\le q(\mu).
\]
Combining this with Theorem~\ref{thm:black-box-dual-solver} gives
\[
\mathcal D^\star
\le
q(\mu)
\le
\mathcal D^\star+
\exp(-\log^C P),
\]
which is equivalent to the stated bound after multiplying by $-1$.
\end{proof}

Thus, with high probability, the first stage yields a $P^{1+o(1)}$ algorithm for the dual together with a feasible mass vector and an additive estimate of the primal optimum value. Section~\ref{sec:primal-recovery} uses the returned mass coordinates to recover an explicit primal point in the same asymptotic running time.
 \section{Feasible Hyperedge Masses and Support Queries\sectionpunct}
\label{sec:mass-support}

This section carries out Steps~3--4 of the roadmap. Its input is the lifted formulation and the mass vector returned in Step~2. We identify the feasible mass set $\mathcal M(s)$ and reduce the support query to a min-cost-flow problem with linear costs on the same lifted graph $G^\uparrow$. The output is a characterization of $\mathcal M(s)$, a reduction for support queries, and the rule that turns an LP-dual node potential $\pi$ for the support flow into a support maximizer on $V$.

Crucially, the detailed transport variables $p,n$ are not used in primal recovery except through the scalar masses $\mu$: the first-stage routing inside each hyperedge gadget is discarded, and support queries are driven only by budgets derived from $\mu$.

For $z\in\R^E$, recall the quadratic mass notation
\[
\norm{z}_{w^{-1}}^2:=\sum_{e\in E}\frac{z_e^2}{w_e},
\qquad
q(z):=\frac12\norm{z}_{w^{-1}}^2.
\]

Explicitly, the feasible hyperedge-mass set is
\begin{equation}
\label{eq:mass-polytope}
\mathcal M(s):=
\left\{
\mu\in\R_{\ge 0}^E:
\begin{array}{l}
\exists\,p,n\ge 0\ \text{such that}
\sum_{e\ni v}(p_{ev}-n_{ev})=s_v \quad (\forall v\in V),\\[0.3em]
\sum_{v\in e}p_{ev}=\mu_e \quad (\forall e\in E),\\[0.3em]
\sum_{v\in e}n_{ev}=\mu_e \quad (\forall e\in E)
\end{array}
\right\}.
\end{equation}

On the primal side, we will need the following support query: under prescribed hyperedge-range budgets $r$, how large can the response $\ip{s}{x}$ be?
For $r\in\R_{\ge 0}^E$, define
\begin{equation}
\label{eq:support-primal}
L_s(r):=
\max\bigl\{\ip{s}{x}:\ x\in\Xzero,\ R_e(x)\le r_e\ \forall e\in E\bigr\}.
\end{equation}

The next theorem states the facts about the mass set that we will use later. First, it rewrites the quadratic dual objective in terms of $\mathcal M(s)$. Second, it turns the support queries used in primal recovery into linear optimization over $\mathcal M(s)$, and equivalently into min-cost flow on the lifted graph. Third, it shows how an optimal LP-dual node potential of that min-cost flow recovers an optimal support maximizer.
\begin{theorem}[Feasible hyperedge masses and support queries]
\label{thm:mass-support}
Let $\mathcal M(s)$ be the set of feasible hyperedge-mass vectors defined in \eqref{eq:mass-polytope}. Let $b^\uparrow$ be the lifted demand vector on $G^\uparrow$ defined by
\[
b^\uparrow_v=s_v\quad (v\in V),
\qquad
b^\uparrow_{e^+}=b^\uparrow_{e^-}=0\quad (e\in E).
\]
Then
\[
\mathcal D^\star=\min_{\mu\in\mathcal M(s)} q(\mu).
\]
For every $r\in\R_{\ge 0}^E$, define the linear support costs on the lifted arcs by
\[
c^r_a:=
\begin{cases}
r_e, & a=(e^-,e^+)\text{ for some }e\in E,\\
0, & \text{otherwise}.
\end{cases}
\]
Then
\[
L_s(r)=\min_{\mu\in\mathcal M(s)}\sum_{e\in E} r_e\mu_e
=\min\Bigl\{\sum_{a\in A^\uparrow} c^r_a f_a: A^\uparrow f=b^\uparrow,\ f\ge 0\Bigr\}.
\]
Moreover, if $\pi\in\R^{V^\uparrow}$ is any optimal LP-dual node potential for this min-cost-flow LP, meaning
\[
\bigl((A^\uparrow)^\top\pi\bigr)_a\le c^r_a\quad(a\in A^\uparrow),
\qquad
\ip{b^\uparrow}{\pi}=L_s(r),
\]
then the vector
\[
x_v:=\pi_v-\frac{\sum_{u\in V} d_u\pi_u}{\sum_{u\in V} d_u}
\qquad (v\in V)
\]
is an optimal solution of the support maximization problem \eqref{eq:support-primal}.
\end{theorem}

\begin{proof}
The theorem is proved by the next two propositions. Proposition~\ref{prop:mass-projection} establishes the identity $\mathcal D^\star=\min_{\mu\in\mathcal M(s)} q(\mu)$. Proposition~\ref{prop:support-lift} proves the support formula, its lifted min-cost-flow realization, and the recovery of an optimal support maximizer from an optimal LP-dual node potential.
\end{proof}

Figure~\ref{fig:mass-support-recovery} illustrates the role of this theorem in the recovery stage: the quadratic minimization over $\mathcal M(s)$ supplies the normal vector $r=\mu/w$. For an exact minimizer this already identifies an optimal solution of the linearized problem over $\mathcal M(s)$, but it does not give the vertex potential $x$. We therefore solve a min-cost-flow problem with linear costs on the same lifted graph $G^\uparrow$ to obtain an optimal LP-dual node potential, whose restriction to $V$ is the recovered primal point.

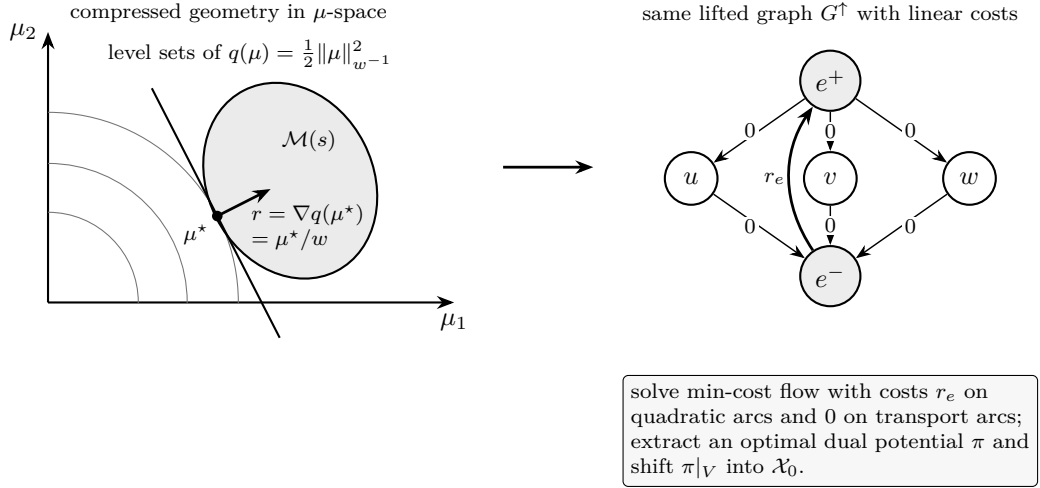
\begin{figure}[t]
\centering
\begin{tikzpicture}[
    >=Stealth,
    x=1cm,
    y=1cm,
    every node/.style={font=\small},
    note/.style={font=\scriptsize, fill=none, inner sep=1.1pt},
    vertex/.style={circle, draw, thick, fill=white, minimum size=7mm, inner sep=0pt},
    auxnode/.style={circle, draw, thick, fill=black!7, minimum size=8mm, inner sep=0pt},
    auxarrow/.style={-{Stealth[length=2.0mm]}, line width=0.55pt},
    mainarrow/.style={-{Stealth[length=2.4mm]}, line width=1.05pt},
    costlabel/.style={font=\scriptsize, fill=white, inner sep=1pt},
    tinybox/.style={draw, rounded corners=2pt, fill=black!3, font=\scriptsize, align=left, inner sep=3pt}
  ]
\begin{scope}[shift={(0,0)}]
    \node[note] at (2.6,4.25) {compressed geometry in $\mu$-space};
    \coordinate (origin) at (0.2,0.4);
    \draw[->, thick] (origin) -- (5.6,0.4) node[below] {$\mu_1$};
    \draw[->, thick] (origin) -- (0.2,4.0) node[left] {$\mu_2$};

    \begin{scope}
      \clip (0.2,0.4) rectangle (5.6,4.0);
      \draw[thin, black!55] (origin) circle [radius=1.20];
      \draw[thin, black!55] (origin) circle [radius=1.85];
      \draw[thin, black!55] (origin) circle [radius=2.53];
    \end{scope}
    \node[note, anchor=west] at (0.95,3.70) {level sets of $q(\mu)=\tfrac12\|\mu\|_{w^{-1}}^2$};

    \coordinate (mustar) at (2.45,1.55);
    \begin{scope}[shift={(3.42,2.02)}, rotate=-62.9]
      \draw[thick, fill=black!8] (0,0) ellipse [x radius=1.35, y radius=1.10];
    \end{scope}
    \node[note] at (3.67,2.55) {$\mathcal M(s)$};

    \draw[thick] (1.58,3.25) -- (3.28,-0.07);
    \fill (mustar) circle (2pt);
    \node[note, anchor=north east] at (2.38,1.48) {$\mu^{\star}$};

    \draw[mainarrow] (mustar) -- (3.16,1.91);
    \node[note, anchor=west, align=left] at (2.86,1.4) {$r=\nabla q(\mu^{\star})$\\$=\mu^{\star}/w$};
  \end{scope}

\draw[mainarrow] (6.25,2.2) -- (7.45,2.2);

\begin{scope}[shift={(7.9,0.0)}]
    \node[note] at (2.7,4.25) {same lifted graph $G^{\uparrow}$ with linear costs};

    \node[auxnode] (ep) at (2.7,3.35) {$e^+$};
    \node[auxnode] (em) at (2.7,0.75) {$e^-$};
    \node[vertex] (a) at (0.85,2.05) {$u$};
    \node[vertex] (b) at (2.7,2.05) {$v$};
    \node[vertex] (c) at (4.55,2.05) {$w$};

    \draw[auxarrow] (ep) -- node[costlabel, left] {$0$} (a);
    \draw[auxarrow] (ep) -- node[costlabel] {$0$} (b);
    \draw[auxarrow] (ep) -- node[costlabel, right] {$0$} (c);
    \draw[auxarrow] (a) -- node[costlabel, left] {$0$} (em);
    \draw[auxarrow] (b) -- node[costlabel] {$0$} (em);
    \draw[auxarrow] (c) -- node[costlabel, right] {$0$} (em);
    \draw[mainarrow] (em) to[bend left=34] node[note, left] {$r_e$} (ep);

    \node[tinybox, below=9mm of em, text width=53mm] (potbox) {solve min-cost flow with costs $r_e$ on quadratic arcs and $0$ on transport arcs; extract an optimal LP-dual node potential $\pi$ and shift $\pi|_V$ into $\Xzero$.};
  \end{scope}
\end{tikzpicture}
\caption{From the feasible mass set to a primal certificate. Left: the first stage minimizes the quadratic objective $q(\mu)$ over the convex feasible mass set $\mathcal M(s)$. At an exact minimizer, the gradient $r=\mu^{\star}/w$ defines a supporting hyperplane of $\mathcal M(s)$, so the same normal vector makes $\mu^\star$ optimal for the linearized mass problem. This is still only mass-side information. Right: to recover a vertex potential, we solve the equivalent linear-cost min-cost-flow problem on the same lifted graph $G^\uparrow$; its LP-dual node potential restricted to the original vertex set yields the support maximizer and hence the recovered primal potential.}
\label{fig:mass-support-recovery}
\end{figure}
 
\subsection{Feasible hyperedge-mass vectors\sectionpunct}

We begin with the first part of Theorem~\ref{thm:mass-support}, namely the reduction of the quadratic dual objective to the feasible mass set.
\begin{proposition}[Feasible hyperedge masses from the lifted flow formulation of the dual]
\label{prop:mass-projection}
The set \eqref{eq:mass-polytope} records the hyperedge masses that can arise in the lifted flow formulation of the dual, in the sense that
\[
\mathcal D^\star=\min_{\mu\in\mathcal M(s)} q(\mu).
\]
More precisely, if $\eta$ is feasible for the dual, then
\[
\mu^\eta_e:=\frac12\norm{\eta_e}_1
\]
belongs to $\mathcal M(s)$ and satisfies
\[
q(\mu^\eta)=\mathcal D(\eta).
\]
Conversely, every $\mu\in\mathcal M(s)$ gives a feasible dual point $\eta$ with
\[
\mathcal D(\eta)\le q(\mu).
\]
\end{proposition}

\begin{proof}
If $\mu\in\mathcal M(s)$, choose witnesses $p,n$ in \eqref{eq:mass-polytope} and define $\eta$ from $(p,n,\mu)$ as in Proposition~\ref{prop:lifted-flow-equivalence}(i). Then $\eta$ is feasible for the dual and
\[
\mathcal D(\eta)\le \sum_{e\in E}\frac{\mu_e^2}{2w_e}=q(\mu).
\]
Taking minima over $\mu\in\mathcal M(s)$ yields
\[
\mathcal D^\star\le \min_{\mu\in\mathcal M(s)} q(\mu).
\]
Conversely, if $\eta$ is feasible for the dual, Proposition~\ref{prop:lifted-flow-equivalence}(ii) with
\[
p_{ev}:=((\eta_e)_v)^+,
\qquad
n_{ev}:=((\eta_e)_v)^-,
\qquad
\mu_e:=\frac12\norm{\eta_e}_1
\]
shows that $\mu=\mu^\eta$ lies in $\mathcal M(s)$ and satisfies
\[
q(\mu)=\sum_{e\in E}\frac{\mu_e^2}{2w_e}=\sum_{e\in E}\frac{\norm{\eta_e}_1^2}{8w_e}=\mathcal D(\eta).
\]
Taking minima over feasible $\eta$ gives the reverse inequality.
\end{proof}

\subsection{Support queries as linear optimization\sectionpunct}

We now prove the support-query part of Theorem~\ref{thm:mass-support}: once we know the feasible mass set, the support problem reduces to linear optimization over $\mu$ and hence to min-cost flow on the lifted graph.

\begin{proposition}[Support LP and min-cost flow with linear arc costs]
\label{prop:support-lift}
For $r\in\R_{\ge 0}^E$, the support value $L_s(r)$ from \eqref{eq:support-primal} satisfies
\begin{equation}
\label{eq:support-dual}
L_s(r)=\min_{\mu\in\mathcal M(s)}\sum_{e\in E} r_e\mu_e.
\end{equation}
Equivalently, on the lifted graph $G^\uparrow$ of Section~\ref{sec:lifted-flow}, let
\[
c^r_a:=
\begin{cases}
r_e, & a=(e^-,e^+)\text{ for some }e\in E,\\
0, & \text{otherwise}.
\end{cases}
\]
Then
\begin{equation}
\label{eq:support-lifted-flow}
L_s(r)=
\min\Bigl\{
\sum_{a\in A^\uparrow} c^r_a f_a:
A^\uparrow f=b^\uparrow,\ f\ge 0
\Bigr\}.
\end{equation}
Finally, if $\pi\in\R^{V^\uparrow}$ is any optimal LP-dual node potential for \eqref{eq:support-lifted-flow}, then
\[
x_v:=\pi_v-\frac{\sum_{u\in V} d_u\pi_u}{\sum_{u\in V} d_u}
\qquad (v\in V)
\]
is an optimal solution of \eqref{eq:support-primal}.
\end{proposition}

\begin{proof}
Introduce variables $u_e,\ell_e\in\R$ and rewrite \eqref{eq:support-primal} as
\[
L_s(r)=
\max\Bigl\{
\ip{s}{x}:
 x\in\Xzero,\
 x_v-u_e\le 0,\
 \ell_e-x_v\le 0\ (v\in e),\
 u_e-\ell_e\le r_e\ (e\in E)
\Bigr\}.
\]
This is an ordinary finite-dimensional linear program with a closed polyhedral feasible region. It is feasible whenever $r\ge 0$, for example at
\[
x=0,
\qquad
u=0,
\qquad
\ell=0.
\]
Its objective is finite even when some $r_e=0$: along any hypergraph path, the constraints $R_e(x)\le r_e$ bound every coordinate difference by the sum of the budgets on that path, and connectedness together with $x\in\Xzero$ therefore bounds $\norm{x}_\infty$ on the feasible region. The variables $u$ and $\ell$ are then bounded as well, so the primal LP has a finite attained optimum. We emphasize that the duality used below is LP strong duality, not a Slater-type argument; zero budgets $r_e=0$ cause no difficulty.

Dualize these inequalities with multipliers $p_{ev},n_{ev},\mu_e\ge 0$, and dualize the constraint $\ip{D\1}{x}=0$ with a free scalar $\alpha$. The Lagrangian is
\[
\mathcal L=
\ip{s}{x}
+\sum_{e\in E}\sum_{v\in e} p_{ev}(u_e-x_v)
+\sum_{e\in E}\sum_{v\in e} n_{ev}(x_v-\ell_e)
+\sum_{e\in E}\mu_e(r_e-u_e+\ell_e)
+\alpha\ip{D\1}{x}.
\]
Maximizing over $u$ and $\ell$ is finite if and only if
\[
\sum_{v\in e} p_{ev}=\mu_e,
\qquad
\sum_{v\in e} n_{ev}=\mu_e
\qquad (\forall e\in E).
\]
Maximizing over $x$ is finite if and only if
\[
s_v-\sum_{e\ni v} p_{ev}+\sum_{e\ni v} n_{ev}+\alpha d_v=0
\qquad (\forall v\in V).
\]
Summing over $v$ and using $\1^\top s=0$ as well as the hyperedge equalities above gives
\[
\alpha\sum_{v\in V} d_v=0,
\]
so $\alpha=0$. Therefore the dual feasibility conditions reduce to
\[
\sum_{e\ni v}(p_{ev}-n_{ev})=s_v
\qquad (\forall v\in V),
\]
together with the two hyperedge-mass equalities. The resulting dual objective is
\[
\sum_{e\in E} r_e\mu_e.
\]
The dual LP is feasible as well: by the same path-flow construction used in Lemma~\ref{lem:lift-interiorization}, there exists a nonnegative lifted flow satisfying $A^\uparrow f=b^\uparrow$ whenever $H$ is connected and $\1^\top s=0$, and translating that flow into $(p,n,\mu)$ yields a feasible point of the displayed dual system. Therefore both the primal and dual LPs are feasible and have finite optimum values, so standard finite-dimensional LP strong duality applies and both optima are attained.
This proves \eqref{eq:support-dual}. The lifted formulation \eqref{eq:support-lifted-flow} is the same system written with arc variables
\[
f_{(e^+,v)}=p_{ev},
\qquad
f_{(v,e^-)}=n_{ev},
\qquad
f_{(e^-,e^+)}=\mu_e.
\]
Finally, the LP dual of \eqref{eq:support-lifted-flow} is the $(x,u,\ell)$ formulation above, so the vertex restriction $\bar x_v:=\pi_v$ of any optimal LP-dual node potential $\pi$ gives an optimal support maximizer modulo additive constants. The formula in the proposition subtracts the $D$-weighted average of this restriction, which is exactly the constant shift needed to put the vector in $\Xzero$. Since both the support constraints and the objective are invariant under adding constants to $x$, this shifted vector remains optimal.
\end{proof}

\subsection{Finite-capacity reduction and potential recovery\sectionpunct}

With Theorem~\ref{thm:mass-support} established, we record one further consequence needed later: the support-query min-cost-flow problem admits a finite-capacity reduction on the same lifted graph.
\begin{proposition}[Finite-capacity reduction for support queries]
\label{prop:support-finite-capacity}
Fix $r\in \R_{\ge 0}^E$ and $s\in \R^V$ with $\1^\top s=0$. Let
\[
U_0:=\sum_{v:s_v>0} s_v=\frac12\norm{s}_1.
\]
Then the lifted min-cost-flow problem with linear arc costs
\[
\min\Bigl\{\sum_{e\in E} r_e f_{(e^-,e^+)}: A^\uparrow f=b^\uparrow,\ f\ge 0\Bigr\}
\]
has an optimal solution $f^\star$ satisfying
\[
f^\star_a \le U_0 \qquad (\forall a\in A^\uparrow).
\]
Consequently, for any choice of capacities $u_a>U_0$ on every arc, adding the constraints $f_a\le u_a$ does not change the optimum.
In particular, for dyadic data one may choose dyadic capacities and then scale by a common power of two to obtain an equivalent capacitated integral min-cost-flow instance on the same lifted graph.
\end{proposition}

\begin{proof}
Start from any optimal feasible flow $f$ for \eqref{eq:support-lifted-flow}. If $C$ is a directed cycle in the support of $f$, let
\[
\tau:=\min_{a\in C} f_a.
\]
Subtracting $\tau$ units of flow on every arc of $C$ preserves the constraints $A^\uparrow f=b^\uparrow$ and the nonnegativity constraints. Because every arc cost is nonnegative, namely $0$ on transport arcs and $r_e\ge 0$ on the arcs $(e^-,e^+)$, this cycle deletion does not increase the objective. Repeating finitely many times, we obtain an optimal feasible flow that is acyclic. Let $f$ again denote such an optimal acyclic feasible flow.

Decompose $f$ into weighted directed paths. Because the node-balance vector is $b^\uparrow=(s,0,0)$, every path starts at a vertex with negative demand and ends at a vertex with positive demand; with the convention $A^\uparrow f=b^\uparrow$, the positive-demand vertices are sinks and the negative-demand vertices are sources. The total path weight is therefore
\[
\sum_{v:s_v>0} s_v = \frac12\norm{s}_1 = U_0.
\]
Since the flow on any arc is the sum of the weights of the paths using that arc, every arc carries at most the total path weight. Hence
\[
f_a\le U_0
\qquad (\forall a\in A^\uparrow).
\]
This proves the existence of an optimal solution $f^\star$ with the stated uniform bound.

Now fix capacities $u_a>U_0$ on every arc. The optimal flow $f^\star$ above satisfies $f^\star_a<u_a$ for all $a$, so it is feasible for the capacitated problem and has the same objective value as in the uncapacitated problem. Therefore adding these upper bounds does not change the optimum. If the data are dyadic, one may choose dyadic capacities $u_a>U_0$ and then scale by a common power of two; this produces an equivalent capacitated integral min-cost-flow instance on the same lifted graph.
\end{proof}

The finite-capacity reduction is useful algorithmically because it lets us call an exact routine for capacitated min-cost flow. For primal recovery, however, the object we ultimately need is an LP-dual node potential for the original support-flow LP without capacities, since Corollary~\ref{cor:support-maximizer-from-lifted-potentials} will interpret that potential as a support maximizer.
The next lemma records the standard residual-potential extraction and shows that, when the added capacities are slack for some optimum, the extracted capacitated potential is already an optimal uncapacitated potential.

\begin{lemma}[Residual potentials for capacitated min-cost flow under the incoming-minus-outgoing convention]
\label{lem:residual-potentials-capacitated-mcf}
Let $G=(N,A)$ be a directed graph, and let $A_G$ denote its incoming-minus-outgoing incidence matrix. Thus, for an arc $a=(p,q)$,
\[
\bigl(A_G^\top \pi\bigr)_a=\pi_q-\pi_p .
\]
Consider the capacitated min-cost-flow problem
\[
\min\{c^\top f: A_G f=b,\ 0\le f\le u\}.
\]
Let $f$ be an exact optimal solution. Form the residual graph $G_f$ as follows. For every arc
$a=(p,q)$, include a forward residual arc $p\to q$ of cost $c_a$ if $f_a<u_a$, and include a reverse residual arc $q\to p$ of cost $-c_a$ if $f_a>0$.

Then $G_f$ has no negative-cost directed cycle. Let $d$ be a shortest-path distance label in the graph obtained from $G_f$ by adding a new super-source with zero-cost arcs to every node, and set $\pi_v:=d_v$ for $v\in N$. Define
\[
\lambda^+_a:=\max\{0,\bigl(A_G^\top\pi\bigr)_a-c_a\}
\qquad (a\in A).
\]
Then $(\pi,\lambda^+)$ is an optimal solution of the capacitated dual
\[
\max_{\pi,\lambda^+\ge0}
\left\{
\ip{b}{\pi}-\sum_{a\in A}u_a\lambda^+_a:
\bigl(A_G^\top\pi\bigr)_a-\lambda^+_a\le c_a\quad (a\in A)
\right\}.
\]
Moreover, suppose that the uncapacitated problem
\[
\min\{c^\top f: A_G f=b,\ f\ge0\}
\]
has an optimal solution $f^{\mathrm{sl}}$ satisfying $f^{\mathrm{sl}}_a<u_a$ for every arc $a$. Then every optimal capacitated dual solution has $\lambda^+=0$. Consequently, the node potential $\pi$ above is also an optimal LP-dual node potential for the uncapacitated problem.
\end{lemma}

\begin{proof}
If $G_f$ had a negative-cost directed cycle, augmenting $f$ by a sufficiently small positive amount along that cycle would preserve $A_Gf=b$, keep the flow inside the bounds $0\le f\le u$, and strictly decrease $c^\top f$. This contradicts the optimality of $f$. Hence the distance labels are finite and well-defined after adding the zero-cost super-source.

Fix an original arc $a=(p,q)$. If the forward residual arc is present, then the distance labels give
\[
\pi_q-\pi_p\le c_a,
\]
or equivalently $\bigl(A_G^\top\pi\bigr)_a\le c_a$. If the reverse residual arc is present, then
\[
\pi_p-\pi_q\le -c_a,
\]
or equivalently $\bigl(A_G^\top\pi\bigr)_a\ge c_a$.

By the definition of $\lambda^+$, the inequalities
\[
\bigl(A_G^\top\pi\bigr)_a-\lambda^+_a\le c_a,
\qquad
\lambda^+_a\ge0
\]
hold for all arcs. Define the lower-bound slack
\[
\sigma_a:=c_a-\bigl(A_G^\top\pi\bigr)_a+\lambda^+_a\ge0.
\]
The residual inequalities imply the complementary slackness relations
\[
\sigma_a f_a=0,
\qquad
\lambda^+_a(u_a-f_a)=0
\qquad (a\in A).
\]
Indeed, if $f_a>0$, the reverse residual arc is present and gives
$\bigl(A_G^\top\pi\bigr)_a\ge c_a$, which forces $\sigma_a=0$; if $f_a<u_a$, the forward residual arc is present and gives
$\bigl(A_G^\top\pi\bigr)_a\le c_a$, which forces $\lambda^+_a=0$.

Since
\[
c=A_G^\top\pi+\sigma-\lambda^+,
\]
we obtain
\[
c^\top f
=
\pi^\top A_Gf+\sigma^\top f-(\lambda^+)^\top f
=
\ip{b}{\pi}-(\lambda^+)^\top u,
\]
where the last step uses the complementary slackness relations. Therefore the primal and dual objective values coincide, and $(\pi,\lambda^+)$ is dual optimal.

For the final claim, let $(\bar\pi,\bar\lambda^+)$ be any optimal capacitated dual solution and set
\[
\bar\sigma_a:=c_a-\bigl(A_G^\top\bar\pi\bigr)_a+\bar\lambda^+_a\ge0.
\]
The strictly slack flow $f^{\mathrm{sl}}$ is optimal for the uncapacitated problem and feasible for the capacitated problem, so it is also optimal for the capacitated problem. Strong duality gives
\[
0
=
c^\top f^{\mathrm{sl}}
-
\left(
\ip{b}{\bar\pi}
-\sum_{a\in A}u_a\bar\lambda^+_a
\right)
=
\sum_{a\in A}\bar\sigma_a f^{\mathrm{sl}}_a
+
\sum_{a\in A}\bar\lambda^+_a\left(u_a-f^{\mathrm{sl}}_a\right).
\]
Every term on the right is nonnegative, and $u_a-f^{\mathrm{sl}}_a>0$ for every arc. Hence $\bar\lambda^+_a=0$ for all $a$. Applying this to the extracted optimal pair proves $\lambda^+=0$. Thus $\pi$ satisfies $A_G^\top\pi\le c$ and achieves the common uncapacitated/capacitated optimum, so it is an optimal uncapacitated LP-dual node potential.
\end{proof}

\begin{corollary}[Support maximizer from lifted-flow potentials]
\label{cor:support-maximizer-from-lifted-potentials}
Fix $r\in\R_{\ge0}^E$ and $s\in\R^V$ with $\1^\top s=0$, and let $b^\uparrow$ be the lifted demand vector obtained from $s$. Give the lifted graph the linear support costs
\[
c_a:=
\begin{cases}
r_e, & a=(e^-,e^+)\text{ for some }e\in E,\\
0, & \text{otherwise}.
\end{cases}
\]
Suppose that $\pi\in\R^{V^\uparrow}$ satisfies
\[
\bigl((A^\uparrow)^\top\pi\bigr)_a\le c_a
\qquad (a\in A^\uparrow)
\]
and
\[
\ip{b^\uparrow}{\pi}=L_s(r).
\]
Then the restriction of $\pi$ to the original vertices, shifted by a constant so that it lies in $\Xzero$, is an optimal solution of the support maximization problem defining $L_s(r)$.
\end{corollary}

\begin{proof}
For every hyperedge $e$ and every $v\in e$, the zero-cost arc $(e^+,v)$ gives
\[
\pi_v-\pi_{e^+}\le 0,
\]
and the zero-cost arc $(v,e^-)$ gives
\[
\pi_{e^-}-\pi_v\le 0.
\]
The support-cost arc $(e^-,e^+)$ gives
\[
\pi_{e^+}-\pi_{e^-}\le r_e.
\]
Therefore, for any $u,v\in e$,
\[
\pi_u-\pi_v
\le
\pi_{e^+}-\pi_{e^-}
\le r_e.
\]
Exchanging $u$ and $v$ shows that the vertex restriction $x=(\pi_v)_{v\in V}$ satisfies
\[
R_e(x)\le r_e
\qquad (e\in E).
\]
Thus the shifted restriction is feasible for the support maximization problem. Since $b^\uparrow$ is supported on the original vertices,
\[
\ip{b^\uparrow}{\pi}=\ip{s}{x}.
\]
Shifting $x$ by a constant to enforce $x\in\Xzero$ does not change this objective because $\1^\top s=0$, and it does not change any hyperedge range. Hence the shifted restriction attains value $L_s(r)$ and is optimal.
\end{proof}

\begin{corollary}[Support queries on the same lifted graph]
\label{cor:same-lift-support-oracle}
The lifted graph $G^\uparrow$ constructed in Proposition~\ref{prop:lifted-flow-equivalence} depends only on the hypergraph $H$, not on the query vector $r$. For every budget vector $r\in \R_{\ge 0}^E$, the support value $L_s(r)$ is obtained by solving a min-cost-flow problem with linear arc costs on this same $O(P)$-arc graph, with zero costs on the transport arcs and cost $r_e$ only on the arc $(e^-,e^+)$. By Proposition~\ref{prop:support-finite-capacity}, support queries reduce to capacitated min-cost flow after choosing arc capacities larger than $\norm{s}_1/2$.
\end{corollary}

\begin{proof}
Immediate from Proposition~\ref{prop:support-lift} and Proposition~\ref{prop:support-finite-capacity}.
\end{proof}
 \section{Recovering a Primal Solution and Bounding the Primal--Dual Gap\sectionpunct}
\label{sec:primal-recovery}

This section carries out Step~5 of the roadmap. It takes a feasible lifted flow $f$ given by the first-stage oracle, its mass vector $\mu$, and the induced dual point $\eta$ from Theorem~\ref{thm:black-box-dual-solver}, together with the support-query description from Theorem~\ref{thm:mass-support}. The output is an explicit rational primal point $x$ and an explicit dyadic dual certificate $\widehat\eta$ with the stated additive guarantees, without changing the asymptotic running time.

Algorithm~\ref{alg:certified-poisson} below summarizes the complete certified solver. The technical recovery analysis focuses on three finite-precision issues inside that solver:
\begin{enumerate}[label=(\roman*),leftmargin=1.5em]
\item read the masses $\mu_e=f_{(e^-,e^+)}$ on the quadratic arcs from the first-stage lifted flow to sufficiently high absolute precision;
\item use the exact analytical budget $r_e=\mu_e/w_e$ only for the error analysis, and solve the algorithmic support problem with an upward-rounded dyadic budget $\hat r_e\ge r_e$ on the same lifted graph;
\item convert the resulting support maximizer into an additive bound on the primal objective.
\end{enumerate}
The first subsection supplies the key inequalities behind this pipeline, and the second explains how to interpret the resulting objective error as a nonlinear energy distance.

\subsection{Recovering a primal point and bounding the gap\sectionpunct}
\label{subsec:recovery-certification}

The first question is why the choice $r_e=\mu_e/w_e$ is the right one. The next lemma answers this point: if $\mu$ is nearly optimal for the quadratic mass objective, then the support problem with budget $r$ already carries essentially the same near-optimality information.
\begin{lemma}[Support gap from a near-optimal mass vector]
\label{lem:support-gap}
Let $\mu\in\mathcal M(s)$, set
\[
r_e:=\frac{\mu_e}{w_e}
\qquad (e\in E),
\]
and let $\mu^\star$ be any minimizer of $q$ over $\mathcal M(s)$. Let $\bar\mu$ be any minimizer of $\sum_e r_e\nu_e$ over $\mathcal M(s)$, and define
\[
\delta:=q(\mu)-q(\mu^\star).
\]
Then
\[
0\le 2q(\mu)-L_s(r)=\sum_{e\in E} r_e(\mu_e-\bar\mu_e)
\le \delta + \sqrt{2\delta}\,\norm{\mu-\bar\mu}_{w^{-1}}.
\]
If $x_r\in\Xzero$ is any maximizer of $L_s(r)$, then
\[
\mathcal P(x_r)-\mathrm{OPT}
\le
\sqrt{2\delta}\,\norm{\mu-\bar\mu}_{w^{-1}}.
\]
\end{lemma}

\begin{proof}
By Proposition~\ref{prop:support-lift},
\[
L_s(r)=\sum_{e\in E} r_e\bar\mu_e,
\]
and since $\bar\mu$ minimizes $\sum_e r_e\nu_e$ over $\mathcal M(s)$ while $\mu\in\mathcal M(s)$,
\[
L_s(r)=\sum_{e\in E} r_e\bar\mu_e
\le \sum_{e\in E} r_e\mu_e
=2q(\mu),
\]
which gives the displayed nonnegativity. The equality
\[
2q(\mu)-L_s(r)=\sum_{e\in E} r_e(\mu_e-\bar\mu_e)
\]
is immediate from $2q(\mu)=\sum_e r_e\mu_e$.

For the upper bound, note that $q(z)=\frac12\sum_e z_e^2/w_e$, so
\[
r_e=\frac{\mu_e}{w_e}=\frac{\mu_e^\star}{w_e}+\frac{\mu_e-\mu_e^\star}{w_e}.
\]
Therefore
\[
\sum_{e\in E} r_e(\mu_e-\bar\mu_e)
=
\sum_{e\in E}\frac{\mu_e^\star(\mu_e-\bar\mu_e)}{w_e}
+
\sum_{e\in E}\frac{(\mu_e-\mu_e^\star)(\mu_e-\bar\mu_e)}{w_e}.
\]
Because $\mu^\star$ minimizes the convex function $q$ over the convex set $\mathcal M(s)$,
\[
\sum_{e\in E}\frac{\mu_e^\star(\nu_e-\mu_e^\star)}{w_e}\ge 0
\qquad (\forall\nu\in\mathcal M(s)).
\]
Applying this first with $\nu=\bar\mu$ and then with $\nu=\mu$ gives
\[
\sum_{e\in E}\frac{\mu_e^\star(\mu_e-\bar\mu_e)}{w_e}
\le
\sum_{e\in E}\frac{\mu_e^\star(\mu_e-\mu_e^\star)}{w_e}
\le \delta,
\]
because
\[
\delta
=q(\mu)-q(\mu^\star)
=\sum_{e\in E}\frac{\mu_e^\star(\mu_e-\mu_e^\star)}{w_e}
+\frac12\norm{\mu-\mu^\star}_{w^{-1}}^2.
\]
The same identity also gives
\[
\norm{\mu-\mu^\star}_{w^{-1}}\le \sqrt{2\delta}.
\]
Hence Cauchy--Schwarz yields
\[
\sum_{e\in E}\frac{(\mu_e-\mu_e^\star)(\mu_e-\bar\mu_e)}{w_e}
\le
\norm{\mu-\mu^\star}_{w^{-1}}\,\norm{\mu-\bar\mu}_{w^{-1}}
\le
\sqrt{2\delta}\,\norm{\mu-\bar\mu}_{w^{-1}},
\]
proving the first claim.

Now let $x_r$ maximize $L_s(r)$. Since $R_e(x_r)\le r_e$ for every hyperedge,
\[
\frac12\sum_{e\in E} w_e R_e(x_r)^2
\le
\frac12\sum_{e\in E} w_e r_e^2
=q(\mu).
\]
Therefore
\[
\mathcal P(x_r)
\le q(\mu)-L_s(r).
\]
Using $\mathrm{OPT}=-q(\mu^\star)$ from Theorem~\ref{thm:poisson-fenchel-dual}, we obtain
\[
\mathcal P(x_r)-\mathrm{OPT}
\le
q(\mu)-L_s(r)+q(\mu^\star)
=
\bigl(2q(\mu)-L_s(r)\bigr)-\delta.
\]
Applying the previous bound gives the result.
\end{proof}

To turn the lemma into an algorithm, we also need polynomial bounds on the budgets and on the support maximizers that arise from them.
\begin{lemma}[Polynomial bounds for the support stage]
\label{lem:support-poly}
Assume the hypotheses of Theorem~\ref{thm:black-box-dual-solver}. Let $\mu\in\mathcal M(s)$ satisfy $q(\mu)=P^{O(1)}$, and define $r_e:=\mu_e/w_e$.
Then the following hold.
\begin{enumerate}
\item $\sum_{e\in E} r_e = P^{O(1)}$ and $\sum_{e\in E}\mu_e=P^{O(1)}$.
\item There exists a minimizer $\bar\mu$ of $\sum_e r_e\nu_e$ over $\mathcal M(s)$ such that $\norm{\bar\mu}_{w^{-1}}=P^{O(1)}$.
\item For every $\rho\in[0,1]$, every $x\in\Xzero$ satisfying $R_e(x)\le r_e+\rho$ for all $e\in E$ also satisfies $\norm{x}_\infty=P^{O(1)}$.
\end{enumerate}
\end{lemma}

\begin{proof}
For the first claim,
\[
\sum_{e\in E} r_e
=\sum_{e\in E}\frac{\mu_e}{w_e}
=\sum_{e\in E}\frac{\mu_e}{\sqrt{w_e}}\frac{1}{\sqrt{w_e}}
\le
\sqrt{\sum_{e\in E}\frac{\mu_e^2}{w_e}}\,\sqrt{\sum_{e\in E}\frac1{w_e}}
=\sqrt{2q(\mu)}\,\sqrt{\sum_{e\in E}\frac1{w_e}}
=P^{O(1)}.
\]
Likewise,
\[
\sum_{e\in E}\mu_e
\le
\sqrt{\sum_{e\in E}\frac{\mu_e^2}{w_e}}\,\sqrt{\sum_{e\in E} w_e}
=\sqrt{2q(\mu)}\,\sqrt{\sum_{e\in E} w_e}
=P^{O(1)}.
\]

For the second claim, use Proposition~\ref{prop:support-lift} and choose an optimal lifted min-cost flow $f$ for \eqref{eq:support-lifted-flow}. Decompose $f$ into source-sink paths and directed cycles. Because all arc costs are nonnegative, deleting every cycle preserves feasibility and does not increase the objective, so there is an optimal cycle-free flow. In such a decomposition, the total path weight is the total positive demand,
\[
\sum_{v:s_v>0} s_v = \frac12\norm{s}_1,
\]
because $\1^\top s=0$. Hence every arc flow is at most $\norm{s}_1/2$, and in particular
\[
\bar\mu_e:=f_{(e^-,e^+)}\le \frac12\norm{s}_1
\qquad (\forall e\in E).
\]
Therefore
\[
\norm{\bar\mu}_{w^{-1}}^2
=\sum_{e\in E}\frac{\bar\mu_e^2}{w_e}
\le \frac14\norm{s}_1^2\sum_{e\in E}\frac1{w_e}
=P^{O(1)}.
\]

For the third claim, fix $u,v\in V$. Since the hypergraph is connected, there is a hypergraph path $e_1,\dots,e_k$ from $u$ to $v$ with $k\le |V|-1\le P$. Along that path,
\[
|x_u-x_v|\le \sum_{i=1}^k R_{e_i}(x)
\le \sum_{e\in E}(r_e+\rho)
=P^{O(1)}.
\]
Thus the range of $x$ over all vertices is $P^{O(1)}$. Because $x\in\Xzero$, its $D$-weighted average is zero, so every coordinate lies between the minimum and maximum values. Hence $\norm{x}_\infty=P^{O(1)}$.
\end{proof}

The algorithmic support query will use rounded data $(\hat r,\hat s)$ lying on a common $\rho$-grid. Proposition~\ref{prop:support-finite-capacity} tells us that sufficiently slack capacities do not change the optimum of the support-flow LP. The exact routine for min-cost flow also requires controlled integer magnitudes after denominators are cleared. The next lemma packages this denominator clearing and bounds the cost of solving the rounded support instance exactly.

\begin{lemma}[Rounded support-flow implementation bound]
\label{lem:rounded-support-exact-mcf}
Let $\rho=2^{-M}$ with $\log(1/\rho)=\log^{O(1)}P$. Let
\[
\hat r\in\rho\Z_{\ge0}^E,
\qquad
\hat s\in\rho\Z^V,
\qquad
\1^\top\hat s=0,
\]
satisfy
\[
\sum_{e\in E}\hat r_e=P^{O(1)},
\qquad
\norm{\hat s}_1=P^{O(1)}.
\]
Let $\hat b^\uparrow$ be the lifted demand vector obtained from $\hat s$. Consider the rounded support-flow instance on $G^\uparrow$ with costs
\[
c_{(e^-,e^+)}=\hat r_e\quad(e\in E),
\qquad
c_a=0\text{ on all transport arcs},
\]
and uniform capacities
\[
u_a:=\norm{\hat s}_1+\rho
\qquad (a\in A^\uparrow).
\]
Then, after scaling demands, capacities, costs, and flows by $\rho^{-1}$, the capacitated problem
\[
\min\{c^\top f:A^\uparrow f=\hat b^\uparrow,\ 0\le f\le u\}
\]
becomes an equivalent integral min-cost-flow instance on $O(P)$ arcs whose integer demand, capacity, and cost magnitudes are all bounded by
\[
\exp(\log^{O(1)}P).
\]
Consequently, this instance can be solved exactly, and a residual LP-dual node potential can be extracted, in $P^{1+o(1)}$ time with high probability.
\end{lemma}

\begin{proof}
Because $\hat s\in\rho\Z^V$ and $\hat r\in\rho\Z_{\ge0}^E$, the scaled demands $\hat b^\uparrow/\rho$ and scaled costs $c/\rho$ are integral. Moreover,
\[
\frac{u_a}{\rho}
=
\frac{\norm{\hat s}_1}{\rho}+1
\in\Z
\qquad (a\in A^\uparrow),
\]
so the scaled capacities are integral as well.

We now bound the magnitudes. The scaled demand and capacity magnitudes are at most
\[
\frac{\norm{\hat s}_1+\rho}{\rho}
\le
P^{O(1)}\rho^{-1}+1
=
\exp(\log^{O(1)}P).
\]
For the costs, all transport arcs have zero cost, and for a quadratic arc $(e^-,e^+)$,
\[
\frac{c_{(e^-,e^+)}}{\rho}
=
\frac{\hat r_e}{\rho}
\le
\frac{\sum_{e'\in E}\hat r_{e'}}{\rho}
=
P^{O(1)}\rho^{-1}
=
\exp(\log^{O(1)}P).
\]
Thus, if $U_{\rm int}$ is chosen to dominate the scaled demand and capacity magnitudes, and $C_{\rm int}$ is chosen to dominate the scaled cost magnitudes, with both at least $2$, then
\[
\log U_{\rm int},\log C_{\rm int}
=
\log^{O(1)}P.
\]
At this point we use the magnitude-sensitive dependence in Theorem~1.1 of \cite{ChenKyngLiuPengProbstGutenbergSachdeva2025}: for integral min-cost-flow instances with $m$ arcs, integral vertex demands, capacities bounded by $U$, and cost magnitudes bounded by $C$, the exact running time is $m^{1+o(1)}\log U\log C$ with high probability.
Therefore the present scaled instance is solved in time
\[
(m^\uparrow)^{1+o(1)}\log U_{\rm int}\log C_{\rm int}
=
P^{1+o(1)}.
\]
The same magnitude bounds apply to the residual instances used in the dual-extraction routine of Lemma~B.9 of the J.~ACM version of \cite{ChenKyngLiuPengProbstGutenbergSachdeva2025} (Lemma~C.9 in earlier arXiv versions), whose work is $O(T_{\rm MCC}(m^\uparrow,C_{\rm int},U_{\rm int}))$ up to the same logarithmic factors.
Hence the potential extraction also costs $P^{1+o(1)}$.
\end{proof}

The support-flow solve above produces the primal potential. The final solver also has to output an explicit dual certificate, not only the first-stage oracle representation of one. The next lemma shows how to materialize the transport coordinates returned by that oracle and repair the small accumulated imbalance so that the returned dyadic dual vector satisfies $B\widehat\eta=s$ exactly while losing only a prescribed additive amount in dual objective value.

\begin{lemma}[Finite materialization and exact repair of the dual certificate]
\label{lem:dual-certificate-repair}
Assume the hypotheses of Theorem~\ref{thm:black-box-dual-solver}, and suppose that the first-stage oracle defines a lifted flow $f$ whose induced dual vector $\eta$ satisfies
\[
\eta\in\mathcal Y,\qquad B\eta=s,\qquad \mathcal D(\eta)=P^{O(1)}.
\]
Let $\Gamma>0$ satisfy $\log(1/\Gamma)=\log^{O(1)}P$. Then, by making $O(P)$ coordinate queries to the first-stage oracle at precision $\Gamma P^{-K_{\rm rep}}$ for a sufficiently large constant $K_{\rm rep}$, one can compute in $P^{1+o(1)}$ time an explicit dyadic vector $\widehat\eta\in\mathcal Y$ such that
\[
B\widehat\eta=s
\]
exactly over the rationals and
\[
\mathcal D(\widehat\eta)\le \mathcal D(\eta)+\Gamma.
\]
Every coordinate of $\widehat\eta$ has bit length $\log^{O(1)}P$.
\end{lemma}

\begin{proof}
Let $\theta:=\Gamma P^{-K_{\rm rep}}$, where $K_{\rm rep}$ will be chosen large enough after the polynomial bounds below. For every incidence $v\in e$, query the two transport coordinates $f_{(e^+,v)}$ and $f_{(v,e^-)}$ to accuracy $\theta/8$, and let
\[
z_{e,v}:=\widetilde f_{(e^+,v)}-\widetilde f_{(v,e^-)}.
\]
After changing the constant in $\theta$, we may write
\[
|z_{e,v}-(\eta_e)_v|\le \theta
\qquad (v\in e).
\]

First enforce the hyperedge-local zero-sum constraints exactly. For each $e\in E$, fix one representative vertex $a(e)\in e$, and define
\[
\bar\eta_{e,v}:=z_{e,v}\quad (v\in e,\ v\ne a(e)),
\qquad
\bar\eta_{e,a(e)}:=-\sum_{v\in e,\ v\ne a(e)}z_{e,v},
\]
with zero coordinates outside $e$. Then $\bar\eta_e\in U_e$ exactly for every $e$. Since $\eta_e$ itself is also in $U_e$,
\[
\sum_{e\in E}\|\bar\eta_e-\eta_e\|_1\le P^{O(1)}\theta.
\]
Let
\[
\Delta s:=s-B\bar\eta .
\]
Because $s$ is dyadic and $\bar\eta$ is dyadic, $\Delta s$ is dyadic. Also $\mathbf 1^\top \Delta s=0$, since $\mathbf 1^\top s=0$ and every $\bar\eta_e$ has zero sum. Moreover,
\[
\|\Delta s\|_1\le P^{O(1)}\theta.
\]

It remains to repair the global equation. Since $H$ is connected, the graph on $V$ that joins two vertices whenever they lie in a common hyperedge has a spanning tree $T$. Root $T$ at an arbitrary vertex $r_0$. For every non-root vertex $v$, let $p(v)$ be its parent and choose one hyperedge $e(v)$ containing both $v$ and $p(v)$. Let $T_v$ be the subtree rooted at $v$, and define the dyadic number
\[
t_v:=\sum_{u\in T_v}(\Delta s)_u .
\]
Initialize $\zeta=0$. For every non-root $v$, add the transfer
\[
t_v(\mathbf e_v-\mathbf e_{p(v)})
\]
to the $e(v)$-coordinate of $\zeta$. Each such transfer lies in $U_{e(v)}$, so $\zeta\in\mathcal Y$. A telescoping calculation on the rooted tree gives
\[
B\zeta=\Delta s .
\]
Indeed, at a non-root vertex $v$, the contribution is
\[
t_v-\sum_{c:\,p(c)=v}t_c=(\Delta s)_v,
\]
and at the root it is
\[
-\sum_{c:\,p(c)=r_0}t_c=(\Delta s)_{r_0},
\]
using $\sum_u(\Delta s)_u=0$. Also,
\[
\sum_{e\in E}\|\zeta_e\|_1
\le 2\sum_{v\ne r_0}|t_v|
\le 2|V|\|\Delta s\|_1
\le P^{O(1)}\theta .
\]

Now set
\[
\widehat\eta:=\bar\eta+\zeta.
\]
Then $\widehat\eta\in\mathcal Y$ and
\[
B\widehat\eta=B\bar\eta+B\zeta=B\bar\eta+\Delta s=s
\]
exactly. All coordinates are obtained from dyadic query answers and the dyadic input vector $s$ by exact rational arithmetic, so $\widehat\eta$ is dyadic. The bit length is $\log^{O(1)}P$, because the query precision has denominator exponent $\log^{O(1)}P$ and only $O(P)$ additions are performed.

It remains to bound the objective increase. Put $\Delta:=\widehat\eta-\eta$. From the estimates above,
\[
\sum_{e\in E}\|\Delta_e\|_1\le P^{O(1)}\theta.
\]
Since $\mathcal D(\eta)=P^{O(1)}$ and the weights are polynomially bounded, we have
\[
\sum_{e\in E}\|\eta_e\|_1=P^{O(1)}.
\]
Therefore, using $w_e^{-1}\le P^{O(1)}$,
\[
\begin{aligned}
\mathcal D(\widehat\eta)-\mathcal D(\eta)
&\le
\sum_{e\in E}
\frac{
2\|\eta_e\|_1\|\Delta_e\|_1+\|\Delta_e\|_1^2
}{8w_e}  \\
&\le P^{O(1)}\theta .
\end{aligned}
\]
Choosing $K_{\rm rep}$ sufficiently large makes the last quantity at most $\Gamma$. This proves the claim.
\end{proof}

\begin{algorithm}[t]
\caption{Certified cut-based hypergraph Poisson solver}
\label{alg:certified-poisson}
\begin{algorithmic}[1]
\Require A connected weighted hypergraph $H=(V,E,w)$, a dyadic demand $s\in\R^V$ with $\1^\top s=0$, and a target exponent $C>0$.
\Ensure A rational primal point $x\in\Xzero$ and an explicit dyadic dual certificate $\widehat\eta\in\mathcal Y$ with $B\widehat\eta=s$.
\State Build the lifted directed graph $G^\uparrow=(V^\uparrow,A^\uparrow)$ with transport arcs $(e^+,v)$ and $(v,e^-)$ for $v\in e$, and one quadratic arc $(e^-,e^+)$ for each $e\in E$.
\State Run the high-accuracy convex-flow first stage on $G^\uparrow$. It returns an oracle for a feasible lifted flow $f$, the masses $\mu_e=f_{(e^-,e^+)}$, and the induced dual point $\eta^{\mathrm{oracle}}$.
\State Query the mass coordinates $\mu_e$, form safe upper masses $\mu_e^\uparrow$, and set $\hat r_e=\rho\left\lceil(\mu_e^\uparrow/w_e)/\rho\right\rceil$.
\State Round $s$ to a zero-sum vector $\hat s\in\rho\Z^V$, then solve exactly the finite-capacity min-cost-flow instance for the rounded support problem on $G^\uparrow$, with budgets $\hat r$ and demand $\hat s$.
\State Extract an optimal LP-dual node potential for the support flow, restrict it to the original vertices, and shift the restriction into $\Xzero$ to obtain $x$.
\State Query all transport coordinates from the first-stage oracle and repair the materialized dyadic vector to obtain $\widehat\eta\in\mathcal Y$ satisfying $B\widehat\eta=s$ exactly.
\State \Return $(x,\widehat\eta)$.
\end{algorithmic}
\end{algorithm}

We can now state and prove the formal version of the informal main theorem from the introduction, namely Theorem~\ref{thm:intro-main}.

\begin{theorem}[Almost-linear-time solver for the hypergraph Poisson problem with $x$ restricted to $\Xzero$]
\label{thm:primal-recovery}
Assume the dyadic input model of Definition~\ref{def:dyadic-input}, that $H$ is connected, that $\1^\top s=0$, and that
\[
\norm{s}_\infty\le P^{K_0},
\qquad
P^{-K_0}\le w_e\le P^{K_0}
\qquad (\forall e\in E)
\]
for some constant $K_0\ge 1$. Then for every fixed constant $C>0$, there is a randomized algorithm that runs in $P^{1+o(1)}$ time and, with high probability over its internal randomness, returns an explicit rational primal point $x\in\Xzero$ and an explicit dyadic dual certificate $\widehat\eta\in\mathcal Y$, with every coordinate of $x$ and $\widehat\eta$ having bit length $\log^{O(1)}P$, satisfying $B\widehat\eta=s$ exactly, and such that
\[
\mathcal D(\widehat\eta)\le \mathcal D^\star + \exp(-\log^C P),
\qquad
\mathcal P(x)\le \mathrm{OPT} + \exp(-\log^C P).
\]
\end{theorem}

\begin{proof}
Fix a target exponent $C>0$, and choose an internal exponent
\[
C_{\mathrm{in}}>\max\{1,C\}.
\]
We analyze Algorithm~\ref{alg:certified-poisson}. The algorithm box is schematic; the proof below fixes the internal precision parameters and verifies the primal guarantee, the exact finite dual certificate, the bit-length bounds, and the running time.

All probabilistic statements below are with respect to the internal randomness of the invoked randomized solvers. We condition on their success events; a union bound over the constant number of such calls preserves high probability.

\paragraph{Lines~1--2: oracle output and mass vector.}
We first run the first-stage algorithm of Theorem~\ref{thm:black-box-dual-solver} with this exponent. This returns an oracle for a feasible lifted flow $f$, its mass vector $\mu$, and the induced feasible dual point $\eta^{\mathrm{oracle}}$, also represented by the oracle, with
\[
\delta:=q(\mu)-q(\mu^\star)
\le \exp(-\log^{C_{\mathrm{in}}} P).
\]
Here $\mu^\star$ is any minimizer of $q$ over $\mathcal M(s)$, so $q(\mu^\star)=\mathcal D^\star$ by Theorem~\ref{thm:mass-support}. Define the exact analytical budget
\[
r_e:=\frac{\mu_e}{w_e}.
\]
Moreover, the proof of Proposition~\ref{prop:convex-flow-assumptions} constructs a feasible lifted flow of polynomial objective value, so $\mathcal D^\star=P^{O(1)}$ and therefore $q(\mu)=P^{O(1)}$. Hence
\[
\norm{\mu}_{w^{-1}}=\sqrt{2q(\mu)}=P^{O(1)}.
\]
Lemma~\ref{lem:support-poly}(1) gives
\[
\sum_{e\in E} r_e=P^{O(1)},
\]
and Lemma~\ref{lem:support-poly}(2) lets us choose a minimizer $\bar\mu$ of $\sum_e r_e\nu_e$ over $\mathcal M(s)$ such that
\[
\norm{\bar\mu}_{w^{-1}}=P^{O(1)}.
\]

If we could solve the support problem \eqref{eq:support-primal} with this unrounded $r$, then Lemma~\ref{lem:support-gap} with this choice of $\bar\mu$ would give
\[
\mathcal P(x_r)-\mathrm{OPT}
\le \sqrt{2\delta}\,\norm{\mu-\bar\mu}_{w^{-1}}
\le \sqrt{2\delta}\bigl(\norm{\mu}_{w^{-1}}+\norm{\bar\mu}_{w^{-1}}\bigr)
= P^{O(1)}\sqrt{\delta}
\]
for a maximizer $x_r$ of $L_s(r)$. It therefore remains to implement that support solve in almost-linear time.

\paragraph{Lines~3--4: mass queries, rounding, and exact support min-cost flow.}
Choose another constant $K_1\ge 2$, which we will fix later. Fix a distinguished vertex $v_0\in V$, and set
\[
M := \left\lceil \frac{K_1 \log P + \log^{C_{\mathrm{in}}} P}{\log 2} \right\rceil,
\qquad
\rho := 2^{-M}.
\]
Then
\[
\rho \le P^{-K_1}\exp(-\log^{C_{\mathrm{in}}} P),
\]
and $\rho$ is a dyadic rational with $\log^{O(1)} P$-bit encoding. By Definition~\ref{def:first-stage-coordinate-oracle-representation}, we query each quadratic-arc mass coordinate to absolute accuracy
\[
\tau:=\frac{\rho}{2} P^{-(K_0+2)},
\]
obtaining dyadic values $\tilde\mu_e$ with
\[
|\tilde\mu_e-\mu_e|\le \tau
\qquad (\forall e\in E).
\]
This adds only $P^{1+o(1)}$ time because we read only $|E|=O(P)$ coordinates at accuracy $\tau=\exp(-\log^{O(1)}P)$. Define safe upper masses
\[
\mu_e^{\uparrow}:=\max\{0,\tilde\mu_e\}+\tau
\qquad (e\in E),
\]
so that
\[
\mu_e\le \mu_e^{\uparrow}\le \mu_e+2\tau
\qquad (\forall e\in E).
\]
Now define the algorithmic rounded budget by
\[
\hat r_e := \rho \left\lceil \frac{\mu_e^{\uparrow}/w_e}{\rho} \right\rceil
\qquad (e \in E).
\]
The exact budget $r$ is used only in the analysis; the min-cost-flow instance below uses $\hat r$. Since $w_e\ge P^{-K_0}$,
\[
\frac{2\tau}{w_e}\le 2\rho P^{-2}\le \rho,
\]
and therefore
\[
r_e \le \hat r_e \le r_e+2\rho
\qquad (\forall e \in E).
\]
For each $v\neq v_0$, let $\hat s_v$ be any nearest multiple of $\rho$ to $s_v$, and set
\[
\hat s_{v_0} := -\sum_{v \neq v_0} \hat s_v.
\]
Then
\[
\hat r\in \rho\Z_{\ge 0}^E,
\qquad
\hat s\in \rho\Z^V,
\qquad
\1^\top \hat s=0,
\]
Moreover, for every $v\neq v_0$ we have $|\hat s_v-s_v|\le \rho/2$, while
\[
|\hat s_{v_0}-s_{v_0}|
=
\left|\sum_{v\neq v_0}(s_v-\hat s_v)\right|
\le \frac{|V|-1}{2}\rho,
\]
because $\1^\top s=0$. Therefore
\[
\norm{\hat s-s}_1 \le |V|\rho \le P\rho.
\]

Apply Proposition~\ref{prop:support-finite-capacity} to the rounded support instance $(\hat r,\hat s)$. There is an optimal lifted flow whose every arc value is at most $\norm{\hat s}_1/2$. Therefore we may impose the uniform capacities
\[
u_a:=\norm{\hat s}_1+\rho
\qquad (a\in A^\uparrow),
\]
which are strictly larger than the arc flows of some optimal solution and hence do not change the optimum. More explicitly, let $\hat b^\uparrow$ be the lifted demand vector obtained from $\hat s$, and define the rounded support costs
\[
c_a:=
\begin{cases}
\hat r_e, & a=(e^-,e^+)\text{ for some }e\in E,\\
0, & \text{otherwise}.
\end{cases}
\]
We distinguish the uncapacitated support-flow LP
\[
(\mathrm{U})
\qquad
\min\{c^\top f: A^\uparrow f=\hat b^\uparrow,\ f\ge 0\}
\]
from the finite-capacity LP actually passed to the exact min-cost-flow routine,
\[
(\mathrm{C})
\qquad
\min\{c^\top f: A^\uparrow f=\hat b^\uparrow,\ 0\le f\le u\}.
\]
The strictly slack optimal flow supplied by Proposition~\ref{prop:support-finite-capacity} is feasible for $(\mathrm{C})$, so $(\mathrm{U})$ and $(\mathrm{C})$ have the same optimum value.

The rounded support instance has the magnitude bounds required by the exact implementation lemma. First,
\[
\sum_{e\in E}\hat r_e
\le
\sum_{e\in E} r_e+2\rho |E|
=
P^{O(1)},
\]
because Lemma~\ref{lem:support-poly} gives $\sum_e r_e=P^{O(1)}$, $|E|\le P$, and $\rho\le1$. Also,
\[
\norm{\hat s}_1
\le
\norm{s}_1+\norm{\hat s-s}_1
\le
P^{O(1)}.
\]
Therefore Lemma~\ref{lem:rounded-support-exact-mcf} applies to the capacitated rounded instance $(\mathrm{C})$: after scaling by $\rho^{-1}$, the exact min-cost-flow solve and the residual-potential extraction both run in $P^{1+o(1)}$ time with high probability.

The exact min-cost-flow call gives an optimal primal flow for the scaled version of $(\mathrm{C})$. To obtain the support maximizer, we use the residual-potential extraction routine invoked in Lemma~\ref{lem:rounded-support-exact-mcf} on the same scaled instance.
Lemma~\ref{lem:residual-potentials-capacitated-mcf} records the correctness of this extraction in our incoming-minus-outgoing incidence convention.
After rescaling the extracted node potentials and upper-bound multipliers by $\rho$ back to the original units, we obtain an optimal capacitated dual pair
\[
\pi\in\R^{V^\uparrow},
\qquad
\lambda^+\in\R_{\ge0}^{A^\uparrow}
\]
for $(\mathrm{C})$. Since the extraction is performed on the $\rho^{-1}$-scaled integral instance, the node potentials before rescaling may be taken to be integer min-cost-flow LP-dual node potentials. After fixing one potential to be zero, their absolute values are bounded by the number of arcs times the scaled cost magnitude, hence by $\exp(\log^{O(1)}P)$. Multiplying by the dyadic factor $\rho$ therefore gives dyadic coordinates of $\pi$ with bit length $\log^{O(1)}P$.
That is, $(\pi,\lambda^+)$ solves
\[
(\mathrm{D_C})
\qquad
\max_{\pi,\lambda^+\ge0}
\left\{\ip{\hat b^\uparrow}{\pi}-\sum_{a\in A^\uparrow}u_a\lambda^+_a:
\bigl((A^\uparrow)^\top\pi\bigr)_a-\lambda^+_a\le c_a\ (a\in A^\uparrow)
\right\}.
\]
We now check that this capacitated LP-dual node potential is already an LP-dual node potential for the uncapacitated rounded support-flow problem $(\mathrm{U})$. Let $f^{\mathrm{sl}}$ be an optimal solution of $(\mathrm{U})$ satisfying
\[
f^{\mathrm{sl}}_a\le\norm{\hat s}_1/2<u_a
\qquad (a\in A^\uparrow),
\]
whose existence follows from Proposition~\ref{prop:support-finite-capacity}. Since $f^{\mathrm{sl}}$ is strictly below the imposed capacities and the capacities do not change the optimum, Lemma~\ref{lem:residual-potentials-capacitated-mcf} implies that the upper-bound multipliers of every optimal capacitated dual solution vanish. Hence
\[
\lambda^+_a=0
\qquad (a\in A^\uparrow).
\]
Therefore $\pi$ satisfies
\[
\bigl((A^\uparrow)^\top\pi\bigr)_a\le c_a
\qquad (a\in A^\uparrow)
\]
and
\[
\ip{\hat b^\uparrow}{\pi}
=\operatorname{opt}(\mathrm{C})
=\operatorname{opt}(\mathrm{U})
=L_{\hat s}(\hat r).
\]
\paragraph{Line~5: extracting and normalizing the support potential.}
By Corollary~\ref{cor:support-maximizer-from-lifted-potentials}, restricting $\pi$ to the original vertices and then shifting by a constant so that the restriction lies in $\Xzero$ gives an optimal support maximizer for the rounded support problem. Denote this shifted restriction by $x\in\Xzero$. The normalization shift is computed by exact rational arithmetic from the restricted potential and the fixed $D$-weights; under the dyadic input model this preserves the $\log^{O(1)}P$ bit-length bound for every coordinate of $x$. Since the lifted demand vector is supported on $V$ and has total mass zero, this shift does not change the objective. Consequently
\[
\ip{\hat s}{x}=L_{\hat s}(\hat r),
\qquad
R_e(x)\le \hat r_e
\qquad (\forall e\in E).
\]

\paragraph{Primal analysis for Lines~3--5.}
Let $x_r\in\Xzero$ be any maximizer of $L_s(r)$. Since $r\le \hat r$ coordinatewise, $x_r$ is feasible for the rounded support problem, so
\[
L_{\hat s}(\hat r)
\ge \ip{\hat s}{x_r}
\ge L_s(r)-\norm{\hat s-s}_1\,\norm{x_r}_\infty.
\]
Therefore
\[
\ip{s}{x}
=\ip{\hat s}{x}+\ip{s-\hat s}{x}
\ge L_s(r)-\norm{\hat s-s}_1\bigl(\norm{x}_\infty+\norm{x_r}_\infty\bigr).
\]
By Lemma~\ref{lem:support-poly}, the exact support maximizer $x_r$ satisfies $\norm{x_r}_\infty=P^{O(1)}$, and the rounded maximizer $x$ satisfies the same bound because $R_e(x)\le \hat r_e\le r_e+2\rho$ and $2\rho\le 1$. Hence
\[
\ip{s}{x}\ge L_s(r)-P^{O(1)}\rho.
\]
On the other hand,
\[
\frac12\sum_{e\in E} w_e R_e(x)^2
\le \frac12\sum_{e\in E} w_e \hat r_e^2
\le \frac12\sum_{e\in E} w_e(r_e+2\rho)^2
= q(\mu)+2\rho\sum_{e\in E}\mu_e+2\rho^2\sum_{e\in E} w_e
= q(\mu)+P^{O(1)}\rho,
\]
where the last step uses Lemma~\ref{lem:support-poly} and the polynomial weight bounds. Consequently,
\[
\mathcal P(x)
\le q(\mu)-L_s(r)+P^{O(1)}\rho.
\]

Subtracting $\mathrm{OPT}=-\mathcal D^\star=-q(\mu^\star)$ and applying Lemma~\ref{lem:support-gap} for the exact budget $r$ gives
\[
\mathcal P(x)-\mathrm{OPT}
\le \sqrt{2\delta}\,\norm{\mu-\bar\mu}_{w^{-1}} + P^{O(1)}\rho
\le P^{O(1)}\bigl(\sqrt{\delta}+\rho\bigr).
\]
Equivalently, there is a constant $A\ge 1$ such that
\[
\mathcal P(x)-\mathrm{OPT}
\le P^A\bigl(\sqrt{\delta}+\rho\bigr).
\]
Because $C_{\mathrm{in}}>\max\{1,C\}$, we have
\[
A\log P+\log^C P+\log 2=o(\log^{C_{\mathrm{in}}} P).
\]
Hence, for all sufficiently large $P$,
\[
P^A\sqrt{\delta}
\le
P^A\exp\!\left(-\tfrac12\log^{C_{\mathrm{in}}} P\right)
\le \frac12\exp(-\log^C P).
\]
Now choose $K_1\ge \max\{2,A+1\}$. Since
\[
\rho \le P^{-K_1}\exp(-\log^{C_{\mathrm{in}}} P),
\]
we also obtain, for all sufficiently large $P$,
\[
P^A\rho
\le
P^{A-K_1}\exp(-\log^{C_{\mathrm{in}}} P)
\le \frac12\exp(-\log^C P).
\]
Thus, for all sufficiently large $P$,
\[
\mathcal P(x)-\mathrm{OPT}
\le \exp(-\log^C P).
\]
\paragraph{Line~6: materializing and repairing the dual certificate.}
Let
\[
\Gamma:=\frac12\exp(-\log^C P).
\]
By Theorem~\ref{thm:black-box-dual-solver} and the polynomial bound above on $q(\mu)$, the first-stage dual vector $\eta^{\mathrm{oracle}}$ satisfies the hypotheses of Lemma~\ref{lem:dual-certificate-repair}.
Apply Lemma~\ref{lem:dual-certificate-repair} to the first-stage dual vector $\eta^{\mathrm{oracle}}$, querying all incidence transport coordinates at the precision prescribed there. This produces an explicit dyadic vector $\widehat\eta\in\mathcal Y$ satisfying
\[
B\widehat\eta=s
\]
exactly and
\[
\mathcal D(\widehat\eta)
\le
\mathcal D(\eta^{\mathrm{oracle}})+\Gamma .
\]
By increasing the internal exponent $C_{\mathrm{in}}$ if necessary, the first-stage guarantee gives
\[
\mathcal D(\eta^{\mathrm{oracle}})
\le
\mathcal D^\star+\Gamma
\]
for all sufficiently large $P$. Hence
\[
\mathcal D(\widehat\eta)
\le
\mathcal D^\star+\exp(-\log^C P).
\]

\paragraph{Running time and success probability.}
The total running time is still $P^{1+o(1)}$, since it is the sum of the convex-flow solve in the first stage, the $O(P)$ mass-coordinate queries at precision $\tau$, the exact solve for the rounded support-flow instance and the extraction of residual potentials from Lemma~\ref{lem:rounded-support-exact-mcf}, the exact rational materialization of the $\log^{O(1)}P$-bit primal point $x$, and the $O(P)$ transport-coordinate queries and deterministic rational repair from Lemma~\ref{lem:dual-certificate-repair}. The logarithmic dependence on the scaled integer magnitudes in the exact support stage is absorbed because those magnitudes are at most $\exp(\log^{O(1)}P)$. Lemma~\ref{lem:dual-certificate-repair} gives the same $\log^{O(1)}P$ bit-length bound for $\widehat\eta$. The high-probability success guarantee follows from the high-probability guarantees of the invoked randomized solvers and the union bound over the constant number of calls.

It remains only to remove the ``sufficiently large $P$'' qualification used in the primal and dual estimates above. Let $P_0=P_0(C,K_0)$ be large enough that those estimates hold for all $P\ge P_0$. For the remaining cases $P<P_0$, the incidence size and, under Definition~\ref{def:dyadic-input}, the input bit lengths are bounded by constants depending only on the fixed parameters. On this bounded family of sizes, we invoke any exact rational convex-optimization method for the rational finite-dimensional primal--dual pair. This returns exact rational optimal primal and dual solutions, and hence satisfies the required additive guarantees. Increasing the hidden constants in the $P^{1+o(1)}$ running-time bound and in the $\log^{O(1)}P$ bit-length bound absorbs the worst-case cost and output size of this finite-size branch.
\end{proof}

\begin{corollary}[Primal--dual pair with an explicit additive error bound]
\label{cor:primal-dual-gap}
Under the assumptions of Theorem~\ref{thm:primal-recovery}, the randomized algorithm of Theorem~\ref{thm:primal-recovery} returns, with high probability, $x\in\Xzero$ and an explicit dyadic dual certificate $\widehat\eta\in\mathcal Y$ with $B\widehat\eta=s$ exactly and
\[
0\le \mathcal P(x)+\mathcal D(\widehat\eta)\le 2\exp(-\log^C P).
\]
In particular, $(x,\widehat\eta)$ gives an explicit additive error bound for \eqref{eq:poisson-primal} and \eqref{eq:dual-problem}.
\end{corollary}

\begin{proof}
Because $x\in\Xzero$ and $B\widehat\eta=s$, we have
\[
\mathcal P(x)\ge \mathrm{OPT},
\qquad
\mathcal D(\widehat\eta)\ge \mathcal D^\star.
\]
Theorem~\ref{thm:poisson-fenchel-dual} gives $\mathrm{OPT}=-\mathcal D^\star$, so
\[
\mathcal P(x)+\mathcal D(\widehat\eta)\ge 0.
\]
On the other hand, Theorem~\ref{thm:primal-recovery} gives
\[
\mathcal P(x)-\mathrm{OPT}\le \exp(-\log^C P),
\qquad
\mathcal D(\widehat\eta)-\mathcal D^\star\le \exp(-\log^C P),
\]
hence
\[
\mathcal P(x)+\mathcal D(\widehat\eta)
=(\mathcal P(x)-\mathrm{OPT})+(\mathcal D(\widehat\eta)-\mathcal D^\star)
\le 2\exp(-\log^C P).
\]
\end{proof}

\subsection{Objective gap as a nonlinear energy distance\sectionpunct}
\label{subsec:nonlinear-energy-distance}

The theorem above gives an additive objective guarantee. We now record the intrinsic interpretation of that guarantee on the nonlinear energy landscape: the same quantity is a Bregman divergence to the optimal set.

Let
\[
\mathcal X^\star := \argmin_{x\in \Xzero}\mathcal P(x).
\]
Because the energy $\mathcal E_H$ is generally not strongly convex on $\Xzero$, the
optimal set $\mathcal X^\star$ need not be a singleton. For this reason, the natural analogue of
the usual energy-norm error used for graph Laplacians is not a norm distance to a distinguished solution,
but rather a Bregman divergence generated by $\mathcal E_H$.

For $x^\star \in \mathcal X^\star$ and $\xi^\star \in \partial \mathcal E_H(x^\star)$, define the
subgradient Bregman divergence
\[
D_{\mathcal E_H}^{\xi^\star}(x,x^\star)
:=
\mathcal E_H(x)-\mathcal E_H(x^\star)-\langle \xi^\star, x-x^\star\rangle,
\qquad x\in \R^V.
\]

\begin{proposition}[Objective gap as a Bregman divergence]
\label{prop:objective-gap-bregman}
Let $x^\star \in \mathcal X^\star$. Then there exists
$\xi^\star \in \partial \mathcal E_H(x^\star)$ such that
\[
s-\xi^\star \in \operatorname{span}(D\1).
\]
For every such choice and every $x\in \Xzero$,
\[
\mathcal P(x)-\mathrm{OPT}
=
D_{\mathcal E_H}^{\xi^\star}(x,x^\star)
=
\mathcal E_H(x)-\mathcal E_H(x^\star)-\langle \xi^\star, x-x^\star\rangle.
\]
Moreover, if $\xi_1^\star,\xi_2^\star \in \partial \mathcal E_H(x^\star)$ both satisfy
\[
s-\xi_i^\star \in \operatorname{span}(D\1),
\]
then
\[
\langle \xi_1^\star-\xi_2^\star,\, x-x^\star\rangle = 0
\qquad (\forall x\in \Xzero),
\]
so the right-hand side above is independent of the compatible subgradient. In particular,
\[
D_{\mathcal E_H}^{\xi^\star}(x,x^\star)\ge 0
\qquad (\forall x\in \Xzero).
\]
\end{proposition}

\begin{proof}
By Remark~\ref{rem:xzero-poisson}, the optimal solutions of the hypergraph
Laplacian problem are the points $x^\star\in \Xzero$ satisfying
\[
s \in L_H(x^\star)+\operatorname{span}(D\1)
=
\partial \mathcal E_H(x^\star)+\operatorname{span}(D\1).
\]
Hence there exists $\xi^\star\in \partial \mathcal E_H(x^\star)$ such that
\[
s-\xi^\star \in \operatorname{span}(D\1)=\Xzero^\perp.
\]
Now let $x\in \Xzero$. Since $x-x^\star\in \Xzero$, we have
\[
\langle s-\xi^\star,\, x-x^\star\rangle = 0.
\]
Therefore
\[
\begin{aligned}
\mathcal P(x)-\mathrm{OPT}
&= \mathcal E_H(x)-\mathcal E_H(x^\star)-\langle s, x-x^\star\rangle \\
&= \mathcal E_H(x)-\mathcal E_H(x^\star)-\langle \xi^\star, x-x^\star\rangle \\
&= D_{\mathcal E_H}^{\xi^\star}(x,x^\star).
\end{aligned}
\]
If $\xi_1^\star,\xi_2^\star$ are two compatible choices, then
\[
\xi_1^\star-\xi_2^\star \in \operatorname{span}(D\1)=\Xzero^\perp,
\]
so
\[
\langle \xi_1^\star-\xi_2^\star,\, x-x^\star\rangle=0
\qquad (\forall x\in \Xzero),
\]
which proves independence of the compatible subgradient. Finally, nonnegativity follows from the
subgradient inequality for $\xi^\star \in \partial \mathcal E_H(x^\star)$.
\end{proof}

\begin{remark}[Distance to the optimal set]
\label{rem:bregman-distance-optimal-set}
A natural set-valued notion of solver error is
\[
\operatorname{Dist}^{\mathrm{Br}}_{H,s}(x,\mathcal X^\star)
:=
\inf\Bigl\{
D_{\mathcal E_H}^{\xi^\star}(x,x^\star)
:
x^\star\in \mathcal X^\star,\quad
\xi^\star\in \partial \mathcal E_H(x^\star),\quad
s-\xi^\star\in \operatorname{span}(D\1)
\Bigr\}.
\]
Proposition~\ref{prop:objective-gap-bregman} shows that this infimum is attained and equals
\[
\operatorname{Dist}^{\mathrm{Br}}_{H,s}(x,\mathcal X^\star)=\mathcal P(x)-\mathrm{OPT}.
\]
Thus the additive primal guarantee proved above is therefore a nonlinear energy distance guarantee
to the optimal set.
\end{remark}

\begin{corollary}[Nonlinear energy-distance guarantee]
\label{cor:nonlinear-energy-distance-guarantee}
Under the assumptions of Theorem~\ref{thm:primal-recovery}, with the same high probability, the returned primal point $x\in \Xzero$ satisfies
\[
0 \le D_{\mathcal E_H}^{\xi^\star}(x,x^\star)
= \mathcal P(x)-\mathrm{OPT}
\le \exp(-\log^C P)
\]
for every optimal point $x^\star\in \mathcal X^\star$ and every compatible subgradient
$\xi^\star\in \partial \mathcal E_H(x^\star)$ with
\[
s-\xi^\star \in \operatorname{span}(D\1).
\]
Equivalently,
\[
\operatorname{Dist}^{\mathrm{Br}}_{H,s}(x,\mathcal X^\star)
\le
\exp(-\log^C P).
\]
\end{corollary}

\begin{proof}
Combine Theorem~\ref{thm:primal-recovery} with Proposition~\ref{prop:objective-gap-bregman}.
\end{proof}

\begin{corollary}[2-uniform graph specialization]
\label{cor:graph-specialization-energy-norm}
Assume every hyperedge has size two, and let \(L\) denote the weighted graph Laplacian of the
resulting graph. Then
\[
\mathcal E_H(x)=\frac12 x^\top Lx,
\]
and the weighted-zero-mean minimizer \(x^\star\in \Xzero\) is the unique solution of
\[
Lx=s,
\qquad
x\in \Xzero.
\]
Moreover, for every \(x\in \Xzero\),
\[
\mathcal P(x)-\mathrm{OPT}
=
\frac12 (x-x^\star)^\top L(x-x^\star).
\]
Equivalently, if
\[
\|z\|_L := \sqrt{z^\top L z},
\]
then
\[
\mathcal P(x)-\mathrm{OPT} = \frac12 \|x-x^\star\|_L^2.
\]
In particular, with the same high probability, the output of Theorem~\ref{thm:primal-recovery} satisfies
\[
\frac12 \|x-x^\star\|_L^2
\le
\exp(-\log^C P).
\]
Thus, in the ordinary graph case, the additive primal guarantee of
Theorem~\ref{thm:primal-recovery} reduces to an additive energy-norm error guarantee.
\end{corollary}

\begin{proof}
When every hyperedge has size two, \(R_{\{u,v\}}(x)=|x_u-x_v|\), so
\[
\mathcal E_H(x)
=
\frac12 \sum_{\{u,v\}\in E} w_{uv}(x_u-x_v)^2
=
\frac12 x^\top Lx.
\]
Hence
\[
\partial \mathcal E_H(x)=\{Lx\}.
\]
If \(x^\star\in \Xzero\) is optimal, Proposition~\ref{prop:objective-gap-bregman} gives
\[
s-Lx^\star \in \operatorname{span}(D\1).
\]
Taking the inner product with \(\1\), and using
\[
\1^\top s = 0,\qquad \1^\top Lx^\star = 0,\qquad \1^\top D\1 > 0,
\]
we obtain
\[
Lx^\star = s.
\]
Since the graph is connected and the condition \(x^\star\in \Xzero\) fixes the additive constant, this
solution is unique.

Finally, using \(Lx^\star=s\),
\[
\begin{aligned}
\mathcal P(x)-\mathrm{OPT}
&=
\frac12 x^\top Lx
-\frac12 (x^\star)^\top Lx^\star
-\langle Lx^\star, x-x^\star\rangle \\
&=
\frac12 (x-x^\star)^\top L(x-x^\star),
\end{aligned}
\]
which is the claimed identity.
\end{proof}
 \section{Regularized Problems via a Ground Vertex\sectionpunct}
\label{sec:applications-limitations}

This section shows how to reduce regularized versions of the hypergraph Poisson problem to one augmented Poisson instance by adding a ground vertex. We then note a few direct consequences of the resulting solver.

\subsection{Solving regularized problems by adding a ground vertex\sectionpunct}
\label{subsec:shifted-grounded}

For \(\lambda>0\), consider the regularized Poisson problem
\[
s\in L_H(x)+\lambda Dx,
\qquad
x\in\mathbb R^V.
\]
Equivalently, consider the variational objective
\[
\mathcal P_\lambda(x)
:=
\frac12 \sum_{e\in E} w_e R_e(x)^2
+\frac{\lambda}{2}\langle Dx,x\rangle
-\langle s,x\rangle,
\qquad x\in\mathbb R^V.
\]
In this subsection, we do not require \(H\) to be connected. Instead, assume that every vertex has positive weighted degree,
\[
d_v=\sum_{e\ni v} w_e>0
\qquad (\forall v\in V).
\]
Under this assumption, the regularized objective is coercive and no additive-constant normalization is needed.

Rather than design a separate solver for this regularized Poisson problem, we reduce it to the Poisson solver from the previous sections by adding one new ground vertex. The quadratic term $\frac{\lambda}{2}\langle Dx,x\rangle$ then becomes a family of two-vertex hyperedges: once the new ground coordinate is fixed to zero, the range on \(\{v,g\}\) is \(|x_v|\). Variants of this construction also cover unbalanced source vectors and, at the level of modeling, suggest a natural way to impose grounded boundary conditions: any net imbalance in \(s\) is absorbed at \(g\), while grounding a prescribed terminal set would amount to collapsing it to the new vertex \(g\). The next proposition shows that this turns the regularized problem into the same Poisson problem on an augmented hypergraph.

\begin{proposition}[Reducing the regularized Poisson problem by adding a ground vertex]
\label{prop:grounded-shifted-reduction}
Let \(H=(V,E,w)\) be a weighted hypergraph with positive weighted degrees \(d_v=\sum_{e\ni v}w_e>0\). Let \(\lambda>0\) and \(s\in\mathbb R^V\), and let \(\mathcal P_\lambda\) be the regularized objective above.
Let \(g\notin V\) be a new ground vertex and define
\[
\bar V := V\cup\{g\},
\qquad
\bar E := E \cup \{\{v,g\}: v\in V\}.
\]
Assign weights
\[
\bar w_e := w_e \quad (e\in E),
\qquad
\bar w_{\{v,g\}} := \lambda d_v \quad (v\in V),
\]
and define \(\bar s\in\mathbb R^{\bar V}\) by
\[
\bar s_v := s_v \quad (v\in V),
\qquad
\bar s_g := -\1^\top s.
\]
Let
\[
\mathcal P_{\bar H}(\bar x)
:=
\frac12 \sum_{\bar e\in \bar E} \bar w_{\bar e} R_{\bar e}(\bar x)^2
-\langle \bar s,\bar x\rangle,
\qquad
\bar x\in\mathbb R^{\bar V}.
\]
Then
\[
\min_{x\in\mathbb R^V} \mathcal P_\lambda(x)
=
\min\bigl\{\mathcal P_{\bar H}(\bar x): \bar x\in\mathbb R^{\bar V},\ \bar x_g=0\bigr\}.
\]
Equivalently, solving the regularized problem on \(H\) is the same as solving the corresponding problem on the augmented hypergraph \(\bar H=(\bar V,\bar E,\bar w)\) with the ground pinned to zero.
\end{proposition}

\begin{proof}
For any \(x\in\mathbb R^V\), set \(\bar x=(x,0)\). Then for every original hyperedge \(e\in E\),
\[
R_e(\bar x)=R_e(x),
\]
while for each added two-vertex hyperedge \(\{v,g\}\),
\[
R_{\{v,g\}}(\bar x)=|x_v-0|=|x_v|.
\]
Therefore
\[
\mathcal P_{\bar H}(x,0)
=
\frac12 \sum_{e\in E} w_e R_e(x)^2
+\frac12 \sum_{v\in V} \lambda d_v |x_v|^2
-\langle s,x\rangle
=
\mathcal P_\lambda(x).
\]
This proves the claimed identity of the two optimization problems.
\end{proof}

\begin{theorem}[Dual of the regularized Poisson problem]
\label{thm:shifted-dual}
Assume that \(d_v>0\) for every \(v\in V\). Let
\[
Y := \prod_{e\in E} U_e,
\qquad
B\eta := \sum_{e\in E}\eta_e,
\]
and define
\[
\mathcal D_\lambda(\eta)
:=
\sum_{e\in E}\frac{\|\eta_e\|_1^2}{8w_e}
+
\frac{1}{2\lambda}\sum_{v\in V}\frac{(s_v-(B\eta)_v)^2}{d_v},
\qquad
\eta\in Y.
\]
Then
\[
\mathrm{OPT}_\lambda
:=
\min_{x\in\mathbb R^V} \mathcal P_\lambda(x)
=
-\min_{\eta\in Y} \mathcal D_\lambda(\eta).
\]
In particular,
\[
\mathcal D_\lambda^\star:=\min_{\eta\in Y} \mathcal D_\lambda(\eta)=-\mathrm{OPT}_\lambda.
\]
If \(\eta^\star\) minimizes \(\mathcal D_\lambda\), then the unique primal minimizer is
\[
x_v^\star
=
\frac{s_v-(B\eta^\star)_v}{\lambda d_v}
\qquad
(v\in V).
\]
\end{theorem}

\begin{proof}
The dual calculation is the same hyperedgewise Fenchel duality as in Proposition~\ref{prop:hyperedge-conjugate}, now applied to each original hyperedge and combined with the additional quadratic term. Indeed,
\[
\frac12 \sum_{e\in E} w_e R_e(x)^2
=
\sup_{\eta\in Y}
\left\{
\langle B\eta,x\rangle
-
\sum_{e\in E}\frac{\|\eta_e\|_1^2}{8w_e}
\right\}.
\]
Hence
\[
\mathcal P_\lambda(x)
=
\sup_{\eta\in Y}
\left\{
\frac{\lambda}{2}\langle Dx,x\rangle
+
\langle B\eta-s,x\rangle
-
\sum_{e\in E}\frac{\|\eta_e\|_1^2}{8w_e}
\right\}.
\]
Because \(d_v>0\) for every \(v\in V\) by hypothesis and \(\lambda>0\), the quadratic term is strongly convex. Standard Fenchel--Rockafellar duality therefore gives
\[
\mathrm{OPT}_\lambda
=
\sup_{\eta\in Y}
\inf_{x\in\mathbb R^V}
\left\{
\frac{\lambda}{2}\langle Dx,x\rangle
+
\langle B\eta-s,x\rangle
-
\sum_{e\in E}\frac{\|\eta_e\|_1^2}{8w_e}
\right\}.
\]
For fixed \(\eta\), the inner minimization is separable over vertices and equals
\[
-\frac{1}{2\lambda}\sum_{v\in V}\frac{(s_v-(B\eta)_v)^2}{d_v}
-
\sum_{e\in E}\frac{\|\eta_e\|_1^2}{8w_e},
\]
attained at
\[
x_v
=
\frac{s_v-(B\eta)_v}{\lambda d_v}.
\]
Taking the supremum over \(\eta\in Y\) proves the claim.
\end{proof}

\paragraph{Remark (dual interpretation of the ground hyperedges).}
The dual above can also be read directly from the reduction in Proposition~\ref{prop:grounded-shifted-reduction}, where we add a ground vertex. Apply the unregularized dual to the augmented hypergraph \(\bar H\), and write an augmented dual point as the variables \(\eta\in Y\) on the original hyperedges together with one variable \(\xi_v\in U_{\{v,g\}}\) for each added hyperedge. Every such added-hyperedge variable has the form
\[
\xi_v=\alpha_v(\ee_v-\ee_g).
\]
Augmented feasibility at the original vertex \(v\) forces
\[
\alpha_v=s_v-(B\eta)_v,
\]
and the feasibility equation at \(g\) then follows automatically because each original-hyperedge variable has zero coordinate sum. Thus the contribution of the added hyperedge \(\{v,g\}\), whose weight is \(\lambda d_v\), is
\[
\frac{\|\xi_v\|_1^2}{8\lambda d_v}
=
\frac{(s_v-(B\eta)_v)^2}{2\lambda d_v}.
\]
Eliminating the added-hyperedge variables in this way gives exactly the quadratic term in \(\mathcal D_\lambda\). This is the correspondence between augmented and original dual certificates used in the proof of Theorem~\ref{thm:shifted-cut-based-solver}.

We can now state the formal version of the companion theorem from the introduction, namely Theorem~\ref{thm:intro-shifted}.

\begin{theorem}[Almost-linear-time solver for the regularized Poisson problem]
\label{thm:shifted-cut-based-solver}
Assume the dyadic input model of Definition~\ref{def:dyadic-input}, that \(d_v=\sum_{e\ni v} w_e>0\) for every \(v\in V\), that
\[
\norm{s}_\infty \le P^{K_0},
\qquad
P^{-K_0}\le w_e\le P^{K_0}\quad(\forall e\in E),
\]
and that \(\lambda>0\) is \(L_\lambda\)-dyadic for some \(L_\lambda=\log^{O(1)}P\) and
\[
P^{-K_\lambda}\le \lambda\le P^{K_\lambda}
\]
for some constants \(K_0,K_\lambda\ge 1\). Then for every fixed constant \(C>0\), there is a randomized algorithm that runs in \(P^{1+o(1)}\) time and, with high probability over its internal randomness, returns an explicit rational primal point \(x\in\mathbb R^V\) and an explicit dyadic dual certificate \(\widehat\eta\in Y\), with every coordinate of \(x\) and \(\widehat\eta\) having bit length \(\log^{O(1)}P\), such that
\[
\mathcal D_\lambda(\widehat\eta)\le \mathcal D_\lambda^\star+\exp(-\log^C P),
\qquad
\mathcal P_\lambda(x)\le \mathrm{OPT}_\lambda+\exp(-\log^C P).
\]
Consequently, on the same high-probability event,
\[
0\le \mathcal P_\lambda(x)+\mathcal D_\lambda(\widehat\eta)\le 2\exp(-\log^C P).
\]
\end{theorem}

\begin{proof}
Fix the requested exponent \(C>0\), and let \(C_{\mathrm{in}}>C\) be a larger fixed constant to be chosen below. Define the augmented hypergraph
\[
\bar V:=V\cup\{g\},
\qquad
\bar E:=E\cup\{\{v,g\}:v\in V\},
\]
with weights
\[
\bar w_e:=w_e \quad (e\in E),
\qquad
\bar w_{\{v,g\}}:=\lambda d_v \quad (v\in V).
\]
Define \(\bar s\in\R^{\bar V}\) by
\[
\bar s_v:=s_v \quad (v\in V),
\qquad
\bar s_g:=-\1^\top s.
\]
Then \(\1^\top \bar s=0\). Since \(d_v>0\) and \(\lambda>0\), every added hyperedge \(\{v,g\}\) has positive weight \(\lambda d_v\), so the augmented hypergraph \(\bar H=(\bar V,\bar E,\bar w)\) is connected. Its incidence size is
	\[
	\bar P=P+2|V|=O(P).
	\]

	Let \(L=\log^{O(1)}P\) be such that every input weight and demand is \(L\)-dyadic. Because \(d_v>0\), every vertex lies in at least one positive-weight incident hyperedge. Hence
	\[
	d_v=\sum_{e\ni v} w_e
	\]
	satisfies
	\[
	P^{-K_0}\le d_v\le P^{K_0+1}
	\qquad (\forall v\in V),
	\]
	where the upper bound uses that at most \(P\) incidences can meet a fixed vertex.

	Now fix \(v\in V\), and write each incident weight as
	\[
	w_e=a_e2^{-q_e}
	\qquad (e\ni v),
	\]
	with \(q_e\le L\). Let
	\[
	Q_v:=\max_{e\ni v} q_e\le L.
	\]
	Then
	\[
	d_v=\sum_{e\ni v} w_e=2^{-Q_v}\sum_{e\ni v} b_{e,v},
	\qquad
	b_{e,v}:=2^{Q_v}w_e\in\mathbb Z.
	\]
	Because \(w_e\le P^{K_0}\) and \(Q_v\le L\),
	\[
	|b_{e,v}|\le 2^{L}P^{K_0},
	\]
	so each \(b_{e,v}\) has bit length \(L+O(\log P)\). There are at most \(P\) terms in the sum, hence
	\[
	d_v=B_v2^{-Q_v}
	\]
	for an integer \(B_v\) with bit length \(L+O(\log P)\). In particular, each \(d_v\) is \((L+O(\log P))\)-dyadic.

	Since \(\lambda\) is \(L_\lambda\)-dyadic, the product \(\lambda d_v\) is dyadic with denominator exponent at most \(L+L_\lambda\). Moreover,
	\[
	P^{-(K_0+K_\lambda)} \le \lambda d_v \le P^{K_0+K_\lambda+1}
	\qquad (\forall v\in V),
	\]
	so the numerator of \(\lambda d_v\) has bit length \(L+L_\lambda+O(\log P)\). Therefore every added weight
	\[
	\bar w_{\{v,g\}}:=\lambda d_v
	\]
	is dyadic with \(\log^{O(1)}P\)-bit encoding.
	Also,
	\[
	\norm{\bar s}_\infty\le \max\{\norm{s}_\infty,\ |\1^\top s|\}
	\le \max\{P^{K_0},|V|P^{K_0}\}
	=P^{O(1)}.
	\]
	Since \(\bar P=O(P)\), these are polynomial bounds in the augmented input size as well, and the dyadic bit lengths remain \(\log^{O(1)}\bar P\). Thus the augmented instance satisfies the same dyadic input model needed to apply Theorem~\ref{thm:primal-recovery}. Let \(\bar D\) be the degree matrix of \(\bar H\), and write
	\[
	\bar X_0:=\{\bar x\in\R^{\bar V}:\ip{\bar D\bar x}{\1}=0\}.
	\]

Apply Theorem~\ref{thm:primal-recovery} to the augmented instance \((\bar H,\bar s)\) with exponent parameter \(C_{\mathrm{in}}\). On the high-probability success event of that randomized algorithm, this returns \(\bar x\in \bar X_0\) and an explicit dyadic dual certificate \(\widehat{\bar\eta}\) that is feasible for the augmented instance \(\bar H\). Define
\[
\bar{\mathcal D}(\widehat{\bar\eta}):=\sum_{\bar e\in \bar E}\frac{\|\widehat{\bar\eta}_{\bar e}\|_1^2}{8\bar w_{\bar e}},
\]
and let \(\overline{\mathrm{OPT}}:=\min_{\bar y\in \bar X_0} \mathcal P_{\bar H}(\bar y)\) and \(\bar{\mathcal D}^\star\) denote the dual optimum on \(\bar H\). Then
\[
\mathcal P_{\bar H}(\bar x)\le \overline{\mathrm{OPT}}+\exp(-\log^{C_{\mathrm{in}}}\bar P),
\qquad
\bar{\mathcal D}(\widehat{\bar\eta})\le \bar{\mathcal D}^\star+\exp(-\log^{C_{\mathrm{in}}}\bar P).
\]
Because \(\bar P=O(P)\) and \(C_{\mathrm{in}}>C\), we may choose \(C_{\mathrm{in}}\) so that
\[
\exp(-\log^{C_{\mathrm{in}}}\bar P)\le \exp(-\log^C P).
\]

Now shift the returned primal point by
\[
\tilde x:=\bar x-\bar x_g\1.
\]
Since \(\1^\top \bar s=0\), the objective \(\mathcal P_{\bar H}\) is invariant under adding constants, so
\[
\mathcal P_{\bar H}(\tilde x)=\mathcal P_{\bar H}(\bar x).
\]
Also \(\tilde x_g=0\). Let \(x:=\tilde x|_V\). By Proposition~\ref{prop:grounded-shifted-reduction},
\[
\mathcal P_\lambda(x)=\mathcal P_{\bar H}(\tilde x).
\]

Write the augmented dual certificate \(\widehat{\bar\eta}\) as follows: for each original hyperedge \(e\in E\), let \(\widehat\eta_e\) be the corresponding component, and for each added hyperedge \(\{v,g\}\), let \(\widehat\xi_v\in U_{\{v,g\}}\) be the corresponding component. Then \(\widehat\eta:=(\widehat\eta_e)_{e\in E}\in Y\). Since every vector in \(U_{\{v,g\}}\) has the form
\[
\widehat\xi_v=\widehat\alpha_v(\ee_v-\ee_g)
\]
for some scalar \(\widehat\alpha_v\), exact feasibility of \(\widehat{\bar\eta}\) at the original vertex \(v\) gives
\[
\widehat\alpha_v=s_v-(B\widehat\eta)_v.
\]
Therefore
\[
\frac{\|\widehat\xi_v\|_1^2}{8\lambda d_v}
=\frac{\widehat\alpha_v^2}{2\lambda d_v}
=\frac{(s_v-(B\widehat\eta)_v)^2}{2\lambda d_v}.
\]
Summing over \(v\in V\), we obtain
\[
\bar{\mathcal D}(\widehat{\bar\eta})
=
\sum_{e\in E}\frac{\|\widehat\eta_e\|_1^2}{8w_e}
+
\sum_{v\in V}\frac{\|\widehat\xi_v\|_1^2}{8\lambda d_v}
=
\mathcal D_\lambda(\widehat\eta).
\]
By Theorem~\ref{thm:primal-recovery}, every coordinate of \(\bar x\) and \(\widehat{\bar\eta}\) has bit length \(\log^{O(1)}\bar P=\log^{O(1)}P\). Restricting \(\widehat{\bar\eta}\) to the original hyperedges preserves this bound for \(\widehat\eta\), and each coordinate of \(x\) is the difference \(\bar x_v-\bar x_g\) of two such dyadic numbers, so every coordinate of \(x\) also has bit length \(\log^{O(1)}P\).

By Proposition~\ref{prop:grounded-shifted-reduction}, minimizing \(\mathcal P_{\bar H}\) over the affine slice \(\{\bar x:\bar x_g=0\}\) gives \(\mathrm{OPT}_\lambda\). Since \(\1^\top \bar s=0\), adding constants does not change \(\mathcal P_{\bar H}\), so
\[
\overline{\mathrm{OPT}}=\mathrm{OPT}_\lambda.
\]
Applying Theorem~\ref{thm:poisson-fenchel-dual} to the augmented instance gives
\[
\bar{\mathcal D}^\star=-\overline{\mathrm{OPT}}.
\]
By Theorem~\ref{thm:shifted-dual},
\[
\mathcal D_\lambda^\star=-\mathrm{OPT}_\lambda,
\]
so \(\bar{\mathcal D}^\star=\mathcal D_\lambda^\star\). Together with \(\bar{\mathcal D}(\widehat{\bar\eta})=\mathcal D_\lambda(\widehat\eta)\), this yields
\[
\mathcal P_\lambda(x)\le \mathrm{OPT}_\lambda+\exp(-\log^C P),
\qquad
\mathcal D_\lambda(\widehat\eta)\le \mathcal D_\lambda^\star+\exp(-\log^C P).
\]
Finally,
\[
\mathcal P_\lambda(x)\ge \mathrm{OPT}_\lambda,
\qquad
\mathcal D_\lambda(\widehat\eta)\ge \mathcal D_\lambda^\star,
\]
so duality gives
\[
0\le \mathcal P_\lambda(x)+\mathcal D_\lambda(\widehat\eta)
=(\mathcal P_\lambda(x)-\mathrm{OPT}_\lambda)+(\mathcal D_\lambda(\widehat\eta)-\mathcal D_\lambda^\star)
\le 2\exp(-\log^C P).
\]
The running time is \(\bar P^{1+o(1)}=P^{1+o(1)}\), and the high-probability success guarantee is inherited from Theorem~\ref{thm:primal-recovery}.
\end{proof}

\subsection{Direct consequences\sectionpunct}
\label{sec:applications}

We keep this discussion brief. The main theorems above yield a few immediate query procedures, but this paper does not analyze conductance, local-clustering, or learning guarantees built on top of them.

\paragraph{Pairwise response queries.}
The most direct special case is obtained by taking $s=e_u-e_v$ for a pair $u,v\in V$. Then any optimal point $x^\star\in\Xzero$ may be interpreted as the potential induced by injecting one unit at $u$ and removing one unit at $v$, paralleling the graph Poisson problem. Although the optimal potential need not be unique, the scalar
\[
\mathcal R_H(u,v):=\langle e_u-e_v,x^\star\rangle
\qquad
\text{for any optimal }x^\star\in\Xzero
\]
is independent of the choice of optimal point. Indeed, Proposition~\ref{prop:objective-gap-bregman} gives a compatible subgradient
\(\xi^\star\in\partial\mathcal E_H(x^\star)\) with
\[
s-\xi^\star\in\operatorname{span}(D\1).
\]
Since \(\mathcal E_H\) is convex and positively homogeneous of degree \(2\), the subgradient Euler identity gives
\[
\langle \xi^\star,x^\star\rangle=2\mathcal E_H(x^\star).
\]
Because \(x^\star\in\Xzero\), the compatibility condition implies
\[
\langle s,x^\star\rangle=\langle \xi^\star,x^\star\rangle.
\]
Thus
\[
\mathcal P(x^\star)=\mathcal E_H(x^\star)-\langle s,x^\star\rangle
=-\mathcal E_H(x^\star).
\]
All optimal points have the same value \(\mathcal P(x^\star)=\mathrm{OPT}\), hence the same \(\mathcal E_H(x^\star)\) and therefore the same response \(\langle s,x^\star\rangle\). Equivalently,
\[
\mathcal R_H(u,v)=x_u^\star-x_v^\star.
\]
Thus the exact Poisson problem defines a well-defined response quantity analogous to effective resistance in the cut-based model. Theorem~\ref{thm:primal-recovery} gives an additive objective-gap certificate for the returned approximate primal point; without additional curvature, we do not claim a bound on
\[
\left|\langle e_u-e_v,x\rangle-\mathcal R_H(u,v)\right|
\]
for the approximate unregularized output \(x\). When every hyperedge has size two, this recovers the usual graph analogue.

\paragraph{Resolvents.}
A particularly useful reformulation of the regularized solver is obtained by fixing a reference vector $y\in\R^V$ for which the resulting demand vector $s=\lambda Dy$ satisfies the assumptions of Theorem~\ref{thm:shifted-cut-based-solver} on dyadic input and polynomial bounds. With this choice of $s$,
\[
\mathcal P_\lambda(x)
=
\frac12\sum_{e\in E} w_eR_e(x)^2
+
\frac{\lambda}{2}\langle D(x-y),x-y\rangle
-
\frac{\lambda}{2}\langle Dy,y\rangle,
\]
so Theorem~\ref{thm:shifted-cut-based-solver} computes the proximal point
\[
J_\lambda(y)
:=
\arg\min_{x\in\R^V}
\left\{
\frac12\sum_{e\in E} w_eR_e(x)^2
+
\frac{\lambda}{2}\langle D(x-y),x-y\rangle
\right\}.
\]
Equivalently,
\[
\lambda D\bigl(y-J_\lambda(y)\bigr)\in L_H\bigl(J_\lambda(y)\bigr).
\]

\begin{corollary}[Vector accuracy for regularized resolvents]
\label{cor:regularized-resolvent-vector-error}
Under the assumptions of Theorem~\ref{thm:shifted-cut-based-solver}, let \(x^\star\) be the unique minimizer of \(\mathcal P_\lambda\). On the same high-probability event, the returned primal point \(x\) satisfies
\[
\frac{\lambda}{2}\normD{x-x^\star}^2
\le
\mathcal P_\lambda(x)-\mathrm{OPT}_\lambda
\le
\exp(-\log^C P).
\]
Equivalently,
\[
\normD{x-x^\star}
\le
\sqrt{\frac{2}{\lambda}}\,
\exp\!\left(-\frac12\log^C P\right).
\]
In the resolvent specialization \(s=\lambda Dy\), this unique minimizer is \(J_\lambda(y)\), and hence
\[
\frac{\lambda}{2}\normD{x-J_\lambda(y)}^2
\le
\exp(-\log^C P).
\]
\end{corollary}

\begin{proof}
The hypergraph energy \(x\mapsto \frac12\sum_{e\in E}w_eR_e(x)^2\) is convex, and the term \(\frac{\lambda}{2}\langle Dx,x\rangle\) is \(\lambda\)-strongly convex with respect to \(\normD{\cdot}\). Therefore \(\mathcal P_\lambda\) is \(\lambda\)-strongly convex with respect to \(\normD{\cdot}\). Since \(x^\star\) minimizes \(\mathcal P_\lambda\),
\[
\mathcal P_\lambda(x)
\ge
\mathcal P_\lambda(x^\star)
+
\frac{\lambda}{2}\normD{x-x^\star}^2
=
\mathrm{OPT}_\lambda
+
\frac{\lambda}{2}\normD{x-x^\star}^2.
\]
Combining this inequality with the primal objective guarantee in Theorem~\ref{thm:shifted-cut-based-solver} proves the first display and the equivalent norm bound. If \(s=\lambda Dy\), the identity defining \(J_\lambda(y)\) differs from \(\mathcal P_\lambda\) only by the constant \(\frac{\lambda}{2}\langle Dy,y\rangle\), so \(x^\star=J_\lambda(y)\).
\end{proof}

Accordingly, Theorem~\ref{thm:shifted-cut-based-solver} and Corollary~\ref{cor:regularized-resolvent-vector-error} give an almost-linear-time algorithm for evaluating the resolvent. The algorithm returns an approximate proximal point with an explicit guarantee on the vector error in the \(D\)-norm. Equivalently, it gives one certified backward-Euler step, or one certified proximal step, for the same cut-based hypergraph Laplacian.

\subsection{Personalized PageRank and clustering\sectionpunct}
\label{sec:pagerank}
We now specialize the resolvent guarantee to a concrete vector used in hypergraph clustering. In the normalization of Takai, Miyauchi, Ikeda, and Yoshida \cite[Section~3.3]{TakaiMiyauchiIkedaYoshida2020}, the random-walk normalized operator maps $p\in\R^V$ to $L_H(D^{-1}p)$. For a seed distribution $z\in\R^V_{\ge0}$ with $\1^\top z=1$ and a teleportation parameter $\alpha\in(0,1)$, the personalized PageRank vector $p=\operatorname{pr}_\alpha(z)$ is the unique solution of
\begin{equation}
\label{eq:pagerank-inclusion}
z\in p+\frac{1-\alpha}{2\alpha}L_H(D^{-1}p).
\end{equation}
The factor $2$ comes from the lazy-walk convention in that definition. Set
\begin{equation}
\label{eq:pagerank-parameters}
\lambda:=\frac{2\alpha}{1-\alpha},\qquad y:=D^{-1}z,\qquad s:=\lambda z.
\end{equation}
Writing $p=Dx$, inclusion~\eqref{eq:pagerank-inclusion} becomes $s\in L_H(x)+\lambda Dx$, so
\[
\operatorname{pr}_\alpha(z)=DJ_\lambda(D^{-1}z).
\]
This identity holds over the reals. The following computational consequence requires the transformed inputs to satisfy our finite-precision assumptions.

\begin{corollary}[Certified approximation of personalized PageRank]
\label{cor:pagerank}
Let $H$ have positive weighted degrees, let $z$ be a seed distribution, and let $\alpha\in(0,1)$. Assume that the weights, $s=\lambda z$, and $\lambda=2\alpha/(1-\alpha)$ satisfy the dyadic input and polynomial bounds of Theorem~\ref{thm:shifted-cut-based-solver}. For every fixed $C>0$, a randomized $P^{1+o(1)}$-time algorithm returns $\widehat p=Dx$ such that, with high probability,
\[
\|\widehat p-\operatorname{pr}_\alpha(z)\|_{D^{-1}}
\le\sqrt{\frac{2}{\lambda}}\exp\!\left(-\frac12\log^C P\right).
\]
More generally, a primal objective error at most $\delta$ in the regularized solve implies a vector error at most $\sqrt{2\delta/\lambda}$ in this norm.
\end{corollary}
\begin{proof}
With the parameters in \eqref{eq:pagerank-parameters}, the regularized minimizer is $x^\star=J_\lambda(D^{-1}z)$. Corollary~\ref{cor:regularized-resolvent-vector-error} gives $\|x-x^\star\|_D^2\le2\delta/\lambda$. Since
\[
\|D(x-x^\star)\|_{D^{-1}}^2=\|x-x^\star\|_D^2,
\]
the asserted bound follows with $\delta=\exp(-\log^C P)$. Forming $s$ and multiplying the output by $D$ take $P\log^{O(1)}P$ bit operations under the stated assumptions.
\end{proof}

The input condition is imposed on $\lambda$ and $\lambda z$, not merely on $\alpha$: a dyadic $\alpha$ need not yield a dyadic $\lambda$. No dyadic representation of $D^{-1}z$ is needed, because the solver is called with $s=\lambda z$. For example, a single-vertex seed $z=\ee_v$ is admissible whenever the weights and $\lambda$ satisfy the stated assumptions.

In the clustering algorithms of \cite{TakaiMiyauchiIkedaYoshida2020}, the PageRank vector is used to order vertices by $p_v/d_v$ and examine the resulting sweep cuts. Single-vertex seeds give local clustering queries, and a collection of such queries is used for global clustering. Corollary~\ref{cor:pagerank} supplies a global, almost-linear-time approximation for each requested PageRank vector; it does not give a running time proportional to the output cluster size. Its guarantee is a vector error bound. Transferring the conductance guarantees for exact PageRank sweep cuts to approximate vectors requires a separate analysis of the sweep step, which we do not claim here.
 \appendix

\section{Details of the convex-flow black box\sectionpunct}
\label{app:chen-black-box}

This appendix records precisely how the general convex-flow theorem of
Chen, Kyng, Liu, Peng, Probst Gutenberg, and Sachdeva is used in
Section~\ref{sec:convex-flow}.  It has two roles.  First,
Subsection~\ref{app:chen-statement} records the theorem and the output
convention in the notation and epigraph orientation used here.  Second,
Subsection~\ref{app:derive-theorem-42} explains how the black-box theorem
stated in Section~\ref{sec:convex-flow} follows from this convention and the
sign conversion between incidence conventions.  The verification of
Assumption~8.2 of \cite{ChenKyngLiuPengProbstGutenbergSachdeva2025} for the lifted instance is carried out in
Proposition~\ref{prop:convex-flow-assumptions}, and the finite-precision
materialization of the final certificate is handled in
Lemma~\ref{lem:dual-certificate-repair}.  No new algorithmic ingredient is
introduced here.

\subsection{The convex-flow theorem and the output convention used here\sectionpunct}
\label{app:chen-statement}

The authors of \cite{ChenKyngLiuPengProbstGutenbergSachdeva2025} consider a directed graph \(G=(V,A)\), \(|A|=m\), with their
arc--vertex incidence matrix \(B_{\rm Ch}\in\mathbb{R}^{A\times V}\).  Their
convention is that \(B_{\rm Ch}\) has \(+1\) at the tail and \(-1\) at the head
of an arc, so the demand constraint is
\[
        B_{\rm Ch}^{\top} f = d .
\]
For convex arc functions
\(h_a:\mathbb{R}\to\mathbb{R}\cup\{+\infty\}\), write
\[
        h(f) := \sum_{a\in A} h_a(f_a),
        \qquad
        F^* := \inf\{h(f): B_{\rm Ch}^{\top}f=d\}.
\]
In the epigraph orientation used in this paper, the local domain for the
barrier is
\[
        \mathcal X_a
        := \operatorname{int}\operatorname{epi} h_a
        =
        \{(x,y): x\in\operatorname{int}(\operatorname{dom}h_a),\ y>h_a(x)\}.
\]
Thus an arc with finite domain \([0,\Lambda]\) is handled through the open
interval \(0<x<\Lambda\) in the barrier domain, while the closed capacity
restriction remains encoded by \(h_a(x)=+\infty\) outside \([0,\Lambda]\).
The coordinates \(y_a\) are local epigraph variables, not network-flow
variables: the only conserved variable in the framework of \cite{ChenKyngLiuPengProbstGutenbergSachdeva2025} is the
arc-flow vector \(f\).

We use the following output convention, obtained from the high-accuracy
implementation of \cite{ChenKyngLiuPengProbstGutenbergSachdeva2025} and specialized to our notation.  The source
guarantee we use from Theorem~8.13 of \cite{ChenKyngLiuPengProbstGutenbergSachdeva2025} is the existence of a
high-accuracy feasible flow.  The finite string and coordinate-query routine
below are an implementation convention for accessing the coordinates of that
flow to the precision needed by our later certificate construction; they are
not a separate claim that raw dyadic query answers satisfy conservation
exactly.

\begin{theorem}[JACM Theorem~8.13 of \cite{ChenKyngLiuPengProbstGutenbergSachdeva2025}, specialized output convention]
\label{thm:chen-general-convex-flow-app}
Let \(G=(V,A)\) have \(m\) arcs, let \(d\in\mathbb{R}^V\) be a demand vector,
and let \(h_a:\mathbb{R}\to\mathbb{R}\cup\{+\infty\}\) be convex arc
functions.  Assume that, for each arc \(a\), the open epigraph
\(\mathcal X_a\) is equipped with a \(\nu\)-self-concordant barrier
\(\psi_a:\mathcal X_a\to\mathbb{R}\), and that the following six conditions
hold for some \(K=\log^{O(1)}m\):
\begin{enumerate}
    \item Gradients and Hessians of \(\psi_a\) are available in
    \(\widetilde O(1)\) time per arc query, to the precision used by the
    algorithm.
    \item The finite arc domain is contained in a polynomial box,
    \(|x|\le m^K\) whenever \(h_a(x)<+\infty\); the demand satisfies
    \(\|d\|_\infty\le m^K\); and the finite values of \(h_a\) obey
    \(|h_a(x)|\le O(m^K+|x|^K)\).
    \item Each barrier has been shifted by an additive constant so that
    \[
        \inf\{\psi_a(x,y):(x,y)\in\mathcal X_a,\ |x|,|y|\le m^K\}=0 .
    \]
    \item There are strictly feasible variables \((f^{(0)},y^{(0)})\), with
    \(B_{\rm Ch}^{\top}f^{(0)}=d\) and
    \((f^{(0)}_a,y^{(0)}_a)\in\mathcal X_a\), such that
    \(|y^{(0)}_a|\le m^K\),
    \[
        m^{-K}I\preceq
        \nabla^2\psi_a(f^{(0)}_a,y^{(0)}_a)
        \preceq m^K I,
        \qquad
        \psi_a(f^{(0)}_a,y^{(0)}_a)\le K
    \]
    for every arc \(a\).
    \item The parameters \(\alpha,\varepsilon,\kappa\) appearing in the
    algorithm are chosen below \(1/(1000\nu)\).
    \item For every point in the polynomial box \(|x|,|y|\le m^K\) with
    \(\psi_a(x,y)\le\widetilde O(1)\),
    \[
        \nabla^2\psi_a(x,y)\preceq \exp(\log^{O(1)}m)I .
    \]
\end{enumerate}
Then, for every fixed constant \(C>0\), the high-accuracy implementation from
\cite{ChenKyngLiuPengProbstGutenbergSachdeva2025} can be used as follows.  It runs in \(m^{1+o(1)}\) time and, with
high probability, returns a finite string \(S\) together with a query routine
\(Q_S\).  There is an underlying real flow \(f\) satisfying
\(B_{\rm Ch}^{\top}f=d\) such that
\[
        h(f)
        \le
        \inf_{B_{\rm Ch}^{\top}f'=d} h(f')
        +
        \exp(-\log^C m).
\]
For every arc \(a\in A\) and every requested absolute precision \(\tau\) with
\(\log(1/\tau)=\log^{O(1)}m\), the query routine returns a dyadic number
\(Q_S(a,\tau)=\widetilde f_a\) satisfying
\[
        |\widetilde f_a-f_a|\le\tau .
\]
The work for \(O(m)\) such coordinate queries at these precisions is included
in the \(m^{1+o(1)}\) running time in the applications below.  The dyadic
coordinate values read from \(S\) are approximate answers and need not
themselves satisfy \(B_{\rm Ch}^{\top}f=d\) exactly.
\end{theorem}

\paragraph{Finite-precision output convention.}
The theorem above is the general arc-separable convex-flow, high-accuracy
part of the result in \cite{ChenKyngLiuPengProbstGutenbergSachdeva2025}, not the exact integral min-cost-flow part.  The
source theorem is used for the high-accuracy real feasible flow and objective
bound; the finite string and query routine are the output convention under
which we access that flow.  Thus we never use that theorem as a
source of an explicit finite dyadic flow certificate.  The only precision
regime used later is
\(\tau=\exp(-\log^{O(1)}m)\), and the theorem statement above records the
corresponding query guarantee and running-time convention.

After the incidence conversion in Subsection~\ref{app:derive-theorem-42}, the
constraint \(B_{\rm Ch}^{\top}f=d_{\rm Ch}\) is the statement
\(A^\uparrow f=b^\uparrow\) for the underlying real lifted flow.  It is not a
claim that the raw queried dyadic coordinates obey the demand constraints
exactly.  Whenever the final solver needs an explicit finite certificate, it
queries the required coordinates at precision \(\exp(-\log^{O(1)}P)\) and then
performs the deterministic materialization and repair of
Lemma~\ref{lem:dual-certificate-repair}; only the repaired vector is asserted
to satisfy the dual feasibility equations exactly over the rationals.

The strict inequality in the epigraph is only an algorithmic interior
condition.  For any fixed flow \(f\) with finite arc costs and any
\(\delta>0\), the choice \(y_a=h_a(f_a)+\delta\) is strictly feasible in the
epigraph and has objective \(\sum_a y_a=h(f)+m\delta\).  Conversely, every
strict epigraph point satisfies \(\sum_a y_a>h(f)\).  Hence the strict
epigraph formulation has the same infimum as the closed convex-flow objective,
provided boundary points of the finite arc domains can be approached from the
interior.  Lemma~\ref{lem:lift-interiorization} verifies this for the capped
lifted instance used in this paper.

\subsection{Deriving Theorem~\ref{thm:chen-convex-flow}\sectionpunct}
\label{app:derive-theorem-42}

Theorem~\ref{thm:chen-convex-flow} is the same high-accuracy black-box
convention as Theorem~\ref{thm:chen-general-convex-flow-app}, with the
notation specialized to the graph \(G=(V,A)\), demand \(d\), and arc objective
\[
        h(f)=\sum_{a\in A} h_a(f_a).
\]
The only domain convention to remember is that the barrier lives on
\(\operatorname{int}\operatorname{epi}h_a\), not on the closed boundary of
\(\operatorname{epi}h_a\).  Thus an arc with finite domain \([0,\Lambda]\)
uses \(0<x<\Lambda\) in its barrier domain.

For the lifted graph of Section~\ref{sec:lifted-flow}, our incidence matrix
\(A^\uparrow\) uses incoming minus outgoing, whereas the convention in \cite{ChenKyngLiuPengProbstGutenbergSachdeva2025} uses tail
minus head.  Therefore the black-box call uses
\[
        B_{\rm Ch}^{\top}=-A^\uparrow,
        \qquad
        d_{\rm Ch}=-b^\uparrow .
\]
With this choice,
\[
        B_{\rm Ch}^{\top}f=d_{\rm Ch}
        \quad\Longleftrightarrow\quad
        A^\uparrow f=b^\uparrow .
\]
Lemma~\ref{lem:chen-epigraph-compatibility} records the corresponding
open-epigraph specialization for the lifted cost functions: transport arcs have
open epigraph \(0<x<\Lambda,\ y>0\), quadratic arcs have open epigraph
\(0<x<\Lambda,\ y>x^2/(2w_e)\), the objective is \(\sum_a y_a\), and the
strict epigraph formulation has the same infimum as the closed separable-cost
problem.  In particular, the variables \(y_a\) are local epigraph variables,
not additional conserved flow variables.  This proves the black-box statement
used in Theorem~\ref{thm:chen-convex-flow}.
 \section*{Acknowledgments\sectionpunct}
The author thanks the anonymous reviewers for their helpful comments.
The author used ChatGPT (GPT-5.5 Pro) to assist in checking proofs, and OpenAI Codex to assist with revisions to the exposition, the derivation of the PageRank consequence, and manuscript preparation. The author assumes responsibility for all mathematical arguments, references, and other content.

\bibliographystyle{abbrv}
\bibliography{main}
\end{document}